\documentclass[acmsmall,screen]{acmart}
\newcommand{\figrule}{}

\usepackage{framed}
\usepackage{hyperref}
\usepackage{xspace}
\usepackage{fancyvrb}
\usepackage{listings}

\usepackage{booktabs}
\hypersetup{colorlinks=true,allcolors=blue}
\usepackage{amsmath,amssymb,latexsym,amstext}
\usepackage{array}
\usepackage{url}
\usepackage{comment}
\usepackage{listings}
\usepackage{showexpl} 

\usepackage{enumitem}
\usepackage{booktabs}

\usepackage{stmaryrd}
\usepackage{xcolor}
\usepackage{mathtools}
\usepackage{apptools}
\usepackage{chngcntr}
\usepackage{scalerel}
\usepackage{mathpartir}
\usepackage{graphicx}
\usepackage{caption}
\usepackage{csvsimple}
\usepackage{pdfpages}
\usepackage{tikz}
\usepackage{anyfontsize}
\usepackage{listings}
\usepackage{amssymb}

\def\withcolor{}

\ifdefined\withcolor
	 \definecolor{haskellblue}{rgb}{0.0, 0.0, 1.0}
	 \definecolor{haskellstr}{rgb}{0.2, 0.2, 0.6}
	 \definecolor{haskellred}{rgb}{1.0, 0.0, 0.0}
     \definecolor{gray_ulisses}{gray}{0.55}
     \definecolor{castanho_ulisses}{rgb}{0.71,0.33,0.14}
     \definecolor{preto_ulisses}{rgb}{0.41,0.20,0.04}
     \definecolor{green_ulises}{rgb}{0.2,0.75,0}
     \definecolor{haskelltypes}{rgb}{0.71,0.33,0.14}
     \definecolor{logiccolor}{rgb}{0.0, 0.0, 1.0}
\else
	\definecolor{haskellblue}{gray}{0.1}
	\definecolor{haskellstr}{gray}{0.1}
	\definecolor{haskellred}{gray}{0.1}
	\definecolor{gray_ulisses}{gray}{0.1}
	\definecolor{castanho_ulisses}{gray}{0.1}
	\definecolor{preto_ulisses}{gray}{0.1}
	\definecolor{green_ulisses}{gray}{0.1}
	\definecolor{haskelltypes}{gray}{0.1}
    \definecolor{logiccolor}{gray}{0,1}
\fi

\definecolor{subtleOpHighlight}{rgb}{0.4, 0.2, 0.0}

\definecolor{lcolor}{gray}{0.0}
\definecolor{lappcolor}{gray}{0.0}
\definecolor{lappascolor}{gray}{0.0}

\def\codesize{\small}
\newcommand\showfocus[1]{\color{purple}{\textbf{#1}}}
\newcommand\showlogic[1]{\color{logiccolor}{#1}}

\lstdefinelanguage{HaskellUlisses} {
	basicstyle=\ttfamily\codesize,
	moredelim=[is][\showfocus]{\#}{\#},
	moredelim=[is][\showlogic]{!}{!},
	sensitive=true,
	morecomment=[l][\ttfamily\itshape\codesize]{--},
	morecomment=[s][\ttfamily\itshape\codesize]{\{-}{-\}},
	morestring=[b]",
	stringstyle=\color{haskellstr},
	showstringspaces=false,
	numberstyle=\codesize,
	numberblanklines=true,
	showspaces=false,
	breaklines=true,
	showtabs=false,
  literate={
	         {<!}{{{\color{lcolor}<!}}}2
           {`}{{{$^{\backprime}{}$}}}1
           {?}{{{\color{lcolor}?}}}1
           {<=}{{$\leq$}}1
           {theta}{{$\theta$}}1
           {gf}{{{\color{lappascolor}f}}}1
           {rmap}{{{\color{lappcolor}map}}}3
           {gmap}{{{\color{lappascolor}map}}}3
           {r.}{{{\color{lappcolor}.}}}1
           {g.}{{{\color{lappascolor}.}}}1
           {r++}{{{\showfocus{++}}}}2
           {g++}{{{\color{lappascolor}++}}}2
           {gfib}{{{\color{lappascolor}fib}}}3
           {rfib}{{{\color{lappcolor}fib}}}3
           {r++}{{{\color{lappcolor}++}}}2
           {op}{{$\odot$}}1
           {env}{{$\Gamma$}}1
           {|-}{{$\vdash$}}1
           {<=!}{{{\color{lcolor}<=!}}}3
           {!=}{{$\neq$}}1
           {->}{{$\rightarrow$}}2
           {dollar}{{$\texttt{\$}$}}1
           {Set_mem}{{$\in$}}1
           {Set_cup}{{$\cup$}}1
           {Set_cap}{{$\cap$}}1
           {Set_emp}{{$\emptyset$}}1
           {Set_sub}{{$\subseteq$}}1
           {<=>}{{$\Leftrightarrow$}}3
           {=>}{{$\Rightarrow$}}2
           {1->}{{$\rightarrow$}}1
           {1=>}{{$\Rightarrow$}}1
           {||-}{{$\vdash$}}1
           {|->}{{$\mapsto$}}2
           {<:}{{$\preceq$}}1
           {Inarritu}{Inarritu}8},
	emph=
	{[1]
		FilePath,IOError,abs,acos,acosh,all,and,any,appendFile,approxRational,asTypeOf,asin,
		asinh,atan,atan2,atanh,basicIORun,break,catch,ceiling,chr,compare,concat,concatMap,
		cos,cosh,curry,cycle,decodeFloat,denominator,digitToInt,div,divMod,drop,
		dropWhile,either,elem,encodeFloat,enumFrom,enumFromThen,enumFromThenTo,enumFromTo,
		error,even,exponent,fail,mapMaybe,filter,flip,floatDigits,floatRadix,floatRange,floor,
		foldl,foldl1,foldr1,fromDouble,fromEnum,fromInt,fromInteger,fromIntegral,
		fromRational,fst,gcd,getChar,getContents,getLine,head,inRange,index,init,intToDigit,
		interact,ioError,isAlpha,isAlphaNum,isAscii,isControl,isDenormalized,isDigit,isHexDigit,
		isIEEE,isInfinite,isLower,isNaN,isNegativeZero,isOctDigit,isPrint,isSpace,isUpper,iterate,
		last,lcm,length,lex,lexDigits,lexLitChar,lines,log,logBase,lookup,mapM,mapM_,max,
		maxBound,posMax,negMax,maximum,maybe,min,minBound,minimum,mod,negate,notElem,null,numerator,odd,
		or,ord,pi,pred,primExitWith,print,product,putChar,putStr,putStrLn,quot,
		quotRem,range,rangeSize,read,readDec,readFile,readFloat,readHex,readIO,readInt,readList,readLitChar,
		readLn,readOct,readParen,readSigned,reads,readsPrec,realToFrac,recip,rem,repeat,replicate,return,
		reverse,round,scaleFloat,scanl,scanl1,scanr,scanr1,seq,sequence,sequence_,show,showChar,showInt,
		showList,showLitChar,showParen,showSigned,showString,shows,showsPrec,significand,signum,sin,
		sinh,snd,span,splitAt,sqrt,subtract,succ,tail,take,takeWhile,tan,tanh,threadToIOResult,toEnum,
		toInteger,toLower,toRational,toUpper,truncate,uncurry,undefined,unlines,until,unwords,unzip,
		unzip3,userError,words,writeFile,zip,zip3,zipWith,zipWith3,listArray,doParse,empty,for,initTo,
        assert,compose,checkGE,maxEvens,empty,create,get,set,initialize,idVec,fastFib,fibMemo,
        ex1,ex2,ex3,incr,inc,dec,isPos,positives,find,flatten, expand,exAll,
				ind,evenLen,lenAppend,exDistOr,allDistAnd,len,size,union,fromList,initUpto,trim,
        insertSort,decsort,qsort,reverse,append,upperCase, ifM, whileM, get, decrM, diff,
        project, select, elts, keys, dkeys, dfun, addKey, pTrue, emptyRD, rFalse,
                dom, rng, isI, isD, isS, movie1, movie2,  toI, toS, toD, good_titles, runState, ret,
                update, getCtr, setCtr, ctr, rdCtr, wrCtr, ifTest, whileTest, posCtr, zeroCtr, decr, decCtr,
                pread , pwrite , plookup , pcontents, pcreateF , pcreateFP, pcreateD, active, caps, pset, eqP,
                write, contents, alloc, derivP, copyP, createDir, store, copyRec, copySpec,
                forM_, when, flookup, fread, createDir, pcreateFile, isFile, copyFrame, ?
	},
	emphstyle={[1]\color{haskellblue}},
	emph=
	{[2]Show,Eq,Iso,VerifiedOrd,Ord,Num,UpClosed,Comp,Wit,Witness,Inductive,Meet,Flip,TRUE,
	    Peano,Nat,Pos,Neg,IntGE,Plus,List,PAnd, POr, POrL, POrR,
        Bool,Char,Double,Either,Float,IO,Integer,Int,Maybe,Up,Mono,Identity,
        Ordering,Rational,Ratio,ReadS,ShowS,String,Word8,
        InPacket,Tree,Vec,NullTerm,IncrList,DecrList,
        UniqList,BST,MinHeap,MaxHeap,World,RIO,IO,HIO,Post,Pre, OptEq,
        Privilege, Prop, Chain, ChainTy, Range, Dict, RD, Dom, Set, P, Univ, Schema, MovieSchema, RT,
        TDom, TRange, MoviesTable, RTSubEqFlds, RTEqFlds, Disjoint, Union, Ret, Seq, Trans, Map,
        Pure, Then, Else, Exit, Inv, OneState, Priv, Path, FH, Stable,
	    Nat, Monoid, VerifiedMonoid, VerifiedComMonoid, Plus_2_2_eq_4, Plus_2_2, Nat_up, Int_up,
	    AppendNilId, AppendAssoc,MapFusion,
	    Plus_comm, Par, Term,
	    Formula, Assignment, Body, Accel, Real, Accel', RVar, VerifiedCommutativeMonoid, CommutativeMonoid
	},
	emphstyle={[2]\color{castanho_ulisses}},
	emph=
	{[3]
		case,class,newtype,data,deriving,do,else,if,import,in,infixl,infixr,instance,let,
		module,of,primitive,then,refinement,type,where,forall,bound,
		measure,reflect,predicate, instance, class
	},
	emphstyle={[3]\color{preto_ulisses}\textbf},
	emph=
	{[4]
		quot,rem,div,mod,elem,notElem,seq
	},
	emphstyle={[4]\color{castanho_ulisses}\textbf},
	emph=
	{[5]
		EQ,GT,LT,Left,Right
	},
	emphstyle={[5]\color{preto_ulisses}\textbf},
	emph=
	{[6]
	    axiomatize, measure, inline
	},
	emphstyle={[6]\color{lcolor}}
}

\lstnewenvironment{code}
{\lstset{language=HaskellUlisses}}
{}

\lstnewenvironment{mcode}
{\lstset{language=HaskellUlisses,columns=fullflexible,keepspaces,mathescape}}
{}

\lstnewenvironment{mcodef}
{\lstset{language=HaskellUlisses,columns=fullflexible,keepspaces,mathescape,frame=single}}
{}

\lstMakeShortInline[language=HaskellUlisses,mathescape,keepspaces,mathescape,basicstyle=\ttfamily\codesize,breakatwhitespace]"

\lstdefinelanguage{Pseudo} {
	basicstyle=\ttfamily\codesize,
	sensitive=true,
  mathescape=true,
	morecomment=[l][\color{gray_ulisses}\ttfamily\codesize]{--},
	morecomment=[s][\color{gray_ulisses}\ttfamily\codesize]{\{-}{-\}},
	morestring=[b]",
	showstringspaces=false,
	numberstyle=\codesize,
	numberblanklines=true,
	showspaces=false,
	breaklines=true,
	showtabs=false
}

\lstdefinelanguage{java} {
    keywordstyle=[1],
    keywordstyle=[2]\color{ForestGreen},
    keywordstyle=[3]\color{Bittersweet},
    keywordstyle=[4]\color{RoyalPurple},
    morekeywords={region,private,synchronized}
}

\ifdefined\withcolor
        \definecolor{typecol}{rgb}{0.0,0.5,0.0}
        \definecolor{funcol}{rgb}{0.0,0.1,0.9}
\else
        \definecolor{typecol}{gray}{0.0}
        \definecolor{funcol}{gray}{0.0}
\fi

\lstdefinelanguage{pseudo2}{
  language=Python,
	basicstyle=\ttfamily\normalsize,
  mathescape=true,
  morekeywords={type,def,do,let},
  emph={[1] \Expr,\Pred, HP, FP, Int},
  emphstyle={[1]\itshape\color{typecol}},
  emph={[2] \pbe,normalize},
  emphstyle={[2]\itshape\color{funcol}},
  emph={[3] repeat,until},
  emphstyle={[3]\textbf}
}

\lstnewenvironment{code2}{\lstset{language=pseudo2}}{}

\usepackage{multicol,lipsum}
\usepackage{multirow}
\usepackage{array}

\usepackage{wrapfig}

\usepackage{epsfig}
\usepackage[outdir=./]{epstopdf}

\usepackage{commands}
\usepackage[inference]{semantic}

\usepackage{booktabs}   
\usepackage{subcaption} 

\begin{document}

\setcopyright{acmcopyright}
\acmPrice{}
\acmDOI{10.1145/3158141}
\acmYear{2018}
\copyrightyear{2018}
\acmJournal{PACMPL}
\acmVolume{2}
\acmNumber{POPL}
\acmArticle{53}
\acmMonth{1}


\title[]{Refinement Reflection: Complete Verification with SMT}

\author{Niki Vazou}

\affiliation{
  \institution{University of Maryland}
  \country{USA}
}
\email{nvazou@cs.umd.edu}

\author{Anish Tondwalkar}
\affiliation{
	\institution{University of California, San Diego}
	\country{USA}
}
\email{atondwal@cs.ucsd.edu}

\author{Vikraman Choudhury}
\affiliation{
	\institution{Indiana University}
	\country{USA}
}

\author{Ryan G. Scott}
\affiliation{
	\institution{Indiana University}
	\country{USA}
}

\author{Ryan R. Newton}
\affiliation{
	\institution{Indiana University}
	\country{USA}
}

\author{Philip Wadler}
\affiliation{
	\institution{University of Edinburgh and Input Output HK}
	\country{UK}
}
\email{wadler@inf.ed.ac.uk}

\author{Ranjit Jhala}
\affiliation{
	\institution{University of California, San Diego}
	\country{USA}
}
\email{jhala@cs.ucsd.edu}

\renewcommand{\shortauthors}{N. Vazou, A. Tondwalkar, V. Choudhury, R. G. Scott, R. R. Newton, P. Wadler, and R. Jhala}


\begin{abstract}
We introduce \emph{Refinement Reflection},
a new framework for building SMT-based
deductive verifiers.
The key idea is to reflect the code
implementing a user-defined function
into the function's (output) refinement
type.
As a consequence, at \emph{uses} of
the function, the function definition
is instantiated in the SMT logic in a precise fashion
that permits decidable verification.
Reflection allows the user to write 
\emph{equational proofs} of programs 
just by writing other programs 
\eg using pattern-matching
and recursion to perform case-splitting
and induction.
Thus, via the propositions-as-types principle,
we show that reflection permits the
\emph{specification} of arbitrary
functional correctness properties.
Finally, we introduce a proof-search algorithm called
\emph{Proof by Logical Evaluation}
that uses techniques from model
checking and abstract interpretation,
to completely automate equational
reasoning.
We have implemented reflection in
\toolname and used it to verify that
the widely used instances of the Monoid,
Applicative, Functor, and Monad typeclasses
actually satisfy key algebraic laws
required to make the clients safe, 
and have used reflection to build the first 
library that actually verifies assumptions 
about associativity and ordering that are 
crucial for safe deterministic parallelism.
\end{abstract}

\maketitle

\section{Introduction}\label{sec:intro}

Deductive verifiers fall roughly
into two camps.
Satisfiability Modulo Theory (SMT)
based verifiers (\eg \dafny and \fstar)
use fast decision procedures to
automate the verification of programs that
only require reasoning over a fixed
set of theories like linear arithmetic,
string, set and bitvector operations.
These verifiers, however, encode the
semantics of user-defined functions
with universally-quantified axioms
and use incomplete (albeit effective)
heuristics to instantiate those axioms.
These heuristics make it difficult to
characterize the kinds of proofs that
can be automated, and hence,
explain why a given proof attempt
fails~\cite{Leino16}.
At the other end, we have
Type-Theory (TT) based theorem provers
(\eg \coq and \agda)
that use type-level computation
(normalization) to facilitate
principled reasoning about
terminating user-defined functions,
but which
require the user to supply
lemmas or rewrite hints
to discharge proofs
over decidable theories.


We introduce \emph{Refinement Reflection},
a new framework for building SMT-based
deductive verifiers, which permits the
specification of arbitrary properties
and yet enables complete, automated
SMT-based reasoning about user-defined
functions.
In previous work, refinement types
\citep{ConstableS87,Rushby98} --- which decorate
basic types (\eg "Integer") with SMT-decidable
predicates (\eg "{v:Integer | 0 <= v && v < 100}") --- were
used to retrofit so-called shallow
verification, such as array bounds checking, into
several languages:
ML~\cite{pfenningxi98,GordonRefinement09,LiquidPLDI08},
C~\cite{deputy,LiquidPOPL10},
Haskell~\citep{Vazou14},
TypeScript~\cite{Vekris16}, and
Racket~\cite{RefinedRacket}.
In this work, 
we extend refinement types with refinement reflection,
leading to the following three contributions.

\mypara{1. Refinement Reflection}
Our first contribution is the notion of
\emph{refinement reflection}.
To reason about user-defined functions,
the function's implementation is
\emph{reflected} into its (output)
refinement-type specification,
thus converting the function's
type signature into a precise
description of the function's behavior.
This simple idea has a profound consequence:
at \emph{uses} of the function, the standard
rule for (dependent) function application
yields a precise means of reasoning about
the function~(\S~\ref{sec:types-reflection}).

\mypara{2. Complete Specification}
%
%
Our second contribution is a
\emph{library of combinators}
that lets programmers compose
sophisticated \emph{proofs}
from basic refinements and
function definitions.
Our proof combinators let programmers
use existing language mechanisms, like
branches (to encode case splits),
recursion (to encode induction), and
functions (to encode auxiliary lemmas),
to write proofs that look very much like
their pencil-and-paper
analogues~(\S~\ref{sec:overview}).
Furthermore, since proofs are literally
just programs, we use the principle of
propositions-as-types~\cite{Wadler15} 
(known as the Curry-Howard isomorphism~\cite{CHIso})
to show
that SMT-based verifiers can
express any natural deduction proof, 
thus providing a pleasant
implementation of
natural deduction
that can be used for
pedagogical purposes~(\S~\ref{sec:higher-order}).

\mypara{3. Complete Verification}
While equational proofs can be very \elegant
and expressive, writing them out can
quickly get exhausting.
Our third contribution is
\emph{Proof by Logical Evaluation} (\pbesym)
a new proof-search algorithm
that automates
equational reasoning.
The key idea in \pbesym is to mimic
type-level computation within SMT-logics
by representing functions in a
\emph{guarded form} \cite{Dijkstra1975}
and repeatedly unfolding function application
terms by instantiating them with their definition
corresponding to an \emph{enabled} guard.
We formalize a notion of equational proof
and show that the above strategy is
\emph{complete}: \ie it is guaranteed
to find an equational proof if one exists.
Furthermore, using techniques from
the literature on Abstract
Interpretation~\cite{CousotCousot77}
and Model Checking \cite{CGL92},
we show that the above proof search
corresponds to a \emph{universal}
(or \emph{must}) abstraction of
the concrete semantics of the
user-defined functions.
Thus, since those functions are total,
we obtain the pleasing
guarantee that proof search
terminates~(\S~\ref{sec:pbe}).

\smallskip
%
We evaluate our approach by implementing
refinement reflection and \pbesym in
\toolname~\citep{Vazou14}, thereby
turning Haskell into a theorem prover.
Repurposing an existing programming language allows us
to take advantage of a mature compiler and an ecosystem of
libraries, while keeping proofs and programs in the same
language.
We demonstrate the benefits of this
conversion by proving typeclass laws.
Haskell's typeclass machinery has led
to a suite of expressive abstractions
and optimizations which, for correctness,
crucially require typeclass \emph{instances}
to obey key algebraic laws.
We show how reflection and \pbesym
can be used to verify that widely
used instances of the Monoid,
Applicative, Functor, and Monad
typeclasses satisfy the respective
laws.
Finally, we use reflection to
create the first deterministic
parallelism library that actually
verifies assumptions about
associativity and ordering
that ensure determinism (\S~\ref{sec:evaluation}).

Thus, our results demonstrate
that Refinement Reflection and
Proof by Logical Evaluation
identify a new design for
deductive verifiers which,
by combining the complementary
strengths of SMT- and TT- based
approaches, enables complete
verification of expressive
specifications spanning
decidable theories and
user defined functions.

\section{Overview}
\label{sec:overview}
\label{sec:examples}


We start with an overview of how SMT-based refinement
reflection lets us write proofs as plain
functions and how \pbesym automates equational reasoning.

\subsection{Refinement Types}\label{subsection:overview:refinements}

First, we recall some preliminaries about specification
and verification with refinement types.

\mypara{Refinement types} are the source program's (here
Haskell's) types refined with logical predicates drawn
from an SMT-decidable logic~\citep{ConstableS87,Rushby98}.
For example, we define "Nat" as the set of "Integer" values
"v" that satisfy the predicate "0 <= v" from
the quantifier-free logic of linear arithmetic and
uninterpreted functions (QF-UFLIA~\cite{SMTLIB2}):
\begin{mcode}
    type Nat = { v:Integer | 0 <= v }
\end{mcode}

\mypara{Specification \& Verification}
%
Throughout this section,
to demonstrate the proof features we add to \toolname,
we will use the textbook Fibonacci function which we type as follows.
\begin{mcode}
    fib :: Nat -> Nat
    fib 0 = 0
    fib 1 = 1
    fib n = fib (n-1) + fib (n-2)
\end{mcode}
To ensure termination, the input type's
refinement specifies a \emph{pre-condition}
that the parameter must be "Nat".
The output type's refinement specifies
a \emph{post-condition} that the result
is also a "Nat".
Refinement type checking automatically
verifies that if "fib" is invoked with
a non-negative "Integer", then it terminates
and yields a non-negative "Integer".

\mypara{Propositions}
We define a data type representing propositions
as an alias for unit:
\begin{mcode}
    type $\typp$ = ()
\end{mcode}
\VC{Prop is also in HOL, and it's encoded
  as Bool, so this might be confusing,
  maybe this should be Proof instead?}
which can be \emph{refined} with propositions about
the code, \eg that $2 + 2$ equals $4$
\begin{mcode}
    type Plus_2_2 = { v: $\typp$ | 2 + 2 = 4 }
\end{mcode}
For simplicity, in \toolname, we abbreviate the above
to "type Plus_2_2 = { 2 + 2 = 4 }".

\mypara{Universal \& Existential Propositions}
Using the standard encoding of~\citet{CHIso},
known as the Curry-Howard isomorphism,
refinements encode universally-quantified propositions as
\emph{dependent function} types of the form:
\begin{mcode}
    type Plus_comm = x:Integer -> y:Integer -> { x + y = y + x }
\end{mcode}
As "x" and "y" refer to arbitrary inputs,
any inhabitant of the above type is a proof
that "Integer" addition commutes.

Refinements encode existential quantification via
\emph{dependent pairs} of the form:
\begin{mcode}
    type Int_up = n:Integer -> (m::Integer, {n < m})
\end{mcode}
The notation "(m :: t, t')" describes dependent pairs
where the name "m" of the first element can appear
inside refinements of the second element.
Thus, "Int_up" states the proposition that for
every integer "n", \emph{there exists}
one that is larger than "n".

While quantifiers cannot appear directly
inside the refinements, dependent functions and
pairs allow us to specify quantified propositions.
One limitation of this encoding is that quantifiers
cannot exist inside refinement's logical connectives (like $\land$ and $\lor$).
In \S~\ref{sec:higher-order}, we describe how to
encode logical connectives using data types,
\eg conjunction as a product and disjunction as
a union, how to specify arbitrary,
quantified propositions
using refinement types, \ie have complete specifications,
and how to verify those propositions using refinement
type checking.

\mypara{Proofs}
We \emph{prove} the above propositions by
writing Haskell programs, for example
\begin{mcode}
    plus_2_2 :: Plus_2_2        plus_comm :: Plus_comm        int_up :: Int_up
    plus_2_2 = ()               plus_comm = \x y -> ()        int_up = \n -> (n+1,())
\end{mcode}
Standard refinement typing reduces the above to
the respective \emph{verification conditions} (VCs)
$$ \ttrue \Rightarrow 2 + 2 = 4 \quad \quad \quad
   \forall \ x,\ y\ .\ \ttrue \Rightarrow x + y = y + x \quad \quad \quad
   \forall \ n\ . n < n+1 $$
%
%
which are easily deemed valid by the SMT solver, allowing us
to prove the respective propositions.

\mypara{Soundness and Lazy Evaluation}
Readers familiar with Haskell's lazy 
semantics may be concerned that ``bottom'', 
which inhabits all types, makes our proofs suspect.
Fortunately, as described in \citet{Vazou14}, \toolname, 
by default, checks that user defined functions 
provably terminate and are total (\ie return non-bottom values, 
and do not throw any exceptions),
which makes our proofs sound.
\toolname checks that each function is terminating using 
a \textit{termination metric} 
\ie a natural number that decreases at each recursive call. 
For instance, if we generalize the signature 
of "fib" to "Integers", as shown below, then 
\toolname reports a termination error: 
%
\begin{mcode}
    fib :: n:Integer -> Integer
\end{mcode}
\toolname generates a type error in the definition of "fib" since in 
the recursive calls "fib (n-1)" and "fib (n-2)"
the arguments "n-1" and "n-2"
cannot be proved to be non-negative and less than "n".
Both these proof obligations are satisfied when 
the domain of "fib" is restricted to natural numbers. 
\begin{mcode}
    fib :: n:Nat -> Nat / [n]
\end{mcode}
The above type signature is explicitly annotated 
with the user specified termination metric "/ [n]"
declaring that "n" is the decreasing natural number. 
\toolname heuristically assumes that the termination metric 
is always the first argument of the function (that can be mapped to natural numbers),
thus, the above explicit termination metric can be omitted. 

Not all Haskell functions terminate. 
The "lazy" annotation 
deactivates termination 
checking, \eg for the 
the "diverge" function 
shown below.
Haskell terms marked as "lazy" 
could \emph{unsoundly} be used 
as proof terms, much like \coq's 
unsoundness with respect to "Admitted". 
For example, the following is 
accepted by \toolname 
\begin{mcode}
    lazy diverge 
    diverge :: x:Integer -> { x = 0 }
    diverge x = diverge x 
\end{mcode}

\mypara{Totality Checking}
\toolname further checks that all user-specified functions are totally defined. 
For instance the below definition 
\begin{mcode}
    fibPartial 0 = 0 
    fibPartial 1 = 1 
\end{mcode}
generates a totality error. 
The above definition is completed, by GHC, 
with an "error" invocation
\begin{mcode}
    fibPartial 0 = 0 
    fibPartial 1 = 1 
    fibPartial _ = error "undefined" 
\end{mcode}
Thus, to check totality, \toolname simply ascribes the precondition "false" to "error"
\begin{mcode}
    error :: { v:String | false } -> a
\end{mcode}
Thus, to typecheck "fibPartial", \toolname needs to prove that the call to "error" 
is dead-code, \ie happens under an inconsistent environment. As this check fails, 
\toolname generates a totality error. 
As with termination checking,
\toolname provides an \textit{unsound}
error function "unsoundError" with no false precondition.

In the sequel (and in our evaluation) 
we assume that proof terms are generated without 
any unsound uses of "lazy" or "unsafeError".
However, note that lazy, diverging functions 
can be soundly verified, and diverging code 
can soundly co-exist with terminating proof 
terms: we refer the reader to \citet{Vazou14} for details.

\subsection{Refinement Reflection}\label{subsec:overview:reflection}

Suppose we wish to prove properties about the "fib"
function, \eg that "{fib 2 = 1}".
Standard refinement type checking runs into two problems.
First, for decidability and soundness, arbitrary
user-defined functions cannot belong in the refinement
logic, \ie we cannot {refer} to "fib" in a refinement.
Second, the only specification that a refinement type checker
has about "fib" is its type "Nat -> Nat" which is
too weak to verify "{fib 2 = 1}".
To address both problems, we "reflect fib" into the
logic which sets the three steps of refinement
reflection in motion.

\mypara{Step 1: Definition}
The annotation creates an \emph{uninterpreted function}
"!fib! :: Integer -> Integer" in the refinement logic.
By uninterpreted, we mean that the logical "!fib!"
is not connected to the program function
"fib"; in the logic, "!fib!" only satisfies the
\emph{congruence axiom}
$\forall n, m.\ n = m\ \Rightarrow\ \fib{n} = \fib{m}$.
On its own, the uninterpreted function is
not terribly useful: we cannot check "{!fib! 2 = 1}"
as the SMT solver cannot prove the
VC 
$\ttrue \Rightarrow \fib{2} = 1$ 
which requires reasoning about "fib"'s
definition.

\mypara{Step 2: Reflection}
In the next key step, we reflect the \emph{definition}
of "fib" into its refinement type by automatically
strengthening the user defined type for "fib" to:
\begin{mcode}
    fib :: n:Nat -> { v:Nat | v = !fib! n && fibP n }
\end{mcode}
where "fibP" is an alias for a refinement
\emph{automatically derived} from the
function's definition:
\begin{mcode}
    fibP n $\doteq$ n == 0 $\Rightarrow$ !fib! n = 0
           $\wedge$ n == 1 $\Rightarrow$ !fib! n = 1
           $\wedge$ n  > 1 $\Rightarrow$ !fib n! = !fib! (n-1) + !fib! (n-2)
\end{mcode}

\mypara{Step 3: Application}
With the reflected refinement
type, each application of "fib"
in the code automatically
\emph{unfolds} the definition
of "fib" \textit{once} in the logic.
We prove "{!fib! 2 = 1}" by:
\begin{mcode}
    pf_fib2 :: { !fib! 2 = 1 }
    pf_fib2 = let { t0 = #fib# 0; t1 = #fib# 1; t2 = #fib# 2 } in ()
\end{mcode}
We write in bold red, "#f#", to highlight places where the
unfolding of "f"'s definition is important.
Via refinement typing, the above yields the
following VC that is discharged by SMT, even
though "!fib!" is uninterpreted:
$$(\fibdef\ 0 \ \wedge\ \fibdef\ 1 \ \wedge\ \fibdef\ 2) \  \Rightarrow\ (\fib{2} = 1)$$
%
%
The verification of "pf_fib2" relies
merely on the fact that "fib" is applied
to (\ie unfolded at) "0", "1", and "2".
The SMT solver automatically \emph{combines}
the facts, once they are in the antecedent.
Thus, the following is also verified:
\begin{mcode}
    pf_fib2' :: {v:[Nat] | !fib! 2 = 1 }
    pf_fib2' = [ #fib# 0, #fib# 1, #fib# 2 ]
\end{mcode}

In the next subsection, we will continue to use explicit,
step-by-step proofs as above, but we will introduce tools 
for proof composition.  
Then, in \S~\ref{sec:overview:pbe} we will show how to 
eliminate unnecessary details from such proofs, using 
\emph{Proof by Logical Evaluation} (\pbesym).


\subsection{Equational Proofs}\label{subsection:overview:equational}

We can {structure} proofs to follow the
style of \emph{calculational} or \emph{equational}
reasoning popularized in classic texts~\cite{Dijkstra76,Bird89}
and implemented in \agda~\citep{agdaequational}
and \dafny~\citep{LeinoPolikarpova16}.
To this end, we have developed a library
of proof combinators that permits reasoning
about equalities and linear arithmetic.

\mypara{``Equation'' Combinators}
We equip \toolname with a family of
equation combinators, "op", for
logical operators in the theory QF-UFLIA,
"op" $\in \{ =, \not =, \leq, <, \geq, > \}$.
In Haskell code, to avoid collisions with existing operators,
we further append a colon ``":"'' to these operators, so that ``$=$''
becomes the Haskell operator "(=:)".
The refinement type of "op"  \emph{requires}
that $x \odot y$ holds and then \emph{ensures}
that the returned value is equal to "x".
For example, we define "(=:)" as:
\begin{mcode}
    (=:) :: x:a -> y:{ a | x = y } -> { v:a | v = x }
    x =: _ = x
\end{mcode}
and use it to write the following equational proof:
\begin{mcode}
    fib2_1 :: { !fib! 2 = 1 }
    fib2_1 =  #fib# 2 =: #fib# 1 + #fib# 0 =: 1 ** QED
\end{mcode} 
where "** QED" constructs proof terms by
casting expressions to \typp in a post-fix
fashion.
\begin{mcode}
    data $\tqed$ = QED            (**) :: a -> $\tqed$ -> $\typp$
                              _ ** QED = ()
\end{mcode}

\mypara{Proof Arguments}
Often, we need to compose lemmas into larger
theorems. For example, to prove "!fib! 3 = 2" we
may wish to reuse "fib2_1" as a lemma.
We do so by defining a variant of "(=:)", written as "(=?)",
that takes an explicit proof argument: 
\begin{mcode}
    (=?) :: x:a -> y:a -> { $\typp$ | x = y } -> { v:a | v = x }
    x =? _ _ = x
\end{mcode}
We use the "(=?)" combinator to prove that "!fib! 3 = 2".
\begin{mcode}
    fib3_2 :: { !fib! 3 = 2 }
    fib3_2 = (#fib# 3 =: fib 2 + #fib# 1 =? 2 dollar fib2_1) ** QED
\end{mcode}
Here "fib 2" is not important to unfold, because "fib2_1"
already provides the same information.

\mypara{``Because'' Combinators}
Observe that the proof term "fib3_2" needs parentheses, 
since Haskell's "(dollar)" operator has the lowest (\ie 0) precedence. 
To omit parentheses in proof terms, we define
a ``because'' combinator that operates exactly like Haskell's 
"(dollar)" operator, but has the same precedence as the proof combinators 
"(=:)" and "(=?)".
\begin{mcode}
    ($\because$) :: ($\typp$ -> a) -> $\typp$ -> a
    f $\because$ y = f y
\end{mcode}
We use the ``because'' combinator to remove the parentheses from the proof term 
"fib3_2".
\begin{mcode}
    fib3_2 :: { !fib! 3 = 2 }
    fib3_2 = #fib# 3 =: fib 2 + #fib# 1 =? 2 $\because$ fib2_1 ** QED
\end{mcode}

\mypara{Optional Proof Arguments}
Finally, we unify both combinators "(=:)" and "(=?)"
using type classes to define the class method "(=.)" 
that takes an \textit{optional} proof argument, and 
generalize such definitions for each operator in "op". 
We define the class "OptEq" to have one method "(=.)" 
that takes two arguments of type "a", to be compared 
for equality, and returns a value of type "r". 
\begin{mcode}
    class OptEq a r where
      (=.) :: a -> a -> r
\end{mcode}
When instantiating the result type to be the same as the argument type "a", 
"(=.)" behaves as "(=:)".
\begin{mcode}
    instance (a~b) => OptEq a b where
      (=.) :: x:a -> { y:a | x = y } -> { v:b | v = x }
\end{mcode}
When the result type is instantiated to be 
a function that given a proof term "$\typp$"
returns the argument type "a", 
"(=.)" behaves exactly as "(=?)".
\begin{mcode}
    instance (a~b) => OptEq a ($\typp$ -> b) where
      (=.) :: x:a -> y:a -> { $\typp$ | x = y } -> { v:a | v = x }
\end{mcode}

Thus, "(=.)" takes two arguments to be compared for equality
and, optionally, a proof term argument. 
With this, the proof term "fib3_2" is simplified to  
\begin{mcode}
    fib3_2 :: { !fib! 3 = 2 }
    fib3_2 = #fib# 3 =. fib 2 + #fib# 1 =. 2 $\because$ fib2_1 ** QED
\end{mcode}

\mypara{Arithmetic and Ordering}
Next, lets see how we can use arithmetic and ordering to
prove that "fib" is (locally) increasing, \ie for all $n$,
$\fib{n} \leq \fib{(n+1)}$.
\begin{mcode}
    type Up f = n:Nat -> { f n <= f (n + 1) }

    fibUp   :: Up fib
    fibUp 0 =  #fib# 0 $<$. #fib# 1                          ** QED
    fibUp 1 =  #fib# 1 <=. fib 1 + #fib# 0     =. #fib# 2     ** QED
    fibUp n =  #fib# n <=. fib n + #fib# (n-1) =. #fib# (n+1) ** QED
\end{mcode} 


\mypara{Case Splitting}
The proof "fibUp" works by splitting cases on the value of "n".
In the cases "0" and "1", we simply assert
the relevant inequalities. These are verified as the
reflected refinement unfolds the definition of
"fib" at those inputs.
The derived VCs are (automatically) proved
as the SMT solver concludes $0 < 1$ and $1 + 0 \leq 1$
respectively.
When "n" is greater than one, "fib n" is unfolded
to  "fib (n-1) + fib (n-2)", which, as "fib (n-2)"
is non-negative, completes the proof.

\mypara{Induction \& Higher Order Reasoning}
Refinement reflection smoothly accommodates induction
and higher-order reasoning.
For example, let's prove that every function "f"
that increases locally (\ie "f z <= f (z+1)" for all "z")
also increases globally (\ie "f x <= f y" for all "x < y")
%
\begin{mcode}
    type Mono = f:(Nat -> Integer) -> Up f -> x:Nat -> y:{ x < y } -> { f x <= f y }

    fMono :: Mono / [y]
    fMono f up x y
      | x+1 == y = f x <=. f (x+1) $\because$ up x <=. f y                          ** QED
      | x+1  < y = f x <=. f (y-1) $\because$ fMono f up x (y-1) <=. f y $\because$ up (y-1) ** QED
\end{mcode}
We prove the theorem by induction
on "y" as specified by the
annotation "/ [y]" which states
that "y" is a well-founded
termination metric that decreases
at each recursive call~\citep{Vazou14}.
If "x+1 == y", then we call the "up x" proof argument.
Otherwise, "x+1 < y", and we use
the induction hypothesis \ie apply
"fMono" at "y-1", after which
transitivity of the less-than
ordering finishes the proof.
We can \emph{apply} the general
"fMono" theorem to prove that
"fib" increases monotonically:
\begin{mcode}
    fibMono :: n:Nat -> m:{ n < m } -> { !fib! n <= !fib! m }
    fibMono = fMono fib fibUp
\end{mcode}

\subsection{Complete Verification: Automating Equational Reasoning}
\label{sec:overview:pbe}

While equational proofs can be
very \elegant, writing them out
can quickly get exhausting.
Lets face it: "fib3_2" is doing
rather a lot of work just to
prove that "fib 3" equals "2"!
Happily, the \emph{calculational}
nature of such proofs allows us
to develop the following proof
search algorithm \pbesym that
is inspired by model checking~\cite{CGL92}
\begin{itemize}[leftmargin=*]
\item \emphbf{Guard Normal Form:}
First, as shown in the definition
of "fibP" in~\S~\ref{subsec:overview:reflection}, each reflected
function is transformed into a
\emph{guard normal form}
$\wedge_i (\pred_i \Rightarrow \fapp{\rfun}{x} = \body_i)$
\ie a collection of \emph{guards} $\pred_i$
and their corresponding definition $\body_i$.

\item \emphbf{Unfolding:}
Second, given a VC of the form
$\Penv \Rightarrow \pred$,
we iteratively \emph{unfold}
function application terms in
$\Penv$ and $\pred$ by
\emph{instantiating} them with
the definition corresponding to
an \emph{enabled} guard, where
we check enabled-ness by
querying the SMT solver.
For example, given the VC
$\ttrue \Rightarrow \fib{3} = 2$,
the guard $3 > 1$ of the application $\fib{3}$ is
trivially enabled,
\ie is true, hence we
strengthen the hypothesis
$\Penv$ with the equality
$\fib{3} = \fib{(3-1)} + \fib{(3-2)}$
corresponding to unfolding the
definition of "fib" at "3".

\item \emphbf{Fixpoint:}
We repeat the unfolding process
until either the VC is proved
or we have reached a fixpoint,
\ie no further unfolding is
enabled.
For example, the fixpoint
computation of $\fib{3}$ unfolds the definition
of "fib" at "3", "2", "1", and "0"
and then stops as no further
guards are enabled.
\end{itemize}


\mypara{Automatic Equational Reasoning}
In \S~\ref{sec:pbe} we formalize
a notion of \emph{equational proof}
and show that the proof search procedure
\pbesym enjoys two key properties.
First, that it is guaranteed to find
an equational proof if one can be constructed 
from unfoldings of function definitions.
(The user must still provide instantiations 
of lemmas and induction hypotheses.) 
Second, that under certain conditions readily
met in practice, it is guaranteed to terminate.
These two properties allow us to use \pbesym
to predictably automate proofs: the programmer needs
only to supply the relevant induction hypotheses
or helper lemma applications.
The remaining long chains of calculations
are performed automatically via SMT-based
\pbesym.
That is, the user must provide case statements 
and the recursive structure, but can elide 
the long chains of \lstinline{=.} applications.
To wit, with complete proof search,
the proofs of~\S~\ref{subsection:overview:equational} shrink to:
\begin{mcode}
    fib3_2 :: {!fib! 3 = 2}       fMono :: Mono / [y]
    fib3_2 = ()                 fMono f up x y    
                                  | x+1 == y = up x
                                  | x+1  < y = up (y-1) &&& fMono up x (y-1)
\end{mcode}
%
where the combinator "p &&& q = ()"
inserts the propositions "p" and "q" to the
VC hypothesis.

\mypara{\pbesym vs. Axiomatization}
Existing SMT based verifiers like
\dafny~\citep{dafny} and \fstar~\citep{fstar}
use the classical \emph{axiomatic} approach
to verify assertions over user-defined
functions like "fib".
In these systems, the function is encoded in the logic
as a universally quantified formula (or axiom):
$\forall n.\ \fibdef\ n$
after which the SMT solver may
instantiate the above axiom at
"3", "2", "1" and "0" in order to
automatically prove "{fib 3 = 2}".

The automation offered by axioms
is a bit of a devil's bargain,
as axioms render VC checking
\emph{undecidable}, and in
practice automatic axiom
instantiation can easily
lead to infinite ``matching loops''.
For example, the existence of
a term \fib{n} in a VC can
trigger the above axiom,
which may then produce
the terms \fib{(n-1)}
and \fib{(n-2)}, which
may then recursively
give rise to further
instantiations
\emph{ad infinitum}.
To prevent matching
loops an expert must
carefully craft ``triggers''
or alternatively, provide
a ``fuel'' parameter
that bounds the depth of
instantiation~\citep{Amin2014ComputingWA}.
Both these approaches
ensure termination,
but can cause the
axiom to not be
instantiated at
the right places,
thereby rendering
the VC checking
\emph{incomplete}.
The incompleteness
is illustrated by
the following example
from the \dafny benchmark
suite~\cite{dafny-github}
\begin{mcode}
    pos n | n < 0     = 0                test   :: y:{y > 5} -> {!pos! n = 3 + !pos! (n-3)}
          | otherwise = 1 + pos (n-1)    test _ = ()
\end{mcode}
%
\dafny (and \fstar's) fuel-based approach fails
to check the above, when the fuel value is less
than $3$. One could simply raise-the-fuel-and-try-again
but at what point does the user know when to stop?
In contrast, \pbesym (1)~does not require any fuel
parameter, (2)~is able to automatically perform
the required unfolding to verify this example,
\emph{and} (3)~is guaranteed to terminate.

\begin{figure}[t!]
\centering
\begin{minipage}[t]{0.45\textwidth}
\begin{mcodef}
app_assoc :: AppendAssoc
app_assoc [] ys zs
  = ([] ++ ys) ++ zs
  =. ys ++ zs
  =. [] ++ (ys ++ zs)         ** QED
app_assoc (x:xs) ys zs
  = ((x : xs) ++  ys)  ++ zs
  =. (x : (xs ++  ys)) ++ zs
  =.  x :((xs ++  ys)  ++ zs)
      $\because$ app_assoc xs ys zs
  =.  x : (xs ++ (ys   ++ zs))
  =. (x : xs) ++ (ys   ++ zs) ** QED
\end{mcodef}
\end{minipage}
\hspace{0.14in}
\begin{minipage}[t]{0.49\textwidth}
\begin{mcodef}
app_assoc             :: AppendAssoc
app_assoc []     ys zs = ()
app_assoc (x:xs) ys zs = app_assoc xs ys zs
\end{mcodef}
\begin{mcodef}
app_right_id          :: AppendNilId
app_right_id []        = ()
app_right_id (x:xs)    = app_right_id xs
\end{mcodef}
\begin{mcodef}
map_fusion            :: MapFusion
map_fusion f g []      = ()
map_fusion f g (x:xs)  = map_fusion f g xs
\end{mcodef}
\end{minipage}
\caption{\textbf{(L)} Equational proof of append associativity. \quad
  \textbf{(R)} \pbesym proof, also of append-id and map-fusion.
}
\label{fig:list-laws}
\end{figure}

\subsection{Case Study: Laws for Lists}
\label{sec:overview:lists}

Reflection and \pbesym are not limited to integers.  We end the
overview by showing how they verify textbook properties of lists
equipped with append ("++") and "map" functions:
\begin{mcode}
    reflect (++) :: [a] -> [a] -> [a]       reflect map :: (a -> b) -> [a] -> [b]
    []     ++ ys = ys                       map f []     = []
    (x:xs) ++ ys = x : (xs ++ ys)           map f (x:xs) = f x : map f xs
\end{mcode}
In \S~\ref{subsec:embedding}
we will describe how the reflection mechanism
illustrated via "fibP" is extended to account for ADTs
using SMT-decidable selection and projection operations,
which reflect the definition of "xs ++ ys" into the refinement as:
$\ite{\mathtt{isNil}\ \mathit{xs}}
    {\mathit{ys}}
    {\mathtt{sel1}\ \mathit{xs}\
      \dcons\
      (\mathtt{sel2}\ \mathit{xs} \ \mathtt{++}\  \mathit{ys})}$.
%
%
We require an explicit "reflect" annotation
as not all Haskell functions can be 
reflected into logic, either because it 
is unsound to do so (\eg due to divergence)
or because of limitations of our current 
implementation. 
Recall that \toolname verifies that all reflected functions, 
like "(++)" and "map" here, are total~\citep{Vazou14}
and rejects the code otherwise.

\mypara{Laws}
We can specify various laws about lists with refinement types.
For example, the below laws state that
(1)~appending to the right is an \emph{identity} operation,
(2)~appending is an \emph{associative} operation, and
(3)~map \emph{distributes} over function composition:
\begin{mcode}
    type AppendNilId = xs:_ -> { xs ++ [] = xs }
    type AppendAssoc = xs:_ -> ys:_ -> zs:_ -> { xs ++ (ys ++ zs) = (xs ++ ys) ++ zs }
    type MapFusion   = f:_  -> g:_  -> xs:_ -> { map (f . g) xs = map f (map g xs) }
\end{mcode}

\mypara{Proofs}
On the right in Figure~\ref{fig:list-laws}
we show the proofs of these laws using \pbesym,
which should be compared to the classical equational
proof~\eg~by \citet{WadlerCalculating}, shown on the left.
With \pbesym, the user need only
provide the high-level structure
--- the case splits and invocations
of the induction hypotheses ---
after which \pbesym automatically
completes the rest of the equational
proof.
Thus using SMT-based \pbesym, "app_assoc"
shrinks down to its essence: an induction
over the list "xs".
The difference is even more stark with
"map_fusion" whose full equational proof
is omitted, as it is twice as long.

\mypara{\pbesym vs. Normalization}
The proofs in Figure~\ref{fig:list-laws}
may remind readers familiar with Type-Theory
based proof assistants (\eg \toolfont{Coq} or
\toolfont{Agda}) of the notions of
\emph{type-level normalization} and
\emph{rewriting} that permit similar
proofs in those systems.
While our approach of \pbesym is
inspired by the idea of type
level computation, it differs
from it in two significant ways.
First, from a \emph{theoretical} point of view,
SMT logics are not equipped with any notion
of computation, normalization, canonicity
or rewriting.
Instead, our \pbesym algorithm shows how
to \emph{emulate} those ideas by asserting
equalities corresponding to function
definitions (Theorem~\ref{thm:completeness}).
Second, from a \emph{practical} perspective,
the combination of (decidable) SMT-based 
theory reasoning and \pbesym's proof search 
can greatly simplify verification.
For example, consider the "swap" function
from~\citet{appel-vfa}'s \coq textbook:
\begin{mcode}
    swap :: [Integer] -> [Integer]
    swap (x1:x2:xs) = if x1 > x2 then x2:x1:x2 else x1:x2:xs
    swap xs         = xs
\end{mcode}
In Figure~\ref{fig:swap} we show four proofs
that "swap" is idempotent:
Appel's proof using \coq (simplified by
the use of a hint database and the
arithmetic tactic "omega"), its variant in
\agda (for any Decidable Partial Order),
the \pbesym proof, and 
a proof using the \dafny verifier.
It is readily apparent that \pbesym's
proof search, working
hand-in-glove with SMT-based theory
reasoning, makes proving the result
trivial in comparison to \coq or \agda.
Of course, proof assistants like
\agda, \coq, and \toolfont{Isabelle}
emit easily checkable certificates
and have decades-worth of tactics,
libraries, and proof scripts that
enable large scale proof engineering.
On the other hand, \dafny's fuel-based 
axiom instantiation automatically 
unfolds the definition of "swap" twice, 
thereby completing the proof 
without any user input.
These heuristics are 
orthogonal to  \pbesym and can 
be combined with it, if the user 
wishes to trade off predictability 
for even more automation.

%

\begin{figure}[t!]
  \begin{center}
    \includegraphics[width=5.5in]{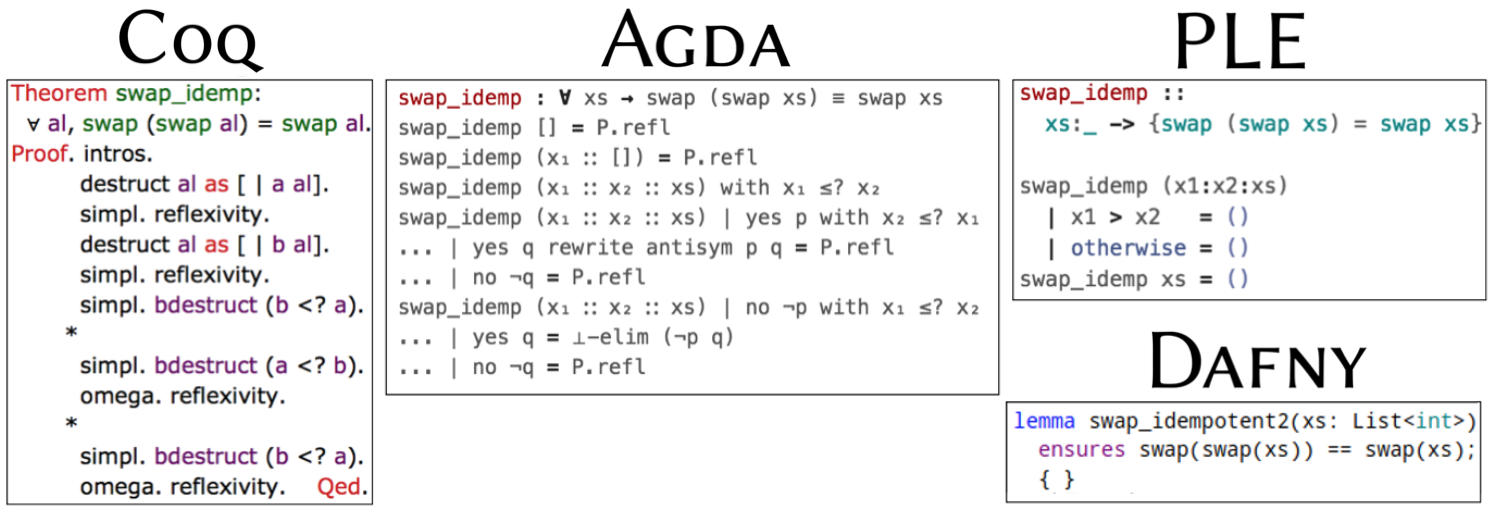}
  \end{center}
  \vspace{-3mm}
  \caption{Proofs that \texttt{swap} is idempotent with
    \coq, \agda, \dafny and \pbesym.
  }
  \label{fig:swap}
\end{figure}

\mypara{Summary}
We saw an overview of an SMT-automated refinement 
type checker that achieves SMT-decidable checking 
by restricting verification conditions to be 
\emph{quantifier-free} and hence, decidable.
In existing SMT-based verifiers (\eg \dafny) there 
are two main reasons to introduce quantifiers, namely
(1)~to express quantified \emph{specifications} and 
(2)~to encode the semantics of \emph{user-defined functions}.
Next, we use propositions-as-types to encode 
quantified specifications and in~\S~\ref{sec:formalism} 
we show how to encode the semantics of user-defined functions 
via refinement reflection.

\newcommand\expror[2]{\ensuremath{\texttt{Either}\ #1\ #2}\xspace}
\newcommand\exprand[2]{\ensuremath{(#1, #2)}\xspace}
\newcommand\exprorleft[1]{\ensuremath{\texttt{Left}\ #1}\xspace}
\newcommand\exprorright[1]{\ensuremath{\texttt{Right}\ #1}\xspace}
\newcommand\tyand[2]{\ensuremath{(#1, #2)}\xspace}
\newcommand\tyor[2]{\ensuremath{\texttt{Either}\ #1\ #2}\xspace}
\newcommand\tyex[3]{\ensuremath{(#1::#2,#3)}\xspace}
\newcommand\tyall[3]{\ensuremath{#1:#2 \rightarrow #3}\xspace}

\section{Embedding Natural Deduction with Refinement Types}
\label{sec:higher-order}

In this section we show how user-provided
\emph{quantified} specifications
can be naturally encoded using
$\lambda$-abstractions and
dependent pairs to encode
universal and existential
quantification, respectively.
Proof terms can be generated
using the standard natural
deduction derivation rules,
following Propositions as Types~\cite{Wadler15}
(also known as the Curry-Howard
isomorphism~\cite{CHIso}).
What is new is that we exploit this encoding to show for the first time that a refinement
type system can represent any proof
in Gentzen's natural deduction~\cite{Gentzen35}
while still taking advantage of
SMT decision procedures to automate the
quantifier-free portion of natural
deduction proofs.
For simplicity, in this section we assume
all terms are total.

\subsection{Propositions: Refinement Types}

Figure~\ref{fig:hol} maps logical predicates
to types constructed over quantifier-free refinements.

\mypara{Native terms}
Native terms consist of all of the
(quantifier-free) expressions
of the refinement languages.
In \S~\ref{sec:formalism}
we formalize refinement typing
in a core calculus \corelan where
refinements include (quantifier-free)
terminating expressions.

\mypara{Boolean connectives}
Implication $\formula_1 \Rightarrow \formula_2$
is encoded as a \emph{function} from the proof of
$\formula_1$ to the proof of $\formula_2$.
Negation is encoded as an implication where
the consequent is "False".
Conjunction $\formula_1 \land \formula_2$
is encoded as the \emph{pair} "($\formula_1$, $\formula_2$)"
that contains the proofs of \emph{both} conjuncts
and disjunction $\formula_1 \lor \formula_2$
is encoded as the \emph{sum} type "Either"
that contains the proofs of \emph{one of}
the disjuncts, \ie where "data Either a b = Left a | Right b".

\mypara{Quantifiers}
Universal quantification $\forall x. \formula$
is encoded as lambda abstraction $x:\typ \rightarrow \formula$
and eliminated by function application.
Existential quantification $\exists x. \formula$
is encoded as a dependent pair "(x::$\typ$, $\formula$)"
that contains the term $x$ and a proof of a formula
that depends on $x$.
Even though refinement type systems do not
traditionally come with explicit syntax for
dependent pairs, one can encode dependent
pairs in refinements using abstract
refinement types~\cite{Vazou13}
which do not add extra complexity
to the system.
Consequently, we add the syntax for
dependent pairs in Figure~\ref{fig:hol}
as syntactic sugar for abstract refinements.

\begin{figure}[t!]
$$
\begin{array}{rcccc}
                    & \quad  & \textbf{Logical Formula} & \quad & \textbf{Refinement Type}  \\
\text{Native Terms} & \quad&   $e$      & & \proofterm{e}    \\
\text{Implication}  & & \formula_1 \Rightarrow \formula_2 &      & {\formula_1} \rightarrow {\formula_2}     \\
\text{Negation}     & & \neg \formula                  & & \formula \rightarrow \proofterm{\efalse} \\
\text{Conjunction}  & & \formula_1 \land \formula_2    & & \tyand{\formula_1}{\formula_2} \\
\text{Disjunction}  & &  \formula_1 \lor \formula_2    & & \tyor{\formula_1}{\formula_2}   \\
\text{Forall}       & &\forall x. \formula             & & \tyall{x}{\typ}{\formula}                \\
\text{Exists}       & & \exists x. \formula            & & \tyex{x}{\typ}{\formula}                 \\
\end{array}
$$
\caption{Mapping from logical predicates to quantifier-free refinement types.
\proofterm{e} abbreviates \prooftermlarge{e}. Function binders are not relevant for
negation and implication, and hence, elided.}
\figrule
\label{fig:hol}
\end{figure}

\subsection{Proofs: Natural Deduction}

We overload \formula to be both a proposition and
a refinement type.
We connect these two meanings of \formula by
using the Propositions as Types~\cite{Wadler15},
to prove that if there exists an expression (or proof term) with
refinement type \formula, then the proposition
\formula is valid.

We construct proofs terms using Gentzen's
natural deduction system~\cite{Gentzen35},
whose rules map directly to refinement type
derivations.
The rules for natural deduction arise from the propositions-as-types reading
of the standard refinement type
checking rule (to be defined in
\S~\ref{sec:formalism}) \ch{\env}{e}{\formula}
as ``\formula is provable under the assumptions of \env''.
%
We write \gentzenValid{\env}{\formula}
for Gentzen's natural deduction judgment
``under assumption \env, proposition \formula holds''.
Then, each of Gentzen's logical rules
can be recovered from the rules in
Figure~\ref{fig:ch} by rewriting
each judgment \ch{\env}{e}{\formula}
of \corelan as \gentzenValid{\env}{\formula}.
For example, conjunction and universal
elimination can be derived as:
$$
\inference{
	\gentzenValid{\env}{{\formula_1} \lor{\formula_2}} &&
	\gentzenValid{\env,\formula_1}{\formula} &&
	\gentzenValid{\env,\formula_2}{\formula} &&
}{
	\gentzenValid{\env}{\formula}
}
\lor\small{\textsc{\!-E}}
\quad
\inference{
	\gentzenValid{\env}{e_x\ \text{term}} &&
	\gentzenValid{\env}{\forall x. \formula}
}{
	\gentzenValid{\env}{\formula\subst{x}{e_x}}
}
\forall\small{\textsc{\!-E}}
$$

\mypara{Programs as Proofs}
As Figure~\ref{fig:ch} directly maps
natural deduction rules to derivations
that are accepted by refinement typing,
we conclude that if there exists a
natural deduction derivation for a
proposition \formula, then there
exists an expression that has the refinement type \formula.
\begin{theorem}\label{theorem:embedding}
If\ \gentzenValid{\env}{\formula},
then we can construct an $e$ such that\ \ch{\env}{e}{\formula}.
\end{theorem}

Note that our embedding is \textit{not} an isomorphism, 
since the converse of Theorem~\ref{theorem:embedding} does not hold. 
As a counterexample, the law of the excluded middle
can be proved in our system
(\ie we can construct an "Either" term $e$, so that 
$\ch{\bind{p}{\proofterm{\typbool\ | \ \etrue}}}{e}{p\lor\neg p}$), 
but cannot be proved using natural deduction
(\ie $\gentzenValid{{\proofterm{\etrue}}\not}{p\lor\neg p}$).
The reason for that is that our system is using the classical logic
of the SMTs, which includes the law of the excluded middle.
On the contrary, in intuitionistic systems that also encode natural deduction
(\eg \coq, \idris, \nuPRL)
the law of the excluded middle should be axiomatized.

\subsection{Examples}\label{sec:expressiveness:examples}
Next, we illustrate our encoding
with examples of proofs for
quantified propositions
ranging from
textbook logical tautologies,
properties of datatypes like lists,
and induction on natural numbers.
%

\mypara{Natural Deduction as Type Derivation}
We illustrate the mapping from natural deduction rules to typing rules
in Figure~\ref{fig:ex:ch} which uses typing judgments to express
Gentzen's proof of the proposition
$$
\phi \equiv (\exists x. \forall y. (p\ x\ y)) \Rightarrow (\forall y. \exists x. (p\ x\ y))
$$
Read bottom-up, the derivation provides a proof of \formula.
Read top-down, it constructs a \emph{proof} of the formula
as the \emph{term}
$\lambda e\ y. \texttt{case}\ e\ \texttt{of}\ \{
    (x,e_x) \rightarrow (x, e_x\ y)\}$.
This proof term corresponds directly to the following Haskell expression
that typechecks with type \formula.
\begin{mcode}
    exAll :: p:(a -> a -> Bool) -> (x::a, y:a -> {p x y}) -> y:a -> (x::a, {p x y})
    exAll _ = \e y -> case e of {(x, ex) -> (x, ex y)}
\end{mcode}
%

\begin{figure}
$$
\inference{
	\inference{
		\inference{
			\bind{e}{\formulap},
			\bind{y}{\typ_y}
			\vdash
			 e
			 :
	        \formulap &&
	        \inference{
	        \bind{e}{\formulap},
			\bind{y}{\typ_y},
			\bind{x}{\type_x},
			\bind{e_x}{\formulax}
			\vdash
			 e_x
			 :
	         \formulax
	         \\
	        \bind{e}{\formulap},
			\bind{y}{\typ_y},
			\bind{x}{\type_x},
			\bind{e_x}{\formulax}
			\vdash
			 y
			 :
	         \typ_y
	        }{
	        \bind{e}{\formulap},
			\bind{y}{\typ_y},
			\bind{x}{\type_x},
			\bind{e_x}{\formulax}
			\vdash
			 e_x\ y
			 :
	         p\ x\ y
	        }\forallElim
		}{
			\bind{e}{\formulap},
			\bind{y}{\typ_y}
			\vdash
			\texttt{case}\ e\ \texttt{of}\ \{
                                (x,e_x) \rightarrow (x, e_x\ y)
                        \}
                        :
	        \exists x. (p\ x\ y)
		}\existsElim
	}{
	\bind{e}{\formulap}
	\vdash
	\lambda y. \texttt{case}\ e\ \texttt{of}\ \{
    (x,e_x) \rightarrow (x, e_x\ y)
    \}
    :
	\forall y. \exists x. (p\ x\ y)
	}\forallIntro
}
{
\emptyset
\vdash
\lambda e\ y. \texttt{case}\ e\ \texttt{of}\ \{
    (x,e_x) \rightarrow (x, e_x\ y)
\}
:
{(\exists x. \forall y. (p\ x\ y)) \Rightarrow (\forall y. \exists x. (p\ x\ y))}
}
\impIntro
$$
\caption{Proof of $(\exists x. \forall y. (p\ x\ y)) \Rightarrow (\forall y. \exists x. (p\ x\ y))$
where
$\formulap \equiv \exists x. \forall y. (p\ x\ y)
$,
$
\formulax \equiv \forall y. (p\ x\ y)
$.
}
\figrule
\label{fig:ex:ch}
\end{figure}

\mypara{SMT-aided proofs}
The great benefit of using refinement types
to encode natural deduction is that the quantifier-free
portions of the proof can be automated via SMTs.
For every quantifier-free proposition \formula,
you can convert between $\{\formula\}$,
where \formula is treated as an SMT-proposition
and \formula, where \formula is treated as a
type; and this conversion goes \emph{both} ways.
For example, let $\formula \equiv p \land (q \lor r)$.
Then "flatten" converts from \formula to
$\{\formula\}$ and "expand" the other way,
while this conversion is SMT-aided.
\begin{mcode}
    flatten :: p:_ -> q:_ -> r:_-> ({p}, Either {q} {r}) -> {p && (q || r)}
    flatten (pf, Left  qf) = pf &&& qf
    flatten (pf, Right rf) = pf &&& rf

    expand :: p:_ -> q:_ -> r:_ -> {p && (q || r)} -> ({p}, Either {q} {r})
    expand proof | q = (proof, Left  proof)
    expand proof | r = (proof, Right proof)
\end{mcode}

\mypara{Distributing Quantifiers}
%
Next, we construct the proof terms needed
to prove two logical properties:
that existentials distribute over disjunctions
and foralls over conjunctions, \ie
%
\begin{align}
  \formula_\exists \ & \equiv (\exists x. p\ x \lor q\ x) \Rightarrow ((\exists x. p\ x) \lor (\exists x. q\ x))
  \label{eq:ex:or} \\
  \formula_\forall \ & \equiv (\forall x. p\ x \land q\ x) \Rightarrow ((\forall x. p\ x) \land (\forall x. q\ x))
  \label{eq:ex:and}
\end{align}
The specification of these properties requires
nesting quantifiers inside connectives and
vice versa.
The proof of $\formula_\exists$ (\ref{eq:ex:or})
proceeds by existential case splitting and introduction:
\begin{code}
    exDistOr :: p:_ -> q:_ -> (x::a, Either {p x} {q x})
                           -> Either (x::a, {p x}) (x::a, {q x})
    exDistOr _ _ (x, Left  px) = Left  (x, px)
    exDistOr _ _ (x, Right qx) = Right (x, qx)
\end{code}
%
%
%
Dually, we prove $\formula_\forall$ (\ref{eq:ex:and}) via
a $\lambda$-abstraction and case spitting
inside the conjunction pair:
\begin{code}
    allDistAnd :: p:_ -> q:_ -> (x:a->({p x}, {q x}))
                             -> ((x:a->{p x}), (x:a->{q x}))
    allDistAnd _ _ andx = ( (\x -> case andx x of (px, _) -> px)
                          , (\x -> case andx x of (_, qx) -> qx) )
\end{code}
The above proof term exactly corresponds to its natural deduction proof derivation
but using SMT-aided verification can get simplified to the following
\begin{code}
    allDistAnd _ _ andx = (pf, pf)
      where pf x = case andx x of (px, py) -> px &&& py
\end{code}



\mypara{Properties of User Defined Datatypes}
As \formula can describe properties
of data types like lists, we can prove
properties of such types, \eg that for
every list "xs", if there exists a list
"ys" such that "xs == ys ++ ys" ,then
"xs" has even length.
$$\formula \equiv \forall xs. ((\exists ys.\ xs = ys \ \texttt{++}\ ys) \Rightarrow (\exists n. \texttt{len}\ xs = n + n))$$
The proof ("evenLen") proceeds by existential elimination and introduction
and uses the "lenAppend" lemma, which in turn uses induction on the
input list and \pbesym to automate equational reasoning.
\begin{code}
    evenLen :: xs:[a]->(ys::[a],{xs = ys ++ ys})->(n::Int,{len xs = n+n})
    evenLen xs (ys,pf) = (len ys, lenAppend ys ys &&& pf)

    lenAppend :: xs:_ -> ys:_ -> {len (xs ++ ys) = len xs + len ys}
    lenAppend []     _  = ()
    lenAppend (x:xs) ys = lenAppend xs ys
\end{code}

\mypara{Induction on Natural Numbers}
Finally, we specify and verify \emph{induction} on natural numbers:
$$
\formula_{ind} \equiv  (p\ 0 \land (\forall n. p\ (n-1) \Rightarrow p\ n) \Rightarrow \forall n. p\ n)
$$
The proof proceeds by induction (\eg case splitting).
In the base case, "n == 0", the proof calls the left conjunct, which contains a proof of the base case.
Otherwise, "0 < n", and the proof \emph{applies} the induction hypothesis to the right conjunct instantiated at "n-1".
\begin{code}
    ind :: p:_ -> ({p 0},(n:Nat -> {p (n-1)} -> {p n})) -> n:Nat -> {p n}
    ind p (p0, pn) 0 = p0
    ind p (p0, pn) n = pn n (ind p (p0, pn) (n-1))
\end{code}

\newcommand\fst[1]{\ensuremath{\texttt{fst}\ #1}\xspace}
\newcommand\snd[1]{\ensuremath{\texttt{snd}\ #1}\xspace}

\begin{figure}[t]
\setlength\tabcolsep{0pt} 
\begin{tabular}{ccc}
\multirow{2}{*}{
$
\inference{
    \ch{\env}{\fst{e}}{\formula_1} \\
    \ch{\env}{\snd{e}}{\formula_2}
}{
	\ch{\env}{e}{\tyand{\formula_1}{\formula_2}}
}\land\small{\textsc{\!-I}}$}
& \quad &
$\inference{
	\ch{\env}{e}{\tyand{\formula_1}{\formula_2}}
}{
	\ch{\env}{\fst{e}}{\formula_1}
}
\land\small{\textsc{\!-L-E}}$\\
& \quad & $\inference{
	\ch{\env}{e}{\tyand{\formula_1}{\formula_2}}
}{
	\ch{\env}{\snd{e}}{\formula_2}
}
\land\small{\textsc{\!-R-E}}$
\\[0.12in]
$\inference{
    \ch{\env}{e_1}{\formula_1}
}{
	\ch{\env}{\exprorleft{e_1}}{\tyor{\formula_1}{\formula_2}}
}\lor\small{\textsc{\!-L-I}}$
& \quad &
\multirow{2}{*}{$
\inference{
	\ch{\env}{e}{\tyor{\formula_1}{\formula_2}} \\
	\ch{\env,\bind{x_1}{\formula_1}}{e_1}{\formula} &&
	\ch{\env,\bind{x_2}{\formula_2}}{e_2}{\formula}
}{
	\arraycolsep=0.5pt
	\begin{array}{lll}
	\env \vdash \texttt{case}\ e\ \text{of}\ \{&
		\exprorleft{x_1}\rightarrow e_1; & \\
	& \exprorright{x_2}\rightarrow e_2 \}\ & : \formula	\\
	\end{array}
}\lor\small{\textsc{\!-E}}$}\\
\\
$\inference{
    \ch{\env}{e_1}{\formula_2}
}{
	\ch{\env}{\exprorright{e_2}}{\tyor{\formula_1}{\formula_2}}
}\lor\small{\textsc{\!-R-I}}$ & \quad&
\\[0.12in]
$\inference{
	\ch{\env, \bind{x}{\formula_x}}{e}{\formula}
}{
	\ch{\env}{\lambda x. e}{\formula_x \rightarrow \formula}
}\impIntro$ &
\quad &
$\inference{
	\ch{\env}{e}{\formula_x \rightarrow \formula} &&
	\ch{\env}{e_x}{\formula_x}
}{
	\ch{\env}{e\  e_x}{\formula}
}
\Rightarrow\small{\textsc{\!-E}}
$
\\[0.12in]
$
\inference{
	\ch{\env, \bind{x}{\typ}}{e}{\formula}
}{
	\ch{\env}{\lambda x. e}{(x:\typ \rightarrow \formula)}
}\forallIntro$ &
\quad
&
$
\inference{
	\ch{\env}{e_x}{\typ} &&
	\ch{\env}{e}{(x:\typ \rightarrow \formula)}
}{
	\ch{\env}{e\ e_x}{\formula\subst{x}{e_x}}
}
\forallElim
$
\\[0.12in]
$
\inference{
	\ch{\env}{\fst{e}}{\typ}  &&
	\ch{\env, \bind{x}{\typ}}{\snd{e}}{\formula}  &&
}{
	\ch{\env}{e}{\tyex{x}{\typ}{\formula\subst{x}{\fst{e}}}}
}\exists\small{\textsc{\!-I}}$ &\quad \quad &
$
\inference{
	\ch{\env}{e}{\tyex{x}{\typ}{\formula_x}} &&
	\ch{\env, \bind{x}{\typ}, \bind{y}{\formula_x}}{e'}{\formula}
}{
	\ch{\env}{\texttt{case}\ e\ \texttt{of}\ \{ (x, y) \rightarrow e'\}}{\formula}
}
\existsElim
$
\end{tabular}
%
%
\caption{Natural deduction rules for refinement types.
With $[\texttt{fst}|\texttt{snd}]\  e \equiv \texttt{case}\ e\ \texttt{of}\ \{(x_1, x_2)
\rightarrow [x_1 | x_2 ] \}$ .}
\label{fig:ch}
\figrule
\end{figure}

\subsection{Consequences}
To summarize, we use the propositions-as-types
principle to make two important contributions.
First, we show that natural deduction reasoning
can smoothly co-exist with SMT-based verification
to automate the decidable, quantifier-free portions
of the proof.
%

Second, we show for the first time how
natural deduction proofs can be encoded
in refinement type systems like \toolname
and we expect this encoding to extend,
in a straightforward manner to other
SMT-based deductive verifiers
(\eg \dafny and \fstar).
This encoding shows that refinement type systems
are expressive enough to encode any intuitionistic natural deduction proof,
gives a guideline for encoding proofs
with nested quantifiers, and provides
a pleasant implementation of
natural deduction that
is pedagogically useful.

\section{Refinement Reflection: \corelan}
\label{sec:formalism}
\label{sec:types-reflection}
\label{sec:theory}

Refinement reflection encodes recursive functions
in the quantifier-free, SMT logic and it is
formalized in three steps.
First, we develop a core calculus \corelan with an
\emph{undecidable} type system based on denotational
semantics and prove it sound. 
%
Next, in \S~\ref{sec:theory:algorithmic} we define
a language \smtlan that soundly approximates \corelan
while enabling decidable, SMT-based type checking.
Finally, in \S~\ref{sec:pbe} we develop a complete
proof search algorithm to automate equational reasoning.

\subsection{Syntax}

\begin{figure}[t]
\begin{minipage}[t]{0.49\textwidth}
\input{fig-theoryR}
\end{minipage}
\ \
\begin{minipage}[t]{0.49\textwidth}
\input{fig-theoryS}
\end{minipage}

\label{fig:lang-syntax}
\caption{\textbf{(Left) Syntax of \corelan}: Denotational Typing.
         \quad
         \textbf{(Right) Syntax of \smtlan:} Algorithmic Typing.
         }
\label{fig:syntax}
\label{fig:smtsyntax}
\figrule
\end{figure}

Figure~\ref{fig:syntax} summarizes the syntax of \corelan,
which is essentially the calculus \undeclang~\cite{Vazou14}
with explicit recursion and a special $\erefname$ binding
to denote terms that are reflected into the refinement logic.
The elements of \corelan are constants, values,
expressions, binders, and programs.

\mypara{Constants}
The constants of \corelan
include primitive relations $\odot$,
here, the set $\{ =, <\}$.
Moreover, they include the
booleans $\etrue$, $\efalse$,
integers $\mathtt{-1}, \mathtt{0}$, $\mathtt{1}$, \etc,
and logical operators $\mathtt{\land}$, $\mathtt{\lnot}$, \etc.

\mypara{Data Constructors}
Data constructors are special constants.
For example, the data type \tintlist, which represents
finite lists of integers, has two data constructors: $\dnull$ (nil)
and $\dcons$ (cons).

\mypara{Values \& Expressions}
The values of \corelan include
constants, $\lambda$-abstractions
$\efun{x}{\typ}{e}$, and fully
applied data constructors $D$
that wrap values.
The expressions of \corelan
include values, variables $x$,
applications $\eapp{e}{e}$, and
$\mathtt{case}$ expressions.

\mypara{Binders \& Programs}
A \emph{binder} $\bd$ is a series of possibly recursive
$\mathtt{let}$ definitions, followed by an expression.
A \emph{program} \prog is a series of $\erefname$
definitions, each of which names a function
that is reflected into the refinement
logic, followed by a binder.
The stratification of programs via binders
is required so that arbitrary recursive definitions
are allowed in the program but cannot be inserted into the logic
via refinements or reflection.
(We \emph{can} allow non-recursive $\mathtt{let}$
binders in expressions $e$, but omit them for simplicity.)

\subsection{Operational Semantics}
We define $\hookrightarrow$ to be the small step, call-by-name
$\beta$-reduction semantics for \corelan.
We evaluate reflected terms 
as recursive $\mathtt{let}$ bindings, with termination
constraints imposed by the type system:
$$
\erefb{x}{\gtyp}{e}{\prog}
\hookrightarrow
\eletb{x}{\gtyp}{e}{\prog}
$$
We define $\evalsto{}{}$ to be the reflexive,
transitive closure of $\evals{}{}$.
Moreover, we define $\betaeq{}{}$ to be the reflexive,
symmetric, and transitive closure of $\evals{}{}$.

\mypara{Constants} Application of a constant requires the
argument be reduced to a value; in a single step, the
expression is reduced to the output of the primitive
constant operation, \ie ${c\ v} \hookrightarrow \ceval{c}{v}$.
For example, consider $=$, the primitive equality
operator on integers.
We have $\ceval{=}{n} \defeq =_n$
where $\ceval{=_n}{m}$ equals \etrue
iff $m$ is the same as $n$.
%
%

\mypara{Equality}
We assume that the equality operator
is defined \emph{for all} values,
and, for functions, is defined as
extensional equality.
%
%
That is, for all
$f$ and $f'$,
$\evals{(f = f')}{\etrue}$
   $\mbox{iff}$
  $\forall v.\ \betaeq{f\ v}{f'\ v}$.
We assume source \emph{terms} only contain implementable equalities
over non-function types; while function extensional equality only
appears in \emph{refinements}. 
\VC{This statement is very unclear to me.}

\subsection{Types}

\corelan types include basic types, which are \emph{refined} with predicates,
and dependent function types.
\emph{Basic types} \btyp comprise integers, booleans, and a family of data-types
$T$ (representing lists, trees \etc).
For example, the data type \tintlist represents lists of integers.
We refine basic types with predicates (boolean-valued expressions \refa) to obtain
\emph{basic refinement types} $\tref{v}{\btyp}{\refa}$.
We use \tlabel to mark provably terminating
computations and use
refinements to ensure that if
${{e}\text{:}{\tref{v}{\btyp^\tlabel}{\refa'}}}$,
then $e$ terminates.
As discussed by~\citet{Vazou14} termination labels
are checked using refinement types and are used
to ensure that refinements cannot diverge as required
for soundness of type checking under lazy evaluation.
Termination checking is crucial for this work,
as combined with checks for exhaustive definitions (\S~\ref{subsection:overview:refinements}),
it ensures totality (well-formedness)
of expressions as required by
propositions-as-types (\S~\ref{sec:higher-order})
and termination of \pbesym (\S~\ref{sec:pbe}).
Finally, we have dependent \emph{function types} $\trfun{x}{\typ_x}{\typ}$
where the input $x$ has the type $\typ_x$ and the output $\typ$ may
refer to the input binder $x$.
We write $\btyp$ to abbreviate $\tref{v}{\btyp}{\etrue}$
and \tfunbasic{\typ_x}{\typ} to abbreviate \trfun{x}{\typ_x}{\typ}, if
$x$ does not appear in $\typ$.
\VC{Kind of hard to read, even though $x$
  isn't bound, the type still has $x$ as a
  subscript. Compare to, $a:A \to B$, and
  $A \to B$}

\mypara{Denotations}
Each type $\typ$ \emph{denotes} a set of expressions $\interp{\typ}$,
that is defined via the operational semantics~\cite{Knowles10}.
Let $\shape{\typ}$ be the type we get if we erase all refinements
from $\typ$ and $\bhastype{}{e}{\shape{\typ}}$ be the
standard typing relation for the typed lambda calculus.
Then, we define the denotation of types as:
\begin{align*}
\interp{\tref{x}{\btyp}{r}} \defeq &\
  \{ e \mid  \bhastype{}{e}{\btyp}
           , \mbox{ if } \evalsto{e}{w} \mbox{ then } \evalsto{r\subst{x}{w}}{\etrue}
  \}\\
\interp{\tref{x}{\btyp^\tlabel}{r}} \defeq &\
  \interp{\tref{x}{\btyp}{r}} \cap \{ e \mid  \exists\ w. \evalsto{e}{w}\}\\
\interp{\trfun{x}{\typ_x}{\typ}} \defeq &\
  \{ e \mid \bhastype{}{e}{\shape{\tfunbasic{\typ_x}{\typ}}}
           , \forall e_x \in \interp{\typ_x}.\ (\eapp{e}{e_x}) \in \interp{\typ\subst{x}{e_x}}
  \}
\end{align*}

\mypara{Constants}
For each constant $c$ we define its type \constty{c}
such that $c \in \interp{\constty{c}}$.
For example,
$$
\begin{array}{lcl}
\constty{3} &\doteq& \tref{v}{\tint^\tlabel}{v = 3}\\
\constty{+} &\doteq& \trfun{\ttx}{\tint^\tlabel}{\trfun{\tty}{\tint^\tlabel}{\tref{v}{\tint^\tlabel}{v = x + y}}}\\
\constty{\leq} &\doteq& \trfun{\ttx}{\tint^\tlabel}{\trfun{\tty}{\tint^\tlabel}{\tref{v}{\tbool^\tlabel}{v \Leftrightarrow x \leq y}}}\\
\end{array}
$$
\VC{The right hand side is now a
  refinement type, not a set like the
  previous one.}

\subsection{Refinement Reflection}\label{subsec:logicalannotations}
Reflection \emph{strengthens} function output types
with a refinement that \emph{reflects} the
definition of the function in the logic.
We do this by treating each
$\erefname$-binder
(${\erefb{f}{\gtyp}{e}{\prog}}$)
as a $\eletname$-binder
(${\eletb{f}{\exacttype{\gtyp}{e}}{e}{\prog}}$)
during type checking (rule $\rtreflect$ in Figure~\ref{fig:typing}).

\mypara{Reflection}
We write \exacttype{\typ}{e} for the \emph{reflection}
of the term $e$ into the type $\typ$, defined as 
$$
\begin{array}{lcl}
\exacttype{\tref{v}{\btyp{^\tlabel}}{r}}{e}
  & \defeq
  & \tref{v}{\btyp{^\tlabel}}{r \land v = e}\\
\exacttype{\trfun{x}{\typ_x}{\typ}}{\efun{x}{}{e}}
  & \defeq
  & \trfun{x}{\typ_x}{\exacttype{\typ}{e}}
\end{array}
$$
As an example, recall from \S~\ref{sec:overview}
that the "reflect fib" strengthens the type of
"fib" with the refinement "fibP".
That is, let the user specified type of \fibname
be $t_\fibnameblack$ and the its definition be definition
$\lambda n . e_\fibnameblack$.
\begin{align*}
t_\fibnameblack & \defeq {\tfun{\ttnat}{\ttnat}} \\
e_\fibnameblack & \defeq \efib
\end{align*}
Then, the reflected type of \fibname will be:
$$
\exacttype{t_\fibnameblack}{e_\fibnameblack} = \trfun{n}{\ttnat}{\tref{v}{\tint^\tlabel}{0 \leq v \land v = e_\fibname}}
$$

\mypara{Termination Checking}
We defined \exacttype{\cdot}{\cdot} to be a \textit{partial}
function that only reflects provably terminating expressions,
\ie expressions whose result type is marked with \tlabel.
If a non-provably terminating function is reflected
in an \corelan expression then type checking will
fail (with a reflection type error in the implementation).
This restriction is crucial for soundness,
as diverging expressions can lead to inconsistencies.
For example, reflecting the diverging "f x = 1 + f x"
into the logic leads to an inconsistent system that
is able to prove "0 = 1".

\mypara{Automatic Reflection}
Reflection of \corelan expressions into the refinements happens
automatically by the type system, not manually by the user.
The user simply annotates a function $f$ as $\annotReflect\ f$.
Then, the rule $\rtreflect$ in Figure~\ref{fig:typing}
is used to type check the reflected function by strengthening
$f$'s result via \exacttype{\cdot}{\cdot}.
Finally, the rule \rtlet is used to check that the
automatically strengthened type of $f$ satisfies $f$'s implementation.

\RN{Hmm, the ordering is a bit weird here where there are a couple references to
fig 5, but then the next subsec starts with an introductory tone.}

\subsection{Typing Rules}
\begin{figure}[!t]
\input{typing-rules}
\input{wellformedness-rules}
\input{subtyping-rules}
\caption{Typing of \corelan.} 
\label{fig:typing}
\figrule
\end{figure}

Next, we present the type-checking rules of \corelan, as found in Figure~\ref{fig:typing}.

\mypara{Environments and Closing Substitutions}
A \emph{type environment} $\env$ is a sequence of type bindings
$\tbind{x_1}{\typ_1},\ldots,\tbind{x_n}{\typ_n}$. An environment
denotes a set of \emph{closing substitutions} $\sto$ which are
sequences of expression bindings:
$\gbind{x_1}{e_1}, \ldots,$ $\gbind{x_n}{e_n}$ such that:
$$
\interp{\env} \defeq  \{\sto \mid \forall \tbind{x}{\typ} \in \Env.
                                    \sto(x) \in \interp{\applysub{\sto}{\typ}} \}
$$
where $\applysub{\sto}{\typ}$ applies a substitution to a type (and
  likewise $\applysub{\sto}{\prog}$, to a program).

A reflection environment $\RRenv$ is a sequence that binds the names
of the reflected functions with their definitions
$\gbind{f_1}{e_1},\ldots,\gbind{f_n}{e_n}$.
A reflection environment respects a type environment
when all reflected functions satisfy their types:
$$
\env \models \RRenv \defeq
\forall (\gbind{f}{e}) \in \RRenv.\;
\exists \typ.\;
(\tbind{f}{\typ}) \in \env \wedge
(\hastype{\env}{\RRenv}{e}{\typ})
$$

\mypara{Typing}
A judgment \hastype{\env}{\RRenv}{\prog}{\typ} states that
the program $\prog$ has the type $\typ$ in
the type environment $\env$ and the reflection environment $\RRenv$.
That is, when the free variables in $\prog$ are
bound to expressions described by $\env$, the
program $\prog$ will evaluate to a value
described by $\typ$.

\mypara{Rules}
All but two of the rules are standard~\cite{Knowles10,Vazou14}
except for the addition of the reflection environment $\RRenv$ at each rule.
First, rule \rtreflect is used to
extend the reflection environment with the binding of the function name
with its definition ($f \mapsto e$)
and moreover to strengthen the type of each
reflected binder with its definition, as described previously
in \S~\ref{subsec:logicalannotations}.
Second, rule \rtexact strengthens the expression with
a singleton type equating the value and the expression
(\ie reflecting the expression in the type).
This is a generalization of the ``selfification'' rules
from \cite{Ou2004,Knowles10} and is required to
equate the reflected functions with their definitions.
For example, the application "fib 1" is typed as
${\tref{v}{\tint^\tlabel}{\fibdef\ 1 \wedge v = \fib\ 1}}$ where
the first conjunct comes from the (reflection-strengthened)
output refinement of "fib" (\S~\ref{sec:overview}) and
the second comes from rule~\rtexact.

\mypara{Well-formedness}
A judgment \iswellformed{\env}{\typ} states that
the refinement type $\typ$ is well-formed in
the environment $\env$.
Following~\citet{Vazou14}, $\typ$ is well-formed if all
the refinements in $\typ$ are $\tbool$-typed,
provably terminating expressions in $\env$.


\mypara{Subtyping}
A judgment \issubtype{\env}{\RRenv}{\typ_1}{\typ_2} states
that the type $\typ_1$ is a subtype of 
$\typ_2$ in the environments $\env$ and $\RRenv$.
%
%
Informally, $\typ_1$ is a subtype of $\typ_2$ if, when
the free variables of $\typ_1$ and $\typ_2$
are bound to expressions described by $\env$,
the denotation of $\typ_1$
is \emph{contained in} the denotation of $\typ_2$.
Subtyping of basic types reduces to denotational
containment checking, shown in rule \rsubbasecore.
That is, $\typ_1$ is a subtype of $\typ_2$ under $\env$
if for any closing substitution $\sto$ in $\interp{\env}$,
$\interp{\applysub{\sto}{\typ_1}}$ is contained in
$\interp{\applysub{\sto}{\typ_2}}$.

\mypara{Soundness}
We prove that 
typing implies denotational
inclusion and 
evaluation preserves typing.
\begin{theorem}{[Soundness of \corelan]}\label{thm:safety}
\begin{itemize}
\item\textbf{Denotations}
If $\;\hastype{\env}{\RRenv}{\prog}{\typ}$, then
$\forall \sto\in \interp{\env}. \applysub{\sto}{\prog} \in \interp{\applysub{\sto}{\typ}}$.
\item\textbf{Preservation}
If \;\hastype{\emptyset}{\emptyset}{\prog}{\typ}
       and $\evalsto{\prog}{w}$, then $\hastype{\emptyset}{\emptyset}{w}{\typ}$.
\end{itemize}
\end{theorem}

The proofs can be found in~\cite{appendix}.
Theorem~\ref{thm:safety} lets us
prove that if \formula is a \corelan type
interpreted as a proposition
(using the mapping of Figure~\ref{fig:hol})
and if there exists a $\prog$ so that
$\hastype{\emptyset}{\emptyset}{\prog}{\formula}$,
the \formula is valid.
%
%
For example, in \S~\ref{sec:overview} we verified
that the term \fibincrname proves
${\trfun{n}{\tnat}{\ttref{\fib{n} \leq \fib{(n+1)}}}}$.
Via soundness of \corelan, we get that for each valid input "n",
the result refinement is valid.
$$
\forall n. \evalsto{0 \leq n}{\etrue} \Rightarrow \evalsto{\texttt{fib}\ n \leq \texttt{fib}\ {(n+1)}}{\etrue}
$$

\section{Algorithmic Checking: \smtlan}
\label{sec:theory:algorithmic}
\label{sec:algorithmic}

\smtlan is a first order approximation of \corelan
where higher-order features are
approximated with uninterpreted functions and the undecidable
type subsumption rule \rsubbasecore is replaced with a
decidable one (\ie, \rsubbasepbe), yielding an sound
and decidable SMT-based algorithmic type system.
Figure~\ref{fig:smtsyntax} summarizes the syntax
of \smtlan, the \emph{sorted} (SMT-)
decidable logic of quantifier-free equality,
uninterpreted functions and linear
arithmetic (QF-EUFLIA) ~\citep{Nelson81,SMTLIB2}.
The \emph{terms} of \smtlan include
integers $n$,
booleans $b$,
variables $x$,
data constructors $\dc$ (encoded as constants),
fully applied unary \unop and binary \binop operators,
and application $x\ \overline{\pred}$ of an uninterpreted function $x$.
The \emph{sorts} of \smtlan include the built-in
\tint and \tbool to represent
integers and booleans and the 
uninterpreted \tuniv to represent data types.
The interpreted functions of \smtlan, \eg
the logical constants $=$ and $<$,
have the function sort $\sort \rightarrow \sort$.
Other functional values in \corelan, \eg
reflected \corelan functions and
$\lambda$-expressions, are represented as
first-order values with
the uninterpreted sort \tsmtfun{\sort}{\sort}.

\subsection{Transforming \corelan into \smtlan}
\label{subsec:embedding}
\label{subsec:measures}

The judgment
\tologicshort{\env}{e}{\typ}{\pred}{\sort}{\smtenv}{\axioms}
states that a $\corelan$ term $e$ is transformed,
under an environment $\env$, into a
$\smtlan$ term $\pred$.
If ${\tologicshort{\env}{\refa}{}{\pred}{}{}{}}$
and $\env$ is clear from the context we write
\tosmt{\refa} and \tohaskell{\pred}
to denote the translation from \corelan to \smtlan and back.
Most of the transformation rules are identity
and can be found in~\cite{appendix}.
Here we discuss the non-identity ones.

\mypara{Embedding Types}
We embed \corelan types into \smtlan sorts as:
$$\arraycolsep=1.4pt
\begin{array}{lclclclclcl}
\tosmt{\tint}   & \defeq &  \tint  & \qquad\qquad &
\tosmt{T}       & \defeq &  \tuniv & \qquad\qquad &
\tosmt{\tref{v}{\btyp^{[\tlabel]}}{e}} &\defeq & \tosmt{\btyp} \\
\tosmt{\tbool}  & \defeq &  \tbool & \quad\quad &
                &        &         & \quad\quad &
\tosmt{\trfun{x}{\typ_x}{\typ}} & \defeq & \tsmtfun{\tosmt{\typ_x}}{\tosmt{\typ}}
\end{array}$$
\mypara{Embedding Constants}
Elements shared on both \corelan and \smtlan
translate to themselves.
These elements include
booleans,
integers,
variables,
binary
and unary
operators.
SMT solvers do not support currying,
and so in \smtlan, all function symbols
must be fully applied.
Thus, we assume that all applications
to primitive constants and data
constructors are fully applied, \eg by converting
source terms like "(+ 1)" to
"(\z -> z + 1)".

\mypara{Embedding Functions}
As \smtlan is first-order, we
embed $\lambda$s using the
uninterpreted function \smtlamname{}{}.
$$
\inference{
    \tologicshort{\env, \tbind{x}{\typ_x}}{e}{}{\pred}{}{}{} &
        \hastype{\env}{\emptyset}{(\efun{x}{}{e})}{(\trfun{x}{\typ_x}{\typ})}\\
}{
	\tologicshort{\env}{\efun{x}{}{e}}{(\trfun{x}{\typ_x}{\typ})}
	        {\smtlamname{\tosmt{\typ_x}}{\tosmt{\typ}}\ {x}\ {\pred}}
	        {\sort'}{\smtenv, \tbind{f}{\sort'}}{\andaxioms{\{\axioms_{f_1}, \axioms_{f_2}\}}{\axioms}}
}
$$
The term $\efun{x}{}{e}$ of type
${\typ_x \rightarrow \typ}$ is transformed
to
${\smtlamname{\sort_x}{\sort}\ x\ \pred}$
of sort
${\tsmtfun{\sort_x}{\sort}}$, where
$\sort_x$ and $\sort$ are respectively
$\tosmt{\typ_x}$ and $\tosmt{\typ}$,
${\smtlamname{\sort_x}{\sort}}$
is a special uninterpreted function
of sort
${\sort_x \rightarrow \sort\rightarrow\tsmtfun{\sort_x}{\sort}}$,
and
$x$ of sort $\sort_x$ and $\pred$ of sort $\sort$ are
the embeddings of the binder and body, respectively.
As $\smtlamname{}{}$ is an SMT-function,
it \emph{does not} create a binding for $x$.
Instead, $x$ is renamed to
a \emph{fresh} SMT name.

\mypara{Embedding Applications}
We use defunctionalization~\citep{Reynolds72} to 
embed applications. 
%
$$
\inference{
	\tologicshort{\env}{e_x}{\typ_x}{\pred_x}{\tosmt{\typ_x}}{\smtenv}{\axioms'}
	&
	\tologicshort{\env}{e}{\trfun{x}{\typ_x}{\typ}}{\pred}{\tsmtfun{\tosmt{\typ_x}}{\tosmt{\typ}}}{\smtenv}{\axioms}
	&
	\hastype{\env}{\emptyset}{e}{{\typ_x}\rightarrow{\typ}}
}{
	\tologicshort{\env}{e\ e_x}{\typ}{\smtappname{\tosmt{\typ_x}}{\tosmt{\typ}}\ {\pred}\ {\pred_x}}{\tosmt{\typ}}{\smtenv}{\andaxioms{\axioms}{\axioms'}}
}
$$
The term ${e\ e_x}$, where $e$ and $e_x$ have
types ${\typ_x \rightarrow \typ}$ and $\typ_x$,
is transformed to
${\tbind{\smtappname{\sort_x}{\sort}\ \pred\ \pred_x}{\sort}}$
where
$\sort$ and $\sort_x$ are $\tosmt{\typ}$ and $\tosmt{\typ_x}$,
the
${\smtappname{\sort_x}{\sort}}$
is a special uninterpreted function of sort
${\tsmtfun{\sort_x}{\sort} \rightarrow \sort_x \rightarrow \sort}$,
and
$\pred$ and $\pred_x$ are the respective translations of $e$ and $e_x$.

\mypara{Embedding Data Types}
We embed data constructors to a predefined
\smtlan constant ${\smtvar{\dc}}$ of
sort ${\tosmt{\constty{\dc}}}$:
$
	\tologicshort{\env}{\dc}{\constty{\dc}}{\smtvar{\dc}}{\tosmt{\constty{\dc}}}{\emptyset}{\emptyaxioms}
$.
For each datatype, we 
create reflected functions
that \emph{check} the top-level
constructor and \emph{select}
their individual fields.
For example, for lists, we
create the functions
\begin{mcode}
    isNil []     = True         isCons (x:xs) = True          sel1 (x:xs) = x
    isNil (x:xs) = False        isCons []     = False         sel2 (x:xs) = xs
\end{mcode}
The above selectors can be
modeled precisely in the refinement
logic via SMT support for ADTs~\cite{Nelson81}.
To generalize, let $\dc_i$ be a
data constructor such that
$
\constty{\dc_i} \defeq \typ_{i,1} \rightarrow \dots \rightarrow \typ_{i,n} \rightarrow \typ
$.
Then \emph{check}
${\checkdc{{\dc_i}}}$ has the sort
$\tsmtfun{\tosmt{\typ}}{\tbool}$
and \emph{select}
${\selector{\dc}{i,j}}$ has the sort
$\tsmtfun{\tosmt{\typ}}{\tosmt{\typ_{i,j}}}$.

\mypara{Embedding Case Expressions}
We translate case-expressions
of \corelan into nested $\mathtt{if}$
terms in \smtlan, by using the check
functions in the guards and the
select functions for the binders
of each case.
$$
\inference{
	\tologicshort{\env}{e}{\typ_e}{\pred}{\tosmt{\typ_e}}{\smtenv}{\axioms} &&
	\tologicshort{\env}{e_i\subst{\overline{y_i}}{\overline{\selector{\dc_i}{}\ x}}\subst{x}{e}}{\typ}{\pred_i}{\tosmt{\typ}}{\smtenv}{\axioms_i}
}{
	\tologicshort{\env}{\ecase{x}{e}{\dc_i}{\overline{y_i}}{e_i}}{\typ}
	 {\eif{\smtappname{}{}\ \checkdc{\dc_1}\ \pred}{\pred_1}{\ldots} \ \mathtt{else}\ \pred_n}{\tosmt{\typ}}{\smtenv}
	 {\andaxioms{\axioms}{\axioms_i}}
}
$$
The above translation yields the reflected definition
for append "(++)" from~(\S~\ref{sec:overview:lists}).

\mypara{Semantic Preservation}
The translation preserves the semantics of the expressions.
Informally,
if ${\tologicshort{\env}{\refa}{}{\pred}{}{}{}}$,
then for every substitution $\theta$ and every logical model $\sigma$
that respects the environment $\env$
if $\evalsto{\applysub{\theta}{\refa}}{v}$
then $\sigma \models \pred = \tosmt{v}$.
\VC{What's a logical model?}

\subsection{Algorithmic Type Checking}\label{sec:pbe:tc}

We make the type checking from Figure~\ref{fig:typing}
algorithmic by checking subtyping via our novel,
SMT-based \emph{Proof by Logical Evaluation}(\pbesym).
Next, we formalize how \pbesym makes checking algorithmic
and in \S~\ref{sec:pbe} we describe the \pbesym procedure in detail.

\mypara{Verification Conditions}
Recall that in ~\S~\ref{subsec:embedding} we defined
$\tosmt{\cdot}$ as the translation from \corelan
to \smtlan.
Informally, the implication or \emph{verification condition} (VC)
$\tosmt{\env} \Rightarrow {\pred_1} \Rightarrow {\pred_2}$ is \emph{valid} only if the
set of values described by $\pred_1$ is subsumed by
the set of values described by $\pred_2$ under the
assumptions of \env.
$\env$ is embedded into logic by conjoining
the refinements of terminating binders~\cite{Vazou14}:
$$
\tosmt{\env} \defeq \bigcup_{x \in \env} \tosmt{\env, x} 
\quad \mbox{where we embed each binder as} \quad
\tosmt{\env, x} \defeq \begin{cases}
                           \tosmt{e}  & \text{if } \env(x)=\tref{x}{\btyp^{\tlabel}}{e}\\
                           \etrue & \text{otherwise}.
                         \end{cases}
$$

\mypara{Validity Checking}
Instead of directly using the VCs to check validity of programs, we
use the procedure \pbesym that strengthens the assumption environment \tosmt{\env}
with equational properties.
Concretely,
given a
reflection environment $\RRenv$,
type environment $\env$, and
expression $e$,
the procedure $\pbe{\tosmt{\RRenv}}{\tosmt{\env}}{\tosmt{e}}$
--- we will define $\tosmt{\RRenv}$ in \S~\ref{sec:pbe:algo} ---
returns \ttrue only when the expression $e$ evaluates
to \etrue under the reflection and type environments
\RRenv and \env.

\mypara{Subtyping via VC Validity Checking}
We make subtyping, and hence, typing decidable,
by replacing the denotational base subtyping
rule \rsubbasecore with the conservative,
algorithmic version \rsubbasepbe that uses
\pbesym to check the validity of the
subtyping.
$$
\inference{
\pbe{\tosmt{\RRenv}}{\tosmt{\env,\tbind{v}{\tref{v}{\btyp^\tlabel}{\refa_1}}}}{\tosmt{\refa_2}}
}{
	\pbeissubtype{\env}{\RRenv}{\tref{v}{\btyp}{\refa_1}}{\tref{v}{\btyp}{\refa_2}}
}[\rsubbasepbe]
$$

This typing rule is sound as functions reflected in \RRenv
always respect the typing environment \env (by construction)
and because \pbename is sound (Theorem~\ref{thm:pbe:sound}).
\begin{lemma}\label{lem:subtyping} 
If {\pbeissubtype{\env}{\RRenv}{\tref{v}{\btyp}{e_1}}{\tref{v}{\btyp}{e_2}}}
then {\issubtype{\env}{\RRenv}{\tref{v}{\btyp}{e_1}}{\tref{v}{\btyp}{e_2}}}.
\end{lemma}

\mypara{Soundness of \smtlan}
We write $\pbehastype{\env}{\RRenv}{e}{\typ}$
for the judgments that can be derived by the
algorithmic subtyping rule \rsubbasepbe (instead of \rsubbasecore.)
Lemma~\ref{lem:subtyping} implies the soundness of \smtlan.
\begin{theorem}[Soundness of \smtlan]\label{thm:soundness-smt}
If \pbehastype{\env}{\RRenv}{e}{\typ}, then \hastype{\env}{\RRenv}{e}{\typ}.
\end{theorem}

\section{Complete Verification: Proof by Logical Evaluation} \label{sec:pbe}

Next, we formalize our
Proof By Logical Evaluation
algorithm \pbesym and show
that it is sound~(\S~\ref{sec:pbe:algo}),
that it is complete with respect to
equational proofs~(\S~\ref{sec:pbe:complete}),
and that it terminates~(\S~\ref{sec:pbe:terminates}).

\begin{figure}[t!]
\begin{tabular}{>{$}r<{$} >{$}r<{$} >{$}r<{$} >{$}l<{$} >{$}l<{$} }
\emphbf{Terms}    & \pred, \expr, \body  & ::=   & \smtlan\ \text{\kw{if}-free terms from Figure~\ref{fig:smtsyntax}}            & \\ [0.05in]

\emphbf{Functions}    & \Rfun  & ::=    & \rdef{x}{\pred}{\body}                   & \\ [0.05in]
\emphbf{Definitional Environment}  & \Renv  & ::=    & \emptyset \spmid \rfun \mapsto \Rfun, \Renv & \\ [0.05in]
\emphbf{Logical Environment} & \Penv  & ::=    & \emptyset \spmid \pred, \Penv               & \\ [0.05in]
\end{tabular}
\caption{{Syntax of Predicates, Terms and Reflected Functions.}}
\label{fig:pbe:syntax}
\end{figure}

\subsection{Algorithm}\label{sec:pbe:algo}

Figure~\ref{fig:pbe:syntax} describes
the input environments for \pbename.
The logical environment $\Penv$ contains
a set of hypotheses $p$, described in
Figure~\ref{fig:smtsyntax}.
The definitional environment $\Renv$
maps function symbols $f$ to their
definitions $\rdef{x}{\pred}{\body}$,
written as $\lambda$-abstractions
over guarded bodies.
Moreover, the body $\body$ and the
guard $\pred$ contain neither
$\lambda$ nor \texttt{if}.
These restrictions do not impact
expressiveness: $\lambda$s can be
named and reflected, and
\kw{if}-expressions can
be pulled out into top-level
guards using $\defIf\cdot$,
defined in~\citep{appendix}.
A definitional environment $\Renv$
can be constructed from $\RRenv$ as
$$
\tosmt{\RRenv} \defeq \{
   \gbind{f}{\lambda\overline{x}. \defIf{\tosmt{\refa}}} |
  (\gbind{f}{\lambda\overline{x}. \refa})\in\RRenv
  \}
$$

\begin{figure}[t!]
$$\begin{array}{lcl}
\toprule
\uinstsym & : & (\Renv, \Penv) \rightarrow \Penv \\
\midrule
\uinst{\Renv}{\Penv} & = &
  \Penv \cup \
  \bigcup_{\issubterm{\fapp{\rfun}{\expr}}{\Penv}}
    \inst{\Renv}{\Penv}{\rfun}{\exprs} \\[0.1in]

\inst{\Renv}{\Penv}{\rfun}{\exprs} & = &
  \left \{ \SUBSTS{\left(\tosmt{\fapp{\rfun}{x}} = \body_i\right)}{x}{\expr}\;
  \left|\;\, \left( \pred_i \Rightarrow \body_i \right) \in \overline{d}
  ,      \;\; \smtIsValid{\Penv}{\SUBSTS{\pred_i}{x}{\expr}}
  \right.
  \right \}\\
\quad \mbox{where} & & \\
\quad \quad \rbody{x}{d} & = & \Renv(\rfun) \\[0.1in]

\midrule
\pbesym & : & (\Renv, \Penv, \pred) \rightarrow \tbool \\
\midrule
\pbe{\Renv}{\Penv}{\pred} & = & \pbeloop{0}{\,\Penv\cup \
  \bigcup_{\issubterm{\fapp{\rfun}{\expr}}{\pred}}
    \inst{\Renv}{\Penv}{\rfun}{\exprs} } \\
\quad \mbox{where} & & \\
\quad \quad \pbeloop{i}{\Penv_i} & & \\
\quad \quad \spmid \smtIsValid{\Penv_i}{\pred}      & = & \ttrue \\
\quad \quad \spmid \Penv_{i+1} \subseteq \Penv_i    & = & \tfalse \\
\quad \quad \spmid \mbox{otherwise}                 & = & \pbeloop{i+1}{\Penv_{i+1}} \\
\quad \quad \quad  \mbox{where}                     &   & \\
\quad \quad \quad  \quad \Penv_{i+1}                & = & \Penv \cup \uinst{\Renv}{\Penv_i} \\
\bottomrule
\end{array}$$
\caption{\textbf{Algorithm \pbesym:} Proof by Logical Evaluation.}
\label{fig:pbe}
\label{fig:inst}
\label{fig:uinst}
\end{figure}

\mypara{Notation}
We write
$\issubterm{\rfun\ (\expr_1, \dots, \expr_n)}{\Penv}$
if the \smtlan term
$(\text{app} \dots (\text{app}\ f \ \expr_1) \dots \expr_n)$
is a syntactic subterm of some $\expr \in \Penv$.
We abuse notation to write
$\issubterm{\fapp{\rfun}{\expr}}{\expr'}$
for
$\issubterm{\fapp{\rfun}{\expr}}{\{\expr'\}}$.
We write \smtIsValid{\Penv}{\pred}
for SMT validity of the implication
$\Penv \Rightarrow \pred$.

\mypara{Instantiation \& Unfolding}
A term $\predb$ is a \emph{$(\Renv,\Penv)$-instance} if
there exists $\issubterm{\fapp{\rfun}{\expr}}{\Penv}$
such that:
\begin{itemize}
\item $\rlookup{\Renv}{\rfun} \equiv \rdef{x}{\pred_i}{\body_i}$,
\item $\smtIsValid{\Penv}{\SUBSTS{\pred_i}{x}{\expr}}$, and
\item $\predb \equiv \SUBSTS{(\fapp{\rfun}{x} = \body_i)}{x}{\expr}$.
\end{itemize}
A set of terms $Q$ is a \emph{$(\Renv, \Penv)$-instance}
if every $\predb \in Q$ is an $(\Renv, \Penv)$-instance.
The \emph{unfolding} of $\Renv, \Penv$ is the
(finite) set of all $(\Renv, \Penv)$-instances.
Procedure $\uinst{\Renv}{\Penv}$ shown in
Figure~\ref{fig:uinst} computes and returns
the conjunction of $\Penv$ and the unfolding
of $\Renv, \Penv$.
The following properties relate $(\Renv, \Penv)$-instances
to the semantics of \corelan and SMT validity.
Let \refapply{\RRenv}{e} denote the evaluation
of $e$ under the reflection environment $\RRenv$, \ie
$\refapply{\emptyset}{e} \doteq e$
and
$\refapply{(\RRenv, \bind{f}{e_f})}{e} \doteq \refapply{\RRenv}{\text{let rec } f = e_f \text{ in } e}$.

\begin{lemma}\label{lem:pbe:semantics}
For every $\env \models \RRenv$ and $\store\in \embed{\env}$,
\begin{itemize}
\item\textbf{Sat-Inst}  \label{lem:sat-inst}
If $\tosmt{e}$ is a $(\tosmt{\RRenv},\tosmt{\env})$-instance,
then \rsatisfies{\RRenv}{\store}{e}.
\item\textbf{SMT-Approx}\label{lem:smt-approx}
If $\smtIsValid{\tosmt{\env}}{\tosmt{e}}$,
then \rsatisfies{\RRenv}{\store}{e}.
\item\textbf{SMT-Inst} \label{lem:smt-inst}
If $\predb$ is a $(\tosmt{\RRenv},\tosmt{\env})$-instance
and $\smtIsValid{\extendpenv{\tosmt{\env}}{\predb}}{\tosmt{e}}$,
then $\rsatisfies{\RRenv}{\store}{e}$.
\end{itemize}
\end{lemma}


\mypara{The Algorithm}
Figure~\ref{fig:pbe} shows our proof search algorithm
$\pbe{\Renv}{\Penv}{\pred}$ which takes as input
a set of \emph{reflected definitions} $\Renv$,
an \emph{hypothesis} $\Penv$, and
a \emph{goal} $\pred$.
The \pbename procedure recursively \emph{unfolds}
function application terms by invoking $\uinstsym$
until either the goal can be proved using the
unfolded instances (in which case the search
returns $\ttrue$) \emph{or} no new instances
are generated by the unfolding (in which case
the search returns $\tfalse$).

\mypara{Soundness}
First, we prove the soundness of \pbesym. 

\begin{theorem}[\textbf{Soundness}] \label{thm:pbe:sound}
If   $\pbe{\tosmt{\RRenv}}{\tosmt{\env}}{\tosmt{e}}$
then $\forall \store\in\embed{\env}$,
\evalsto{\apply{\store}{\refapply{\RRenv}{e}}}{\etrue}.
\end{theorem}

We prove Theorem~\ref{thm:pbe:sound}
using the Lemma~\ref{lem:pbe:semantics} 
that relates instantiation, SMT validity,
and the exact semantics.
Intuitively, \pbesym is sound as it
reasons about a finite set of instances
by \emph{conservatively} treating all
function applications as \emph{uninterpreted}~\cite{Nelson81}.

\subsection{Completeness} \label{sec:pbe:complete}

Next, we show that our proof search
is \emph{complete} with respect to
equational reasoning.
We define a notion of equational
proof
\eqpf{\Renv}{\Penv}{\expr}{\expr'}
and prove that
if there exists such a proof,
then
\pbe{\Renv}{\Penv}{\expr = \expr'}
is guaranteed to return \ttrue.
To prove this theorem, we introduce
the notion of \emph{bounded unfolding}
which corresponds to unfolding
definitions $n$ times.
We show that unfolding preserves
congruences, and hence, that an equational
proof exists \textit{iff} the goal can be proved with
\emph{some} bounded unfolding.
Thus, completeness follows by showing that
the proof search procedure computes the limit
(\ie fixpoint) of the bounded unfolding.
In \S~\ref{sec:pbe:terminates}
we show that the fixpoint is computable:
there exists an unfolding depth at which \pbesym
reaches a fixpoint and hence terminates.

\mypara{Bounded Unfolding}
For every $\Renv, \Penv$, and $0 \leq n$,
the \emph{bounded unfolding of depth} $n$ is
defined by:
$$\begin{array}{lcl}
  \binst{\Renv}{\Penv}{0}     & \doteq & \Penv \\
  \binst{\Renv}{\Penv}{n+1}   & \doteq & \Penv_n \cup \uinst{\Renv}{\Penv_n}
  \quad \mbox{where}\ \Penv_n = \binst{\Renv}{\Penv}{n}
\end{array}$$
That is, the unfolding at depth $n$ essentially
performs $\uinstsym$ upto $n$ times.
The bounded-unfoldings yield a monotonically
non-decreasing sequence of formulas such that
if two consecutive bounded unfoldings coincide,
then all subsequent unfoldings are the same.

\begin{lemma}[\textbf{Monotonicity}] \label{lem:bounded:mono}
$\forall 0 \leq n.\ \binst{\Renv}{\Penv}{n} \subseteq \binst{\Renv}{\Penv}{n + 1}$.
\end{lemma}

\begin{lemma}[\textbf{Fixpoint}] \label{lem:bounded:fix}
Let $\Penv_i \doteq \binst{\Renv}{\Penv}{i}$.
If $\Penv_n = \Penv_{n+1}$,
then $\forall n < m.\ \Penv_m = \Penv_n$.
\end{lemma}

\mypara{Uncovering}
Next, we prove that every function application
term that is \emph{uncovered} by unfolding to
depth $n$ is congruent to a term in the $n$-depth
unfolding.

\begin{lemma}[\textbf{Uncovering}] \label{lem:uncovering}
Let $\Penv_n \equiv \binst{\Renv}{\extendpenv{\Penv}{\vv = \expr}}{n}$.
If   $\smtIsValid{\Penv_n}{\vv = \expr'}$,
then for every $\issubterm{\fapp{\rfun}{t'}}{\expr'}$
       there exists $\issubterm{\fapp{\rfun}{t}}{\Penv_n}$
         such that $\smtIsValid{\Penv_n}{t_i = t'_i}$.
\end{lemma}

We prove the above lemma by induction on $n$
where the inductive step uses the following
property of congruence closure, which itself
is proved by induction on the structure of
$\expr'$:

\begin{lemma}[\textbf{Congruence}] \label{lem:congruence}
If   $\smtIsValid{\extendpenv{\Penv}{\vv = \expr}}{\vv = \expr'}$
     and
     $\vv \not \in \Penv, \expr, \expr'$,
then for every $\issubterm{\fapp{\rfun}{t'}}{\expr'}$
       there exists $\issubterm{\fapp{\rfun}{t}}{\Penv, \expr}$
         such that $\smtIsValid{\Penv}{t_i = t'_i}$.
\end{lemma}

\mypara{Unfolding Preserves Equational Links}
We use the uncovering Lemma~\ref{lem:uncovering}
and congruence to show that \emph{every instantiation}
that is valid after $n$ steps is subsumed by the $n+1$
depth unfolding.
That is, we show that every possible \emph{link} in
any equational chain can be proved equal to
the source expression via bounded unfolding.

\begin{lemma}[\textbf{Link}] \label{lem:link}
If   $\smtIsValid{\binst{\Renv}{\extendpenv{\Penv}{\vv = \expr}}{n}}{\vv = \expr'}$,
then $\smtIsValid{\binst{\Renv}{\extendpenv{\Penv}{\vv = \expr}}{n+1}}{\uinst{\Renv}{\extendpenv{\Penv}{\vv = \expr'}}}$.
\end{lemma}

\begin{figure}[t]
$$
\inference{
}{
	\eqpf{\Renv}{\Penv}{\expr}{\expr}
}[\rulename{Eq-Refl}]
$$

$$\inference
  { \eqpf{\Renv}{\Penv}{\expr}{\expr''}
    \quad
    \Penv' = \uinst{\Renv}{\extendpenv{\Penv}{\vv = \expr''}}
    \quad
    \smtIsValid{\Penv'}{\vv = \expr'}
  }
  {\eqpf{\Renv}{\Penv}{\expr}{\expr'}}
 [\rulename{Eq-Trans}]
$$

$$\inference
  { \eqpf{\Renv}{\Penv}{\expr_1}{\expr_1'}
    \quad
		\eqpf{\Renv}{\Penv}{\expr_2}{\expr_2'}
		\quad
    \smtIsValid{\Penv}{\expr_1' \binop \expr_2'}
  }
  {\symeqpf{\Renv}{\Penv}{\expr_1}{\binop}{\expr_2}}
  [\rulename{Eq-Proof}]
$$
\caption{\textbf{Equational Proofs:} rules for equational reasoning.}
\label{fig:equational-proof}
\figrule
\end{figure}

\mypara{Equational Proof}
Figure~\ref{fig:equational-proof} formalizes
our rules for equational reasoning.
Intuitively, there is an \emph{equational proof}
that $\expr_1 \binop \expr_2$ under $\Renv, \Penv$,
written by the judgment
$\symeqpf{\Renv}{\Penv}{\expr_1}{\binop}{\expr_2}$,
if by some sequence of repeated function unfoldings,
we can prove that $\expr_1$ and $\expr_2$ are
respectively equal to $\expr_1'$ and $\expr_2'$
such that,
$\smtIsValid{\Penv}{\expr_1' \binop \expr_2'}$
holds.
Our notion of equational proofs
adapts the idea of type level
computation used in TT-based
proof assistants to the setting of
SMT-based reasoning, via the
directional unfolding judgment
\eqpf{\Renv}{\Penv}{\expr}{\expr'}.
In the SMT-realm, the explicit
notion of a normal or canonical
form is converted to the implicit
notion of the equivalence classes
of the SMT solver's congruence
closure procedure (post-unfolding).

\mypara{Completeness of Bounded Unfolding}
We use the fact that unfolding preserves
equational links to show that bounded unfolding
is \emph{complete} for equational proofs.
That is, we prove by induction on the structure
of the equational proof that whenever there is
an \emph{equational proof} of $\expr = \expr'$,
there exists some bounded unfolding that
suffices to prove the equality.

\begin{lemma}\label{lem:bounded-completeness}
  If   $\eqpf{\Renv}{\Penv}{\expr}{\expr'}$,
  then $\exists 0 \leq n.\ \smtIsValid{\binst{\Renv}{\extendpenv{\Penv}{\vv = \expr}}{n}}{\vv = \expr'}$.
\end{lemma}

\mypara{\pbesym is a Fixpoint of Bounded Unfolding}
We show that the proof search procedure \pbesym computes
the least-fixpoint of the bounded unfolding and hence,
returns \ttrue \textit{iff} there exists \emph{some} unfolding
depth $n$ at which the goal can be proved.

\begin{lemma}[\textbf{Fixpoint}]\label{lem:fixpoint}
$\pbe{\Renv}{\Penv}{\expr = \expr'}$ iff
$\exists n.\ \smtIsValid{\binst{\Renv}{\extendpenv{\Penv}{\vv = \expr}}{n}}{\vv = \expr'}$.
\end{lemma}

The proof follows by observing that
$\pbe{\Renv}{\Penv}{\expr = \expr'}$
computes the \emph{least-fixpoint}
of the sequence
$\Penv_i \doteq \binst{\Renv}{\Penv}{i} \label{eq:fix:1}$.
Specifically, we prove by induction on $i$
that at each invocation of $\pbeloop{i}{\Penv_i}$
in Figure~\ref{fig:pbe}, $\Penv_i$ is equal to
$\binst{\Renv}{\extendpenv{\Penv}{\vv = \expr}}{i}$,
which then yields the result.

\mypara{Completeness of \pbesym}
Finally, we combine Lemmas~\ref{lem:fixpoint}
and~\ref{lem:link} to show that
\pbesym is complete, \ie if there is an equational proof
that $\expr \binop \expr'$ under $\Renv, \Penv$,
then \pbe{\Renv}{\Penv}{\expr \binop \expr'} returns $\ttrue$.

\begin{theorem}[\textbf{Completeness}] \label{thm:completeness}
  If \symeqpf{\Renv}{\Penv}{\expr}{\binop}{\expr'}
  then \pbe{\Renv}{\Penv}{\expr \binop \expr'} = \ttrue.
\end{theorem}

\subsection{\pbesym Terminates} \label{sec:pbe:terminates}

So far, we have shown that our proof search
procedure $\pbesym$ is both sound and complete.
Both of these are easy to achieve
simply by \emph{enumerating} all
possible instances and repeatedly
querying the SMT solver.
Such a monkeys-with-typewriters
approach is rather impractical: it
may never terminate.
Fortunately, we show that in
addition to being sound and complete
with respect to equational proofs,
if the hypotheses are transparent,
then our proof search procedure always
terminates.
Next, we describe transparency and
explain intuitively why \pbesym
terminates.
We then develop the formalism needed
to prove the termination
Theorem~\ref{thm:termination}.

\mypara{Transparency}
An environment $\env$ is \emph{inconsistent}
if $\smtIsValid{\tosmt{\env}}{\tfalse}$.
An environment \env is \emph{inhabited}
if there exists some $\store \in \embed{\env}$.
We say $\env$ is \emph{transparent}
if it is either inhabited \emph{or}
inconsistent.
As an example of a \emph{non-transparent}
$\Penv_0$ consider the predicate
$\kw{lenA xs} = 1 + \kw{lenB xs}$,
where \kw{lenA} and \kw{lenB} are both
identical definitions of the list length
function.
Clearly there is no $\store$ that causes
the above predicate to evaluate to $\ttrue$.
At the same time, the SMT solver cannot
(using the decidable, quantifier-free theories)
prove a contradiction as that requires
induction over \kw{xs}.
Thus, non-transparent environments are
somewhat pathological and, in practice,
we only invoke $\pbesym$ on transparent
environments.
Either the environment is inconsistent,
\eg when doing a proof-by-contradiction,
or \eg when doing a proof-by-case-analysis
we can easily find suitable concrete
values via random~\cite{Claessen00}
or SMT-guided generation~\cite{Seidel15}.

\mypara{Challenge: Connect Concrete and Logical Semantics}
As suggested by its name, the
\pbesym algorithm aims to lift the
notion of evaluation or computations
into the level of the refinement logic.
Thus, to prove termination, we must
connect the two different notions of
evaluation,
the \emph{concrete} (operational) semantics
and
the \emph{logical} semantics being used by \pbesym.
This connection is trickier than appears at first glance.
In the concrete realm totality ensures that every
reflected function $\rfun$ will terminate when
run on any \emph{individual} value $\val$.
However, in the logical realm, we are working with
\emph{infinite} sets of values, compactly represented
via logical constraints.
In other words, the logical realm can be viewed
(informally) as an \emph{abstract interpretation}
of the concrete semantics. We must carefully argue
that despite the \emph{approximation} introduced
by the logical abstraction, the abstract
interpretation will also terminate.

\mypara{Solution: Universal Abstract Interpretation}
We make this termination argument in three steps.
First, we formalize how \pbesym performs
computation at the logical level via
\emph{logical steps} and \emph{logical traces}.
We show (Lemma~\ref{lem:step-reductions})
that the logical steps form a so-called
\emph{universal} (or \emph{must})
abstraction of the concrete
semantics~\cite{CousotCousot77,CGL92}.
Second, we show that if \pbesym diverges,
it is because it creates a strictly
increasing infinite chain,
$\binst{\Renv}{\Penv}{0} \subset \binst{\Renv}{\Penv}{1} \ldots$
which corresponds to an \emph{infinite logical trace}.
Third, as the logical computation
is a universal abstraction we use
inhabitation to connect the two
realms, \ie to show that an
infinite logical trace
corresponds to an infinite
concrete trace.
The impossibility of the latter
must imply the impossibility of
the former, \ie \pbesym terminates.
Next, we formalize the above
to obtain Theorem~\ref{thm:termination}.

\mypara{Totality}
A function is \emph{total} when its evaluation
reduces to exactly one value.
The totality of \RRenv can and is checked
by refinement types (\S~\ref{sec:formalism}).
Hence, for brevity, in the sequel
we \emph{implicitly assume}
that $\RRenv$ is total under $\env$.
\begin{definition}[\textbf{Total}]\label{asm:total}
Let $\body \equiv \rdef{x}{\tosmt{\pred}}{\tosmt{e}}$.
$\body$ is \emph{total} under \env and \RRenv, if
forall $\store\in \embed{\env}$:
\begin{enumerate}
\item if $\rsatisfies{\RRenv}{\store}{\pred_i}$, then
$ \exists \val.\ \evalsto{{\apply{\store}{\refapply{\RRenv}{e_i}}}}{\val}$,
\item if $\rsatisfies{\RRenv}{\store}{\pred_i}$ and $\rsatisfies{\Renv}{\store}{\pred_j}$, then $i=j$, and
\item there exists an $i$ so that  $\rsatisfies{\RRenv}{\store}{\pred_i}$.
\end{enumerate}

$\RRenv$ is \emph{total} under \env,
if every $\body \in \tosmt{\RRenv}$
is total under \env and \RRenv.
\end{definition}

\mypara{Subterm Evaluation}
As the reflected functions are total,
the Church-Rosser theorem implies that
evaluation order is not important.
To prove termination, we require
an evaluation strategy, \eg CBV,
in which if a reflected
function's guard is satisfied,
then the evaluation of the
corresponding function body
requires evaluating
\emph{every subterm}
inside the body.
As $\defIf\cdot$ hoists
\kw{if}-expressions out
of the body and into the
top-level guards, the below
fact follows from the
properties of CBV:

\begin{lemma}\label{lem:subterm-eval}
Let $\body \equiv \rdef{x}{\tosmt{\pred}}{\tosmt{e}}$ and
    $\rfun \in \RRenv$.
For every \env, \RRenv, and $\store \in \embed{\env}$,
if $\rsatisfies{\RRenv}{\store}{\pred_i}$ and
   $\issubterm{\fapp{\rfun}{\tosmt{e'}}}{\tosmt{e_i}}$,
then
   $\evalsto{{\apply{\store}{\refapply{\RRenv}{e_i}}}}{\ctxapp{\ctx}{\rawapp{\rfun}{\apply{\store}{\refapply{\RRenv}{\params{e'}}}}}}$.
\end{lemma}

\mypara{Logical Step}
A pair
$\symstep{\fapp{\rfun}{\expr}}{\fapp{\rfun'}{\expr'}}$
is a
$\Renv,\Penv$-\emph{logical step} (abbrev. step),
if
\begin{itemize}
  \item $\rlookup{\Renv}{\rfun} \equiv \rdef{x}{\pred}{\body}$,
  \item $\smtIsValid{\Penv \wedge Q}{\pred_i}$ for some $(\Renv, \Penv)$-instance $Q$, and
  \item $\issubterm{\fapp{\rfun'}{\expr'}}{\SUBSTS{\body_i}{x}{\expr}}$.
\end{itemize}

\mypara{Steps and Reductions}
Next, using Lemmas~\ref{lem:subterm-eval},~\ref{lem:pbe:semantics},
and the definition of logical steps, we show that every
logical step corresponds to a \emph{sequence} of steps
in the concrete semantics:

\begin{lemma}[\textbf{Step-Reductions}]\label{lem:step-reductions}
If   $\symstep{\fapp{\rfun}{\tosmt{e}}}{\fapp{\rfun'}{\tosmt{e'}}}$
     is a logical step under $\tosmt{\RRenv}, \tosmt{\env}$
		 and
		 $\store\in \embed{\env}$,
then $\evalsto{{\fapp{\rfun}{\apply{\store}{\refapply{\RRenv}{e}}}}}
              {\ctxapp{\ctx}{\rawapp{f}{\apply{\store}{\refapply{\RRenv}{\params{e'}}}}}}$
for some context $\ctx$.
\end{lemma}

\mypara{Logical Trace}
A sequence
$\fapp{\rfun_0}{\expr_0},\fapp{\rfun_1}{\expr_1},\fapp{\rfun_2}{\expr_2},\ldots$
is a $\Renv,\Penv$-\emph{logical trace} (abbrev. trace),
if 
$\symstep{\fapp{\rfun_i}{\expr_i}}{\fapp{\rfun_{i+1}}{\expr_{i+1}}}$
is a $\Renv,\Penv$-step, for each $i$.
Our termination proof hinges upon the following
key result: inhabited environments only have
\emph{finite} logical traces.
We prove this result by contradiction.
Specifically, we show by Lemma~\ref{lem:step-reductions}
that an infinite $(\tosmt{\RRenv},\tosmt{\env})$-trace
combined with the fact that \env is inhabited
yields \textit{at least one}
\emph{infinite concrete trace},
which contradicts the totality assumption.
Hence, all the $(\tosmt{\RRenv},\tosmt{\env})$
logical traces must be finite.

\begin{theorem}[\textbf{Finite-Trace}]\label{thm:finite-trace}
If 
   \ \env is inhabited,
then every $(\tosmt{\RRenv},\tosmt{\env})$-trace is finite.
\end{theorem}

\mypara{Ascending Chains and Traces}
If unfolding $\Renv, \Penv$ yields
an infinite chain $\Penv_0 \subset \ldots \subset \Penv_n \ldots$,
then $\Renv, \Penv$ has an infinite logical trace.
We construct the trace
by selecting at level $i$
(\ie in $\Penv_i$)
an application term
$\fapp{\rfun_i}{\expr_i}$
that was created by
unfolding an application
term at level $i-1$
(\ie in $\Penv_{i-1}$).

\begin{lemma} [\textbf{Ascending Chains}] \label{lem:chain}
Let $\Penv_i \doteq \binst{\Renv}{\Penv}{i}$.
If there exists an (infinite) ascending chain
$\Penv_0 \subset \ldots \subset \Penv_n \ldots$,
then there exists an (infinite) logical trace
$\fapp{\rfun_0}{\expr_0},\ldots,\fapp{\rfun_n}{\expr_n},\ldots$.
\end{lemma}

\mypara{Logical Evaluation Terminates}
Finally, we prove that the proof search
procedure $\pbesym$ terminates.
If $\pbesym$ loops forever,
there must be an infinite
strictly ascending chain of
unfoldings $\Penv_i$, and
hence, by Lemma~\ref{lem:chain},
an infinite logical trace, which,
by Theorem~\ref{thm:finite-trace},
is impossible.
\begin{theorem}[\textbf{Termination}] \label{thm:termination}
  If $\env$ is transparent,
  then $\pbe{\tosmt{\RRenv}}{\tosmt{\env}}{\pred}$ terminates.
\end{theorem}


\section{Evaluation}\label{sec:evaluation}

\newcommand\numbertodo[1]{{\color{red}{#1}}}
\def\exectotal{1176}
\def\spectotal{1279}
\def\lhtimetotal{148.76}
\def\lhsmttotal{1626\xspace}
\def\lhproofstotal{1524\xspace}
\def\pbetimetotal{150.88}
\def\pbesmttotal{4068\xspace}
\def\pbeproofstotal{638\xspace}

\def\lhtotaltotal{843}
\def\pbetotaltotal{306}

\newcolumntype{?}{!{\vrule width 1pt}}
\setlength{\tabcolsep}{3pt}

\begin{small}
\begin{table*}

\begin{center}
  \begin{tabular}{ ? l  ? r|r ? r|r|r ? r|r|r? }
    \hline
        \multirow{2}{*}{\textbf{Benchmark}}  & \multicolumn{2}{c?}{\textbf{Common}} & \multicolumn{3}{c?}{\textbf{Without PLE Search}} & \multicolumn{3}{c?}{\textbf{With PLE Search}} \\\cline{2-9}
          & Impl~(l) & Spec~(l)
          & Proof~(l)& Time~(s) & SMT~(q)
          & Proof~(l) & Time~(s) & SMT~(q) \\
    \hline
     \multicolumn{9}{?l?}{\textbf{Arithmetic}}\\
    \hline
    Fibonacci
    & 7 & 10
    & 38& 2.74 & 129
    & 16& 1.92 & 79
    \\
    Ackermann
    & 20 & 73
    & 196& 5.40        & 566
    & 119& 13.80       & 846
    \\
    \hline
    \multicolumn{9}{?l?}{\textbf{Class Laws}  Fig~\ref{fig:laws}}\\
    \hline
    Monoid
    & 33 & 50
    & 109& 4.47 & 34
    & 33& 4.22 & 209
    \\
    Functor
    & 48 & 44
    & 93& 4.97 & 26
    & 14& 3.68 & 68
    \\
    Applicative
    & 62 & 110
    & 241& 12.00 & 69
    & 74& 10.00 & 1090
    \\
    Monad
    & 63 & 42
    & 122& 5.39           & 49
    & 39& 4.89           & 250
    \\
    \hline
     \multicolumn{9}{?l?}{\textbf{Higher-Order Properties}}\\
    \hline
    Logical Properties 
    & 0 & 20 
    & 33& 2.71 & 32
    & 33& 2.74 & 32
    \\
    \hline
    Fold Universal
    & 10 & 44
    & 43& 2.17 & 24
    & 14& 1.46 & 48
    \\
    \hline
     \multicolumn{9}{?l?}{\textbf{Functional Correctness}}\\
    \hline
    SAT-solver
    & 92 & 34
    & 0& 50.00 & 50
    & 0& 50.00 & 50
    \\
    Unification
    & 51 & 60
    & 85& 4.77 & 195
    & 21& 5.64 & 422
    \\
    \hline
     \multicolumn{9}{?l?}{\textbf{Deterministic Parallelism}}\\
    \hline
    Conc. Sets
    & 597 & 329
    & 339& 40.10  & 339
    & 229& 40.70  & 861
    \\
    $n$-body
    & 163 & 251
    & 101& 7.41 & 61
    & 21& 6.27 & 61
    \\
    Par. Reducers
    & 30 & 212
    & 124& 6.63 & 52
    & 25& 5.56 & 52
    \\
     \hline
    \textbf{Total}
    & \exectotal   & \spectotal
    & \lhproofstotal& \lhtimetotal & \lhsmttotal
    & \pbeproofstotal& \pbetimetotal & \pbesmttotal
    \\
    \hline
	\end{tabular}
\end{center}
  \caption{We report verification \textbf{Time}
           (in seconds, on a 2.3GHz Intel{\textregistered} Xeon{\textregistered}
           CPU E5-2699 v3 with 18 physical cores and 64GiB RAM.),
           the number of \textbf{SMT} queries and size of
           \textbf{Proofs} (in lines).
           The \textbf{Common} columns show sizes of common
           \textbf{Implementations} and \textbf{Specifications} .
           We separately consider proofs \textbf{Without} and
           \textbf{With PLE Search}.
           %
         }
\label{table:evaluation}
\end{table*}
\end{small}

We have implemented reflection and \pbesym
in \toolname~\cite{Vazou14}.
Table~\ref{table:evaluation} summarizes our
evaluation which aims to determine
(1)~the kinds of programs and properties that can be verified,
(2)~how \pbesym simplifies \emph{writing} proofs, and
(3)~how \pbesym affects the verification time.

\mypara{Benchmarks}
%
%
We summarize our benchmarks below,
see \citep{appendix} for details.
\begin{itemize}[leftmargin=*]
\item\textbf{Arithmetic}
We proved arithmetic properties for the textbook
Fibonacci function (\cf~\S~\ref{sec:overview})
and the 12 properties of the Ackermann function
from~\cite{ackermann}.
\item\textbf{Class Laws}
We proved the monoid laws for the "Peano",
"Maybe" and "List" data types and the Functor,
Applicative, and Monad laws, summarized
in~Figure~\ref{fig:laws}, for the "Maybe",
"List" and "Identity" monads.
\item\textbf{Higher Order Properties}
We used natural deduction to prove textbook
logical properties as in~\S~\ref{sec:higher-order}.
%
%
We combined natural deduction principles with
\pbesym-search to
prove universality of right-folds, as
described in~\citep{foldrHutton} and
formalized in \agda~\citep{agdaequational}.
\item\textbf{Functional Correctness}
We proved correctness of a SAT solver
and a unification algorithm as implemented
in Zombie~\citep{Zombie}.
We proved that the SAT solver takes as
input a formula "f" and either returns "Nothing"
or an assignment that satisfies "f",
by reflecting the notion of satisfaction.
Then, we proved that if the
unification "unify s t" of two terms
"s" and "t" returns a substitution
"su", then applying "su" to "s" and "t"
yields identical terms.
Note that, while the unification function can itself
diverge, and hence cannot be reflected, our method
allows terminating and diverging functions
to soundly coexist.

\item\textbf{Deterministic Parallelism}
Retrofitting verification
onto an existing language with a mature parallel run-time
allows us to create three deterministic
parallelism libraries that, for the first time,
verify implicit assumptions about associativity
and ordering that are critical for determinism 
(\cf~\cite{appendix} for extended description).
First, we proved that the \emph{ordering laws} hold for keys
inserted into LVar-style concurrent sets \cite{kuper2014freeze}.
Second, we used "monad-par" \cite{monad-par} to implement an $n$-body
simulation, whose correctness relied upon proving that a triple of "Real"
(implementing) 3-d acceleration was a "Monoid".
Third, we built a DPJ-style \cite{DPJ} parallel-reducers library
whose correctness relied upon verifying that the reduced
arguments form a "CommutativeMonoid", and which was the basis of
a parallel array sum.
\end{itemize}

\mypara{Proof Effort}
We split the total lines of code of our benchmarks
into three categories:
\textbf{Spec} represents the refinement types that
encode theorems, lemmas, and function \emph{specifications};
\textbf{Impl} represents the rest of the Haskell
code that defines executable functions;
\textbf{Proofs} represent the sizes of the Haskell proof terms
(\ie functions returning \typp).
Reflection and \pbesym are optionally
enabled using pragmas; the latter is enabled
either for a whole file/module or per
top-level function.

\mypara{Runtime Overhead}
Proof terms have \emph{no} runtime overhead
as they will never be evaluated.
When verification of an executable term
depends on theorems, we use the below 
"withTheorem" function
\begin{mcode}
    withTheorem :: x:a -> $\typp$ -> { v:a | v = x }
    withTheorem x _ = x  
\end{mcode}
that inserts the proof argument into the static 
verification environment, relying upon laziness
to not actually evaluate the proof.
For example, when verification depends on the associativity 
of append on the lists "xs", "ys", and "zs", 
the invocation "withTheorem xs (app_assoc xs ys zs)"
extends the (static) SMT verification environment
with the instantiation of the associativity theorem of 
Figure~\ref{fig:list-laws}. 
This invocation adds no runtime overhead, since 
even though "app_assoc xs ys zs" is an expensive 
recursive function, it will never actually get evaluated. 
To ensure that proof terms are not evaluated in runtime, 
without using laziness, one can add one rewrite rule for 
each proof term that replaces the term with unit.
For example, the rewrite rule for "app_assoc" is 
\begin{code}
    RULES "assoc/runtime"  forall xs ys zs. app_assoc xs ys zs = () 
\end{code}
Such rules are sound, since each proof term is total, 
thus provably reduces to unit.

\mypara{Results}
The highlights of our evaluation are the following.
(1)~Reflection allows for the specification and verification
    of a wide variety of important properties of programs.
(2)~\pbesym drastically reduces the proof effort: by a factor
    of $2-5\times$ --- shrinking the total lines of
    proof from \lhproofstotal to \pbeproofstotal ---
    making it quite modest, about the size of
    the specifications of the theorems.
    Since \pbesym searches for equational
    properties, there are some proofs, that
    rarely occur in practice, that
    \pbesym cannot simplify, \eg the
    logical properties
    from~\S~\ref{sec:higher-order}.
(3)~\pbesym does not impose a performance penalty: even though
    proof search can make an order of magnitude many more SMT
    queries --- increasing the total SMT queries from
    \lhsmttotal without \pbesym to \pbesmttotal with \pbesym---
    most of these queries are simple and it is typically
    \emph{faster} to type-check the compact proofs enabled by
    \pbesym than it is to type-check the $2-5\times$ longer
    explicit proofs written by a human.

\begin{figure}[t!]
\begin{center}
\small
\begin{tabular}{rlrl}
\toprule

\multicolumn{2}{c}{\textbf{Monoid} (for Peano, Maybe, List)}
&
\multicolumn{2}{c}{\textbf{Functor} (for Maybe, List, Id)} \\

{Left Id.}  & $\emempty\ x\ \emappend\  \equiv x$
&
{Id.}    & $\efmap\ \eid\ xs \equiv \eid\ xs$ \\

{Right Id.} & $x\ \emappend\ \emempty \equiv x$
&
{Distr.} & $\efmap\ (g\ecompose\ h)\ xs \equiv (\efmap\ g\ \ecompose\ \efmap\ h)\ xs$\\

{Assoc.}     & $(x\ \emappend\ y)\ \emappend\ z \equiv x\ \emappend\ (y\ \emappend\ z)$
&
& \\

\midrule

\multicolumn{2}{c}{\textbf{Applicative} (for Maybe, List, Id)}
&
\multicolumn{2}{c}{\textbf{Monad} (for Maybe, List, Id)} \\

{Id.}    & $\epure \eid \eseq\ v \equiv v$
&
{Left Id.}   & $\ereturn\ a \ebind f \equiv f\ a$ \\

{Comp.}  & $\epure (\ecompose) \eseq u \eseq v \eseq w \equiv u \eseq (v \eseq w)$
&
{Right Id.}  & $m \ebind \ereturn \equiv m$ \\

{Hom.}   & $\epure\ f\ \eseq\ \epure\ x \equiv \epure\ (f\ x)$
&
{Assoc.}         & $(m\ebind f) \ebind g \equiv m\ebind (\lambda x \rightarrow f\ x \ebind g)$\\

{Inter.} & $u\ \eseq\ \epure\ y \equiv \epure\ (\$\ y) \ \eseq \ u$
&
& \\

\midrule

\multicolumn{2}{c}{\textbf{Ord} (for Int, Double, Either, $(,)$)}
&
\multicolumn{2}{c}{\textbf{Commutative Monoid} (for Int, Double, $(,)$)} \\

{Refl.}    & $x \leq x$
&
{Comm.}   & $x\ \emappend\ y \equiv y\ \emappend\ x$ \\

{Antisym.} & $x \leq y \land y \leq x \implies x \equiv y$
\\
{Trans.}   & $x \leq y \land y \leq z \implies x \leq z$
&
& \emph{(including Monoid laws)}
\\
{Total.}   & $x \leq y \lor y \leq x$ \\
\bottomrule
\end{tabular}
\end{center}
\caption{Summary of Verified Typeclass Laws.}
\label{fig:laws}

\end{figure}

\section{Related Work}\label{sec:related}

\mypara{SMT-Based Verification}
SMT-solvers have been extensively
used to automate program verification
via Floyd-Hoare logics~\cite{Nelson81}.
\toolfont{Leon} introduces an
SMT-based algorithm that is
complete for catamorphisms
(folds) over ADTs~\cite{Suter2010}
and a semi-decision procedure that 
is guaranteed to find satisfying assignments 
(models) for queries over arbitrary recursive
functions, if they exist~\cite{Suter2011}.
Our work is inspired by \dafny's Verified
Calculations~\citep{LeinoPolikarpova16} 
but differs in
(1)~our use of reflection instead of axiomatization,
(2)~our use of refinements to compose proofs, and
(3)~our use of \pbesym to automate reasoning
about user-defined functions.
\dafny (and \fstar~\citep{fstar})
encode user-functions as axioms and use
a fixed fuel to instantiate functions
upto some fixed unfolding
depth~\cite{Amin2014ComputingWA}.
While the fuel-based approach is incomplete,
even for equational or calculational reasoning,
it may, although rare in practice, quickly
time out after a fixed, small number of
instantiations rather than perform an
exhaustive proof search like \pbesym.
Nevertheless, \pbesym demonstrates that it
is possible to develop complete and practical
algorithms for reasoning about user-defined
functions.

\mypara{Proving Equational Properties}
Several authors have proposed tools for proving
(equational) properties of (functional) programs.
Systems of~\citet{sousa16} and~\citet{KobayashiRelational15}
extend classical safety verification algorithms,
respectively based on Floyd-Hoare logic and refinement types,
to the setting of relational or $k$-safety properties
that are assertions over $k$-traces of a program.
Thus, these methods can automatically prove that
certain functions are associative, commutative \etc
but are restricted to first-order properties and
are not programmer-extensible.
Zeno~\citep{ZENO} generates proofs by term
rewriting and Halo~\citep{HALO} uses an axiomatic
encoding to verify contracts.
Both the above are automatic, but unpredictable and not
programmer-extensible, hence, have been limited to
far simpler properties than the ones checked here.
HERMIT~\citep{Farmer15} proves equalities by rewriting
the GHC core language, guided by user specified scripts.
Our proofs are Haskell programs, SMT solvers
automate reasoning, and importantly, we
connect the validity of proofs with the
semantics of the programs.

\mypara{Dependent Types in Programming}
Integration of dependent types into Haskell
has been a long standing goal~\citep{EisenbergS14}
that dates back to Cayenne~\citep{cayenne},
a Haskell-like, fully dependent type
language with undecidable type checking.
%
%
Our approach differs significantly in that
reflection and \pbesym use SMT-solvers
to drastically simplify proofs over decidable
theories.
Zombie~\citep{Sjoberg2015} investigates
the design of a dependently typed language where
SMT-style congruence closure is used to reason
about the equality of terms.
However, Zombie explicitly eschews type-level
computation, as the authors write
``equalities that follow from $\beta$-reduction''
are ``incompatible with congruence closure''.
Due to this incompleteness, the programmer must
use explicit @join@ terms to indicate where
normalization should be triggered, even
so, equality checking is based on fuel,
hence, is incomplete.
%

%

\mypara{Theorem Provers}
Reflection shows how to retrofit deep
specification and verification in the
style of \agda~\citep{agda},
\coq~\citep{coq-book} and
\isabelle~\cite{NPW2002} into existing
languages via refinement typing and
\pbesym shows how type-level computation
can be made compatible with SMT solvers'
native theory reasoning yielding a
powerful new way to automate
proofs~(\S~\ref{sec:overview:lists}).
An extensive comparison~\cite{Vazou17}
between our approach and mature theorem 
provers like \coq, \agda, and \isabelle
reveals that these provers 
have two clear advantages
over our approach:
they emit certificates,
so they rely on a small
trusted computing base,
and they have decades-worth
of tactics, libraries,
and proof scripts that
enable large scale proof
engineering.
Some tactics even enable
embedding of SMT-based
proof search heuristics,
\eg \toolfont{Sledgehammer}~\cite{sledgehammer},
that is widely used in
\isabelle.
However, this search
does not have the
completeness
guarantees of \pbesym.
The issue of extracting checkable
certificates from SMT solvers is well
understood~\cite{Necula97,chen-pldi-2010}
and easy to extend to our setting.
However, the question of extending SMT-based
verifiers with tactics and scriptable proof
search, and more generally, incorporating
\emph{interactivity} in the style of
proof-assistants, perhaps enhanced by
proof-completion hints, remains an
interesting direction for future work.

\section{Conclusions and Future Work}

Our results identify a new design for
deductive verifiers wherein:
(1)~via Refinement Reflection, 
we can encode natural deduction 
proofs as SMT-checkable refinement 
typed programs and
(2)~via Proof by Logical Evaluation
we can combine the complementary
strengths of SMT- (\ie, decision procedures)
and TT- based approaches (\ie, type-level computation)
to obtain completeness guarantees 
when verifying properties 
of user-defined functions.
However, the increased automation 
of SMT and proof-search can 
sometimes make it harder for 
a user to debug \emph{failed} proofs.
In future work, it would be 
interesting to investigate how 
to add interactivity to SMT based 
verifiers, in the form of tactics 
and scripts or algorithms for 
synthesizing proof hints, and to 
design new ways to explain and 
fix refinement type errors.

\begin{acks} 
We thank the anonymous referees,
our shepherd Koen Claessen, and 
Rustan Leino for their feedback 
on earlier versions of this paper.
Special thanks to George Karachalias 
and Valentin Robert for their input 
on proof combinators, 
and to Joachim Breitner and Eric Seidel 
for their suggestion on run-time proof 
elimination.
%
%
This work was supported by the 
EPSRC programme grant EP/K034413/1,
the \grantsponsor{GS100000001}{National Science
Foundation}{http://dx.doi.org/10.13039/100000001}
under Grant Numbers
\grantnum{GS100000001}{CCF-1618756}, 
\grantnum{GS100000001}{CNS-1518765},
\grantnum{GS100000001}{CCF-1518844},
\grantnum{GS100000001}{CCF-1422471},
\grantnum{GS100000001}{CCF-1223850}, and
\grantnum{GS100000001}{CCF-1218344},
and a generous gift from Microsoft Research.
\end{acks}

\bibliography{sw}
\newpage
\appendix

\section{Proof of MapFusion without \pbesym} \label{sec:map-fusion-full}

\begin{mcode}
  map_fusion f g []
  =  (map f . map g) []
  =.  map f (map g [])
  =.  map f []
  =.  map (f . g) []
  **  QED
map_fusion f g (C x xs)
  =  map (f . g) (x : xs)
  =. ((f . g) x) : (map (f . g) xs)
  =. ((f . g) x) : (((map f) . (map g)) xs)
     ? map_fusion f g xs
  =. ((f . g) x) : (map f (map g xs))
  =. (f (g x)) : (map f (map g xs))
  =. map f ((g x) : (map g xs))
  =. map f ((g x) : (map g xs))
  =. map f (map g (x : xs))
  =. map f ((map g) (x : xs))
  =. ((map f) . (map g)) (x : xs)
  ** QED
\end{mcode}

\section{Proofs for \pbesym}

Proofs for \S~\ref{sec:pbe}

\begin{proof} (Of lemma~\ref{lem:uncovering})
We prove the result by induction on $n$.

\smallskip{\emph{Case $n=0$:}}
Immediate as $\expr \equiv \expr'$.

\smallskip{\emph{Case $n=k+1$:}}
Consider any $\expr'$ such that $\smtIsValid{\Penv_{k+1}}{\vv = \expr'}$.
By definition $\Penv_{k+1} = \uinst{\Renv}{\Penv_k}$, hence
$\smtIsValid{\uinst{\Renv}{\Penv_k}, v = \expr}{\vv = \expr'}$.
Consider any $\issubterm{\fapp{\rfun}{t'}}{\expr'}$;
Lemma~\ref{lem:congruence} completes the proof.
\end{proof}

\begin{proof} (Of lemma~\ref{lem:congruence})
Proof by induction on the structure of $\expr'$:
Case $\expr' \equiv x$:
  There are no subterms, hence immediate.
Case $\expr' \equiv c$:
  There are no subterms, hence immediate.
Case $\expr' \equiv \fapp{\rfun}{t'}$:
  Consider the \emph{last link} in $\Penv$
  connecting the equivalence class
  of $\vv$ (and $\expr$) to $\expr'$.
  Suppose the last link is a \emph{congruence link}
  of the form $\expr = \expr'$ where
  $\expr \equiv \fapp{\rfun}{t}$ and
  $\smtIsValid{\Penv}{t = t'}$.
  Then $\issubterm{\fapp{\rfun}{t}}{\Penv, \expr}$
  and we are done.
  Suppose instead, the last link is an \emph{equality link}
  in $\Penv$ of the form $z = \fapp{\rfun}{t'}$.
  In this case, $\issubterm{\fapp{\rfun}{t'}}{\Penv}$
  and again, we are done.
\end{proof}

\begin{proof} (Of lemma~\ref{lem:link})
\begin{align}
\intertext{Let $G_k \doteq \binst{\Renv}{\Penv}{k}$. Let us assume that}
  & \smtIsValid{\binst{\Renv}{\Penv, \vv = \expr}{n}}{\vv = \expr'}
  \label{eq:link:1} \\
\intertext{Consider any instance}
  & \predb \equiv \fapp{\rfun}{t'} = \SUBSTS{\body_i}{x}{t'} \ \mbox{in}\ \uinst{\Renv}{\Penv \wedge \vv = \expr'}
  \label{eq:link:2} \\
\intertext{By the definition of \uinstsym, we have}
  & \issubterm{\fapp{\rfun}{t'}}{\Penv \wedge \vv = \expr'}
    \ \mbox{such that}\ \smtIsValid{\Penv}{\SUBSTS{\pred_i}{x}{t'}}
    \label{eq:link:3} \\
\intertext{By (\ref{eq:link:1}) and Lemma~\ref{lem:uncovering}
  there exists \issubterm{\fapp{\rfun}{t}}{\Penv_n} such that}
  & \smtIsValid{\Penv_n}{t = t'} \notag \\
\intertext{As $\Penv \subseteq \Penv_n$ and (\ref{eq:link:3}), by congruence}
  & \smtIsValid{\Penv_n}{\SUBSTS{\pred_i}{x}{t}} \notag \\
\intertext{Hence, the instance}
 & \fapp{\rfun}{t} = \SUBSTS{\body_i}{x}{t} \ \mbox{is in} \Penv_{n+1} \notag \\
\intertext{That is}
 & \smtIsValid{\Penv_{n+1}}{t = t' \wedge \fapp{\rfun}{t} = \SUBST{\body_i}{x}{t}} \notag \\
\intertext{And so by congruence closure}
 & \smtIsValid{\Penv_{n+1}}{\predb} \notag \\
\intertext{As the above holds for every instance, we have}
 & \smtIsValid{\Penv_{n+1}}{\uinst{\Renv}{\Penv \wedge \vv = \expr'}} \notag
\end{align}
\end{proof}

\begin{proof} (Of lemma~\ref{lem:bounded-completeness})
The proof follows by induction on the structure
of $\eqpf{\Renv}{\Penv}{\expr}{\expr'}$.

\smallskip \emph{Base Case \rulename{Eq-Refl}:}
Follows immediately as $\expr \equiv \expr'$.

\smallskip \emph{Inductive Case \rulename{Eq-Trans}}
\begin{align}
\intertext{In this case, there exists $\expr''$ such that}
  & \eqpf{\Renv}{\Penv}{\expr}{\expr''}
    \label{eq:bc:1} \\
  & \smtIsValid{\uinst{\Renv}{\Penv \wedge \vv = \expr''}}{\vv = \expr'}
    \label{eq:bc:2} \\
\intertext{By the induction hypothesis (\ref{eq:bc:1}) implies there exists $0 \leq n$ such that}
  & \smtIsValid{\binst{\Renv}{\Penv \wedge \vv = \expr}{n}}{\vv = \expr''}
    \notag \\
\intertext{By Lemma~\ref{lem:link} we have}
  & \smtIsValid{\binst{\Renv}{\Penv, \vv = \expr}{n+1}}{\uinst{\Renv}{\Penv \wedge \vv = \expr''}}
    \notag \\
\intertext{Thus, by (\ref{eq:bc:2}) and modus ponens we get}
  & \smtIsValid{\binst{\Renv}{\Penv, \vv = \expr}{n+1}}{\vv = \expr'} \notag
\end{align}
\end{proof}

\begin{proof} (Of Lemma~\ref{lem:fixpoint})
Let $\Penv' = \Penv \wedge \vv = \expr$.

\emph{Case $\Rightarrow$:}
Assume that $\pbe{\Renv}{\Penv}{\expr = \expr'}$.
That is, at some iteration $i$ we have
$\smtIsValid{\Penv_i}{\vv = \expr'}$,
\ie by (\ref{eq:fix:1}) we have
$\smtIsValid{\binst{\Renv}{\Penv'}{i}}{\vv = \expr'}$.

\emph{Case $\Leftarrow$:}
Pick the smallest $n$ such that
$\smtIsValid{\binst{\Renv}{\Penv'}{n}}{\vv = \expr'}$.
Using Lemmas~\ref{lem:bounded:mono}~and~\ref{lem:bounded:fix}
we can then show that forall $0 \leq k < n$, we have
$$
\binst{\Renv}{\Penv'}{k} \subset \binst{\Renv}{\Penv'}{k+1}
$$
and
$$\binst{\Renv}{\Penv'}{k} \not \vdash {\vv = \expr'}$$
Hence, after $n$ iterations of the recursive loop,
$\pbe{\Renv}{\Penv}{\expr = \expr'}$, returns $\ttrue$.
\end{proof}

\mypara{Steps and Values}
Next, we show that if
$\symstep{\fapp{\rfun}{y}}{\expr'}$
is a logical step under an $\env$
that is inhabited by $\store$
then \fapp{\rfun}{y} reduces
to a value under $\store$.
The proof follows by observing
that if $\env$ is inhabited by
$\store$, and a particular step
is possible, then the guard
corresponding to that step
must also be true under
$\store$ and hence, by totality,
the function must reduce to
a value under the given store.


\begin{lemma}[\textbf{Step-Value}]\label{lem:step-value}
If   $\store \in \embed{\env}$
     and
     $\symstep{\fapp{\rfun}{y}}{\expr'}$
	   is a $\tosmt{\RRenv}, \tosmt{\env}$ step
then \rjtranstore{\RRenv}{\store}{\fapp{\rfun}{y}}{\val}.
\end{lemma}

\begin{proof} (Of Lemma~\ref{lem:step-value})
\begin{align}
\bcoz{Assume that}
  & \store \in \embed{\env} \label{eq:stepval:1} \\
\bcoz{Let}
  & \store^*  \doteq \store \mapsingle{\params{x}}{\apply{\store}{\params{y}}} \label{eq:stepval:2}\\
  & \rlookup{\Renv}{\rfun} \doteq \rdef{x}{\pred}{\body} \label{eq:stepval:3} \\
\intertext{As $\symstep{\fapp{\rfun}{y}}{\expr'}$ is a $\tosmt{\RRenv}{\env}$ step,
for some $i$, \tosmt{\RRenv}-instance $Q$ we have}
  & \smtIsValid{\tosmt{\env} \wedge Q}{\SUBSTS{\pred_i}{x}{y}} \notag \\
\bcoz{Hence, by (\ref{eq:stepval:1}) and Lemma~\ref{lem:smt-approx}}
  & \rsatisfies{\RRenv}{\store}{\SUBSTS{\pred_i}{x}{y}} \label{eq:stepval:4}\\
\bcoz{As}
  & \apply{\store^*}{\pred_i} \equiv \apply{\store}{\SUBSTS{\pred_i}{x}{y}} \\
\bcoz{The fact (\ref{eq:stepval:4}) yields}
  & \rsatisfies{\RRenv}{\store^*}{\pred_i} \notag \\
\bcoz{By the Totality Assumption~\ref{asm:total}}
  & \rjtranstore{\RRenv}{\store^*}{\fapp{\rfun}{x}}{\val} \notag \\
\bcoz{That is}
  & \rjtranstore{\RRenv}{\store}{\fapp{\rfun}{y}}{\val} \\
\end{align}
\end{proof}

\mypara{Divergence}
A closed term $\expr$ \emph{diverges under}
$\RRenv$ if there is no $\val$ such that
$\rjtrans{\RRenv}{\expr}{\val}$ .

\begin{lemma}[\textbf{Divergence}] \label{lem:div} 
If $\ \forall 0 \leq i$ we have
$\rjtrans{\RRenv}{\expr_i}{\ctxapp{\ctx}{\expr_{i+1}}}$
then $\expr_0$ diverges under $\Renv$.
\end{lemma}

\begin{proof} (Of Theorem~\ref{thm:finite-trace})
\begin{align}
\bcoz{Assume that}
   & \store \in \embed{\env}
   \label{eq:fin:0}\\
\bcoz{and assume an infinite $\tosmt{\RRenv}, \tosmt{\env}$ trace:}
   & \fapp{f_0}{\expr_0}, \fapp{f_1}{\expr_1}, \ldots
   \label{eq:fin:1}\\
\bcoz{Where additionally}
   & \params{\expr_0} \equiv \params{x_0}
   \label{eq:fin:2} \\
\bcoz{Define}
   & \expr^*_i        \equiv \apply{\store}{\fapp{\rfun_i}{\expr_i}}
   \label{eq:fin:3} \\
\bcoz{By Lemma~\ref{lem:step-reductions}, for every $i \in \nats$}
   & \rjtrans{\RRenv}{\expr^*_i}{\ctxapp{\ctx_i}{\expr^*_{i+1}}}
   \notag \\
\bcoz{Hence, by Lemma~\ref{lem:div}}
   & \expr^*_0\ \mbox{diverges under $\Renv$}
   \notag \\
\bcoz{\ie, by (\ref{eq:fin:2}, \ref{eq:fin:3})}
   & \apply{\store}{\fapp{\rfun_0}{x_0}}\ \mbox{diverges under $\Renv$}
   \label{eq:fin:5} \\
\bcoz{But by (\ref{eq:fin:0}) and Lemma~\ref{lem:step-value}}
   & \rjtranstore{\RRenv}{\store}{\fapp{\rfun_0}{x_0}}{\val}
   \quad \mbox{contradicting (\ref{eq:fin:5})} \notag
\end{align}
Hence, assumption (\ref{eq:fin:1}) cannot hold, \ie all
all the $\Renv, \Penv$ symbolic traces must be finite.
\end{proof}


\begin{proof} (Of Theorem~\ref{thm:termination})
As $\Penv$ is transparent, there are two cases.
\begin{align}
\intertext{\quad \emph{Case: $\env$ is inconsistent.}}
\bcoz{By definition of inconsistency}
  & \smtIsValid{\tosmt{\env}}{\tfalse}
    \notag \\
\bcoz{Hence}
  & \smtIsValid{\tosmt{\env}}{\pred}
    \notag \\
\bcoz{That is}
  & \pbe{\tosmt{\RRenv}}{\tosmt{\env}}{\pred}\ \mbox{terminates immediately.}
    \notag \\
\intertext{\quad \emph{Case: $\env$ is inhabited.}}
\bcoz{That is, exists $\store$ s.t.}
  & \store \in \embed{\env}
    \label{eq:term:1} \\
\intertext{\quad Suppose that $\pbe{\tosmt{\RRenv}}{\tosmt{\env}}{\pred}$ does not terminate.}
\bcoz{That is, there is an infinitely increasing chain:}
  & \Penv_0 \subset \ldots \subset \Penv_n \ldots
  \label{eq:term:2} \\
\bcoz{By Lemma~\ref{lem:chain}}
  & \tosmt{\RRenv},\tosmt{\env}\ \mbox{has an infinite trace}
    \notag
\end{align}
Which, by (\ref{eq:term:1}) contradicts
Theorem~\ref{thm:finite-trace}.
Thus, (\ref{eq:term:2}) is impossible,
\ie $\pbe{\tosmt{\RRenv}}{\tosmt{\env}}{\pred}$
terminates.
\end{proof}

\section{Proof of Section: Embedding Natural Deduction with Refinement Types}

\begin{lemma}[Validity]\label{lemma:validity}
If there exists $e\in\embed{\formula}$\ then \formula is valid.
\end{lemma}
\begin{proof}
We prove the lemma by case analysis in the shape of \formula. 
\begin{itemize}
\item $\formula \equiv \proofterm{p}$.
Since the set
 $\embed{\proofterm{p}} = \{e | \evalsto{p}{\etrue} \}$
 is not empty, then  \evalsto{p}{\etrue}.
\item $\formula \equiv {\formula_1} \rightarrow {\formula_2}$. 
By assumption, there exists an expressions $f$ so that 
$\forall e_x \in \embed{\formula_1}, f\ e_x \in \embed{\formula_2}$. 
So, if there exists an expression 
$e_1 \in \embed{\formula_1}$ that makes $\formula_1$ valid 
then $f\ e_1$ makes $\formula_2$ valid. 
\item $\formula \equiv {\formula} \rightarrow {\efalse}$. 
By assumption, there exists an expressions $f$ so that 
$\forall e_x \in \embed{\formula}, f\ e_x \in \embed{\proofterm{\efalse}}$. 
So, if there exists an expression 
$e_1 \in \embed{\formula}$ that makes $\formula$ valid 
then $f\ e_1$ makes $\proofterm{\efalse}$ valid, which is impossible, 
thus \formula cannot be valid. 
\item $\formula \equiv \tyand{\formula_1}{\formula_2}$. 
If there exists a total expression $e \in \embed{\formula}$
then $e$ evaluates to \exprand{e_1}{e_2}. 
\item $\formula \equiv \tyor{\formula_1}{\formula_2}$. 
If there exists a total expression $e \in \embed{\formula}$
then $e$ evaluates to either 
\exprorleft{e'} or \exprorright{e'}. 
\item $\formula \equiv \tyall{x}{\typ}{\formula}$.
By assumption, there exists an expressions $f$ so that 
$\forall e_x \in \embed{\typ}, f\ e_x \in \embed{\formula\subst{x}{e_x}}$. 
So, if there exists an expression 
$e_1 \in \embed{\typ}$
then $f\ e_1$ makes $\formula\subst{x}{e_1}$ valid. 
\item $\formula \equiv \tyex{x}{\typ}{\formula}$.
By assumption, there exists an expressions $p$ 
that evaluates to a pair $(e_x, e_y)$ 
so that $e_x\in \embed{\typ}$ and $e_y\in\embed{\formula\subst{x}{e_x}}$.
\end{itemize} 
\end{proof}

\begin{theorem}[Validity]
If \hastype{\emptyset}{\emptyset}{e}{\formula} then \formula is valid.
\end{theorem}
\begin{proof}
By direct implication of Lemma~\ref{lemma:validity} and soundness of \corelan (Theorem~\ref{thm:safety}).
\end{proof}

\renewcommand\hastype[3]{\ensuremath{#1 \vdash #2 : #3}}
\renewcommand\issubtype[3]{\ensuremath{#1 \vdash #2 \preceq #3}}
\renewcommand\sub{\ensuremath{\theta}}
\renewcommand\aissubtype[3]{\ensuremath{#1 \vdash_{S} #2 \preceq #3}}
\renewcommand\ahastype[3]{\ensuremath{#1 \vdash_{S} #2 : #3}}
\renewcommand\tref[3]{\ensuremath{\{ \tbind{#1}{#2} \mid #3 \}}}
\renewcommand\tfun[3]{\ensuremath{\tbind{#1}{#2} \rightarrow #3}}

\section{Refinement Reflection: \corelan: Extended Version with Proofs}
\label{sec:appendix:formalism}
\label{sec:appendix:types-reflection}

Our first step towards formalizing refinement
reflection is a core calculus \corelan with an
\emph{undecidable} type system based on
denotational semantics.
We show how the soundness of the type system
allows us to \emph{prove theorems} using \corelan.

%

\subsection{Syntax}
\begin{figure}[t!]
\centering
$$
\begin{array}{rrcl}
\emphbf{Operators} \quad
  & \odot
  & ::= & = \spmid  <
\\[0.03in]

\emphbf{Constants} \quad
  & c
  & ::=
  & \land \spmid \lnot \spmid \odot \spmid +,-,\dots  \\
  && \spmid & \etrue\spmid \efalse \spmid 0, 1,-1, \dots
\\[0.03in]

\emphbf{Values} \quad
  & w & ::=&  c
             \spmid \efun{x}{\typ}{e} \spmid D\ \overline{w}
\\[0.03in]

\emphbf{Expressions} \quad
  & e & ::=    & w \spmid x \spmid \eapp{e}{e}      \\
  &   & \spmid & \ecase{x}{e}{\dc}{\overline{x}}{e}
\\[0.03in]

\emphbf{Binders} \quad
  & \bd & ::= & e \spmid \eletb{x}{\gtyp}{\bd}{\bd}
\\[0.03in]

\emphbf{Program} \quad
  & \prog & ::= & \bd \spmid \erefb{x}{\gtyp}{e}{\prog}
\\[0.03in]


\emphbf{Basic Types} \quad
  & \btyp
  & ::=
  & \tint \spmid \tbool \spmid T
\\[0.03in]

\emphbf{Refined Types} \quad
  & \typ
  & ::=      & \tref{v}{\btyp}{\reft} \spmid \tfun{x}{\typ}{\typ}
\\[0.05in]
\end{array}
$$
\caption{\textbf{Syntax of \corelan}}
\label{fig:appendix:syntax}
\end{figure}

\begin{figure}[t!]
\centering

\emphbf{Contexts}\hfill{{$\fbox{\textit{C}}$}}
$$
\begin{array}{rcl}
  C & ::=    & \bullet \\
    & \spmid & C\ e \spmid c\ C \spmid D\ \overline{e}\ C\ \overline{e}\\
    & \spmid & \ecase{y}{C}{\dc_i}{\overline{x}}{e_i}
  \\[0.03in]
\end{array}
$$

\emphbf{Reductions}\hfill{{$\fbox{\goesto{\prog}{\prog'}}$}}
$$
\begin{array}{rcl}
C[\prog]
  & \hookrightarrow
  & C[\prog'],\quad \text{if}\ \goesto{\prog}{\prog'}
  \\
{c\ v}
  & \hookrightarrow
  & {\ceval{c}{v}}
  \\
{({\efun{x}{\typ}{e})}\ {e'}}
  & \hookrightarrow
  & {\SUBST{e}{x}{e'}}
  \\
{\ecase{y}{\dc_j\ \overline{e}}{\dc_i}{\overline{x_i}}{e_i}}
  & \hookrightarrow
  & {\SUBST{\SUBST{e_j}{y}{\dc_j\ \overline{e}}}{\overline{x_i}}{\overline{e}}}
  \\
{\erefb{x}{\gtyp}{e}{\prog}}
  & \hookrightarrow
  & {\SUBST{\prog}{x}{\efix{(\efun{x}{\gtyp}{e})}}}
  \\
{\eletb{x}{\gtyp}{\bd_x}{\bd}}
  & \hookrightarrow
  & {\SUBST{\bd}{x}{\efix{(\efun{x}{\gtyp}{\bd_x})}}}
  \\
{\efix{\prog}}
  & \hookrightarrow
  & {\prog\ (\efix{\prog})}
\end{array}
$$
\caption{\textbf{Operational Semantics of \corelan}}
\label{fig:semantics}
\end{figure}

Figure~\ref{fig:appendix:syntax} summarizes the syntax of \corelan,
which is essentially the calculus \undeclang~\cite{Vazou14}
with explicit recursion and a special $\erefname$ binding form
to denote terms that are reflected into the refinement logic.
In \corelan refinements $r$ are arbitrary expressions $e$
(hence $r ::= e$ in Figure~\ref{fig:appendix:syntax}).
This choice allows us to prove preservation and progress,
but renders typechecking undecidable.
In \S~\ref{sec:appendix:algorithmic} we will see how to recover
decidability by soundly approximating refinements.

The syntactic elements of \corelan are layered into
primitive constants, values, expressions, binders
and programs.

\mypara{Constants}
The primitive constants of \corelan
include all the primitive logical
operators $\op$, here, the set $\{ =, <\}$.
Moreover, they include the
primitive booleans $\etrue$, $\efalse$,
integers $\mathtt{-1}, \mathtt{0}$, $\mathtt{1}$, \etc,
and logical operators $\mathtt{\land}$, $\mathtt{\lor}$, $\mathtt{\lnot}$, \etc.

\mypara{Data Constructors}
We encode data constructors as special constants.
For example the data type \tintlist, which represents
finite lists of integers, has two data constructors: $\dnull$ (``nil'')
and $\dcons$ (``cons'').

%

\mypara{Values \& Expressions}
The values of \corelan include
constants, $\lambda$-abstractions
$\efun{x}{\typ}{e}$, and fully
applied data constructors $D$
that wrap values.
The expressions of \corelan
include values and variables $x$,
applications $\eapp{e}{e}$, and
$\mathtt{case}$ expressions.

\mypara{Binders \& Programs}
A \emph{binder} $\bd$ is a series of possibly recursive
let definitions, followed by an expression.
A \emph{program} \prog is a series of $\erefname$
definitions, each of which names a function
that can be reflected into the refinement
logic, followed by a binder.
The stratification of programs via binders
is required so that arbitrary recursive definitions
are allowed but cannot be inserted into the logic
via refinements or reflection.
(We \emph{can} allow non-recursive $\mathtt{let}$
binders in $e$, but omit them for simplicity.)

\subsection{Operational Semantics}

Figure~\ref{fig:appendix:syntax} summarizes the small
step contextual $\beta$-reduction semantics for
\corelan.
%
%
We write \evalj{e}{e'}{j} if there exist
$e_1,\ldots,e_j$ such that $e$ is $e_1$,
$e'$ is $e_j$ and $\forall i,j, 1 \leq i < j$,
we have $\evals{e_i}{e_{i+1}}$.
We write $\evalsto{e}{e'}$ if there exists
some finite $j$ such that $\evalj{e}{e'}{j}$.
We define $\betaeq{}{}$ to be the reflexive,
symmetric, transitive closure of $\evals{}{}$.

\mypara{Constants} Application of a constant requires the
argument be reduced to a value; in a single step the
expression is reduced to the output of the primitive
constant operation.
For example, consider $=$, the primitive equality
operator on integers.
We have $\ceval{=}{n} \defeq =_n$
where $\ceval{=_n}{m}$ equals \etrue
iff $m$ is the same as $n$.
%
%
We assume that the equality operator
is defined \emph{for all} values,
and, for functions, is defined as
extensional equality.
That is, for all
$f$ and
$f'$
we have
$\evals{(f = f')}{\etrue}
  \quad \mbox{iff} \quad
  \forall v.\ \betaeq{f\ v}{f'\ v}$.
We assume source \emph{terms} only contain implementable equalities
over non-function types; the above only appears in \emph{refinements}
and allows us to state and prove facts about extensional
equality~\S~\ref{subsec:extensionality}.


\subsection{Types}

\corelan types include basic types, which are \emph{refined} with predicates,
and dependent function types.
\emph{Basic types} \btyp comprise integers, booleans, and a family of data-types
$T$ (representing lists, trees \etc.)
For example the data type \tintlist represents lists of integers.
We refine basic types with predicates (boolean valued expressions \refa) to obtain
\emph{basic refinement types} $\tref{v}{\btyp}{\refa}$.
Finally, we have dependent \emph{function types} $\tfun{x}{\typ_x}{\typ}$
where the input $x$ has the type $\typ_x$ and the output $\typ$ may
refer to the input binder $x$.
We write $\btyp$ to abbreviate $\tref{v}{\btyp}{\etrue}$,
and \tfunbasic{\typ_x}{\typ} to abbreviate \tfun{x}{\typ_x}{\typ} if
$x$ does not appear in $\typ$.
We use $r$ to refer to refinements.

\mypara{Denotations}
Each type $\typ$ \emph{denotes} a set of expressions $\interp{\typ}$,
that are defined via the dynamic semantics~\cite{Knowles10}.
Let $\shape{\typ}$ be the type we get if we erase all refinements
from $\typ$ and $\bhastype{}{e}{\shape{\typ}}$ be the
standard typing relation for the typed lambda calculus.
Then, we define the denotation of types as:
\begin{align*}
\interp{\tref{x}{\btyp}{r}} \defeq &
    \{e \mid  \bhastype{}{e}{\btyp},
              \mbox{ if } \evalsto{e}{w}
              \mbox{ then } \evalsto{r\subst{x}{w}}{\etrue} \}\\
\interp{\tfun{x}{\typ_x}{\typ}} \defeq &
    \{e \mid  \bhastype{}{e}{\shape{\tfunbasic{\typ_x}{\typ}}},
              \forall e_x \in \interp{\typ_x}.\ \eapp{e}{e_x} \in \interp{\typ\subst{x}{e_x}}
    \}
\end{align*}

\mypara{Constants}
For each constant $c$ we define its type \constty{c}
such that $c \in \interp{\constty{c}}$.
For example,
$$
\begin{array}{lcl}
\constty{3} &\doteq& \tref{v}{\tint}{v = 3}\\
\constty{+} &\doteq& \tfun{\ttx}{\tint}{\tfun{\tty}{\tint}{\tref{v}{\tint}{v = x + y}}}\\
\constty{\leq} &\doteq& \tfun{\ttx}{\tint}{\tfun{\tty}{\tint}{\tref{v}{\tbool}{v \Leftrightarrow x \leq y}}}\\
\end{array}
$$
So, by definition we get the constant typing lemma
\begin{lemma}{[Constant Typing]}\label{lemma:constants}
Every constant $c \in \interp{\constty{c}}$.
\end{lemma}
Thus, if $\constty{c} \defeq \tfun{x}{\typ_x}{\typ}$,
then for every value $w \in \interp{\typ_x}$, we require
$\ceval{c}{w} \in \interp{\typ\subst{x}{w}}$.


\subsection{Refinement Reflection}
\label{subsec:appendix:logicalannotations}
The simple, but key idea in our work is to
\emph{strengthen} the output type of functions
with a refinement that \emph{reflects} the
definition of the function in the logic.
We do this by treating each
$\erefname$-binder:
${\erefb{f}{\gtyp}{e}{\prog}}$
as a $\eletname$-binder:
${\eletb{f}{\exacttype{\gtyp}{e}}{e}{\prog}}$
during type checking (rule $\rtreflect$ in Figure~\ref{fig:typing}).

\mypara{Reflection}
We write \exacttype{\typ}{e} for the \emph{reflection}
of term $e$ into the type $\typ$,  defined by strengthening
\typ as:
$$
\begin{array}{lcl}
\exacttype{\tref{v}{\btyp}{r}}{e}
  & \defeq
  & \tref{v}{\btyp}{r \land v = e}\\
\exacttype{\tfun{x}{\typ_x}{\typ}}{\efun{y}{}{e}}
  & \defeq
  & \tfun{x}{\typ_x}{\exacttype{\typ}{e\subst{y}{x}}}
\end{array}
$$
As an example, recall from \S~\ref{sec:overview}
that the "reflect fib " strengthens the type of
"fib" with the reflected refinement "fibP".

\mypara{Consequences for Verification}
Reflection has two consequences for verification.
First, the reflected refinement is \emph{not trusted};
it is itself verified (as a valid output type)
during type checking.
Second, instead of being tethered to quantifier
instantiation heuristics or having to program
``triggers'' as in Dafny~\citep{dafny} or
\fstar~\citep{fstar}
the programmer can predictably ``unfold'' the
definition of the function during a proof simply
by ``calling'' the function, which we have found
to be a very natural way of structuring
proofs~\S~\ref{sec:evaluation}.

\subsection{Refining \& Reflecting Data Constructors with Measures}
\label{subsec:appendix:measures}
\label{subsec:list}


We assume that each data type is equipped with
a set of \emph{measures} which are \emph{unary}
functions whose (1)~domain is the data type, and
(2)~body is a single case-expression over the
datatype~\cite{Vazou14}:
$$\emeasb{f}
         {\gtyp}
         {\efun{x}{\typ}{\ecase{y}{x}{\dc_i}{\overline{z}}{e_i}}}$$
For example, "len" measures the size of an $\tintlist$:
\begin{code}
  measure len :: [Int] -> Nat
  len = \x -> case x of
                []     -> 0
                (x:xs) -> 1 + len xs
\end{code}

\mypara{Checking and Projection}
We assume the existence of measures that
\emph{check} the top-level constructor,
and \emph{project} their individual fields.
%
In \S~\ref{subsec:appendix:embedding} we show how to
use these measures to reflect functions over
datatypes.
%
%
For example, for lists, we assume the existence of measures:
\begin{code}
  isNil []      = True
  isNil (x:xs)  = False

  isCons (x:xs) = True
  isCons []     = False

  sel1 (x:xs)   = x
  sel2 (x:xs)   = xs
\end{code}

\mypara{Refining Data Constructors with Measures}
We use measures to strengthen the types
of data constructors, and we use these
strengthened types during construction
and destruction (pattern-matching).
Let:
(1)~$\dc$ be a data constructor,
   with \emph{unrefined} type
   $\tfun{\overline{x}}{\overline{\gtyp}}{T}$
(2)~the $i$-th measure definition with
   domain $T$ is:
$$\emeasb{f_i}
         {\gtyp}
         {\efun{x}{\typ}{\ecase{y}{x}{\dc}{\overline{z}}{e_{i}}}}
$$
Then, the refined type of $\dc$ is defined:
$$
\constty{\dc} \defeq
   \tfun{\overline{x}}
        {\overline{\typ}}
        {\tref{v}{T}{ \wedge_i f_i\ v = \SUBST{e_{i}}{\overline{z}}{\overline{x}}}}
$$

Thus, each data constructor's output type is refined to reflect
the definition of each of its measures.
For example, we use the measures "len", "isNil", "isCons", "sel1",
and "sel2" to strengthen the types of $\dnull$ and $\dcons$ to:
\begin{align*}
\constty{\dnull}  \defeq & \tref{v}{\tintlist}{r_{\dnull}} \\
\constty{\dcons}  \defeq & \tfun{x}{\tint}
                                   {\tfun{\mathit{xs}}
                                         {\tintlist}
                                         {\tref{v}{\tintlist}{r_\dcons}}}
\intertext{where the output refinements are}
r_{\dnull} \defeq &\ \mathtt{len}\ v = 0
             \land  \mathtt{isNil}\ v
             \land  \lnot \mathtt{isCons}\ v \\
r_{\dcons} \defeq &\ \mathtt{len}\ v = 1 + \mathtt{len}\ \mathit{xs}
             \land  \lnot \mathtt{isNil}\ v
             \land  \mathtt{isCons}\ v \\
             \land & \  \mathtt{sel1}\ v = x
             \land  \mathtt{sel2}\ v = \mathit{xs}
\end{align*}
It is easy to prove that Lemma~\ref{lemma:constants}
holds for data constructors, by construction.
For example, $\mathtt{len}\ \dnull = 0$ evaluates to $\tttrue$.


\subsection{Typing Rules}
\begin{figure}[!t]
\emphbf{Typing}\hfill{\fbox{\hastype{\env}{\prog}{\typ}}}\\
$$
\inference{
	\tbind{x}{\gtyp}\in\env
}{
	\hastype{\env}{x}{\gtyp}
}[\rtvar]
\qquad
\inference{
}{
	\hastype{\env}{c}{\constty{c}}
}[\rtconst]
$$

$$
\inference{
	\hastype{\env}{\prog}{\typ'} &
	\issubtype{\env}{\typ'}{\typ}
}{
	\hastype{\env}{\prog}{\typ}
}[\rtsub]
$$
$$
\inference{
	\hastype{\env}{e}{\tref{v}{\btyp}{\reft_r}}
}{
	\hastype{\env}{e}{\tref{v}{\btyp}{\reft_r\land v = e}}
}[\rtexact]
$$
$$
\inference{
	\hastype{\env, \tbind{x}{\typ_x}}{e}{\typ}
}{
	\hastype{\env}{\efun{x}{\typ}{e}}{\tfun{x}{\typ_x}{\typ}}
}[\rtfun]
$$

$$
\inference{
	\hastype{\env}{e_1}{(\tfun{x}{\typ_x}{\typ})} &&
	\hastype{\env}{e_2}{\typ_x}
}{
	\hastype{\env}{e_1\ e_2}{\typ}
}[\rtapp]
$$

$$
\inference
	{\hastype{\env, \tbind{x}{\gtyp_x}}{\bd_x}{\gtyp_x} &
	 \iswellformed{\env, \tbind{x}{\gtyp_x}}{\typ_x} \\
	 \hastype{\env, \tbind{x}{\gtyp_x}}{\bd}{\gtyp} &
	 \iswellformed{\env}{\typ}}
	{\hastype{\env}{\eletb{x}{\gtyp_x}{\bd_x}{\bd}}{\typ}}
	[\rtlet]
$$

$$
\inference
	{\hastype{\env}
	 				 {\eletb{f}{\exacttype{\gtyp_f}{e}}{e}{\prog}}
					 {\typ}
	}
	{\hastype{\env}
					 {\erefb{f}{\gtyp_f}{e}{\prog}}
					 {\typ}
	}[\rtreflect]
$$

$$
\inference{
	\hastype{\env}{e}{\tref{v}{T}{e_r}} & \iswellformed{\env}{\typ} \\
	& \forall i. \constty{\dc_i} = \tfunbasic{\overline{\tbind{y_j}{\typ_j}}}{\tref{v}{T}{e_{r_i}}} \\
	& \hastype{\env, \overline{\tbind{y_j}{\typ_j}}, \tbind{x}{\tref{v}{T}{e_r \land e_{r_i}}} }{e_i}{\typ}
}{
	\hastype{\env}{\ecase{x}{e}{\dc_i}{\overline{y_i}}{e_i}}{\typ}
}[\rtcase]
$$
\emphbf{Well Formedness}\hfill{\fbox{\iswellformed{\env}{\typ}}}\\

$$
\inference{
  \hastype{\env,\tbind{v}{\btyp}}{\refa}{\tbool^{\tlabel}}
}{
	\iswellformed{\env}{\tref{v}{\btyp}{\refa}}
}[\rwbase]
$$
$$
\inference{
	\iswellformed{\env}{\typ_x} &
	\iswellformed{\env,\tbind{x}{\typ_x}}{\typ}
}{
	\iswellformed{\env}{\tfun{x}{\typ_x}{\typ}}
}[\rwfun]
$$

\emphbf{Subtyping}\hfill{\fbox{\issubtype{\env}{\typ_1}{\typ_2}}}\\

$$
\inference{
\env' \defeq \env,\tbind{v}{\{\btyp^\tlabel | \refa\}} \\
\tologicshort{\env'}{\refa'}{\tbool}{\pred'}{}{}{} &
\smtvalid{\vcond{\env'}{\pred'}}
}{
	\issubtype{\env}{\tref{v}{\btyp}{\refa}}{\tref{v}{\btyp}{\refa'}}
}[\rsubbase]
$$

$$
\inference{
	\issubtype{\env}{\typ_x'}{\typ_x} &
	\issubtype{\env,\tbind{x}{\typ_x'}}{\typ}{\typ'}
}{
	\issubtype{\env}{\tfun{x}{\typ_x}{\typ}}{\tfun{x}{\typ_x'}{\typ'}}
}[\rsubfun]
$$
\caption{Typing of \corelan}
\label{fig:appendix:typing}
\figrule
\end{figure}

Next, we present the type-checking
judgments and rules of \corelan.

\mypara{Environments and Closing Substitutions}
A \emph{type environment} $\env$ is a sequence of type bindings
$\tbind{x_1}{\typ_1},\ldots,\tbind{x_n}{\typ_n}$. An environment
denotes a set of \emph{closing substitutions} $\sto$ which are
sequences of expression bindings:
$\gbind{x_1}{e_1}, \ldots, \gbind{x_n}{e_n}$ such that:
$$
\interp{\env} \defeq  \{\sto \mid \forall \tbind{x}{\typ} \in \Env.
                                    \sto(x) \in \interp{\applysub{\sto}{\typ}} \}
$$

\mypara{Judgments}
We use environments to define three kinds of
rules: Well-formedness, Subtyping,
and Typing~\citep{Knowles10,Vazou14}.
%
A judgment \iswellformed{\env}{\typ} states that
the refinement type $\typ$ is well-formed in
the environment $\env$.
Intuitively, the type $\typ$ is well-formed if all
the refinements in $\typ$ are $\tbool$-typed in $\env$.
%
A judgment \issubtype{\env}{\typ_1}{\typ_2} states
that the type $\typ_1$ is a subtype of 
$\typ_2$ in the environment $\env$.
Informally, $\typ_1$ is a subtype of $\typ_2$ if, when
the free variables of $\typ_1$ and $\typ_2$
are bound tomeasures expressions described by $\env$,
the denotation of $\typ_1$
is \emph{contained in} the denotation of $\typ_2$.
Subtyping of basic types reduces to denotational containment checking.
%
%
That is, for any closing substitution $\sto$
in the denotation of $\env$, for every expression $e$,
if $e \in \interp{\applysub{\sto}{\typ_1}}$ then
$ e \in \interp{\applysub{\sto}{\typ_2}}$.
%
A judgment \hastype{\env}{\prog}{\typ} states that
the program $\prog$ has the type $\typ$ in
the environment $\env$.
That is, when the free variables in $\prog$ are
bound to expressions described by $\env$, the
program $\prog$ will evaluate to a value
described by $\typ$.

\mypara{Rules}
All but three of the rules are standard~\cite{Knowles10,Vazou14}.
First, rule \rtreflect is used to strengthen the type of each
reflected binder with its definition, as described previously
in \S~\ref{subsec:appendix:logicalannotations}.
%
Second, rule \rtexact strengthens the expression with
a singleton type equating the value and the expression
(\ie reflecting the expression in the type).
This is a generalization of the ``selfification'' rules
from \cite{Ou2004,Knowles10}, and is required to
equate the reflected functions with their definitions.
For example, the application $(\fib\ 1)$ is typed as
${\tref{v}{\tint}{\fibdef\ v\ 1 \wedge v = \fib\ 1}}$ where
the first conjunct comes from the (reflection-strengthened)
output refinement of \fib~\S~\ref{sec:overview}, and
the second conjunct comes from rule~\rtexact.
Finally, rule \rtfix is used to type the intermediate
$\texttt{fix}$ expressions that appear, not in the
surface language but as intermediate terms in the
operational semantics.

\mypara{Soundness}
Following \undeclang~\citep{Vazou14}, we can show that
evaluation preserves typing and that typing implies
denotational inclusion.
\begin{theorem}{[Soundness of \corelan]}\label{thm:appendix:safety}
\begin{itemize}
\item\textbf{Denotations}
If $\hastype{\env}{\prog}{\typ}$ then
$\forall \sto\in \interp{\env}. \applysub{\sto}{\prog} \in \interp{\applysub{\sto}{\typ}}$.
\item\textbf{Preservation}
If \hastype{\emptyset}{\prog}{\typ}
       and $\evalsto{\prog}{w}$ then $\hastype{\emptyset}{w}{\typ}$.
\end{itemize}
\end{theorem}

\subsection{From Programs \& Types to Propositions \& Proofs}

The denotational soundness Theorem~\ref{thm:appendix:safety}
lets us interpret well typed programs as proofs of
propositions.

\mypara{``Definitions''}
A \emph{definition} $\defn$ is a sequence of reflected binders:
$$\defn \ ::= \ \bullet \spmid \erefb{x}{\gtyp}{e}{\defn}$$
A \emph{definition's environment} $\env(\defn)$ comprises
its binders and their reflected types:
\begin{align*}
\aenv(\bullet)                    \defeq & \emptyset \\
\aenv(\erefb{f}{\gtyp}{e}{\defn}) \defeq & (f, \exacttype{\gtyp}{e}),\ \env(\defn) \\
\end{align*}
A \emph{definition's substitution} $\sto(\defn)$ maps each binder
to its definition:
\begin{align*}
\sto(\bullet)                     \defeq & \emptysto \\
\sto(\erefb{f}{\gtyp}{e}{\defn})  \defeq & \extendsto{f}{\efix{f}\ e}{\sto(\defn)}
\end{align*}

\mypara{``Propositions''}
A \emph{proposition} is a type
$$\tbind{x_1}{\typ_1} \rightarrow \ldots
  \rightarrow \tbind{x_n}{\typ_n}
  \rightarrow \tref{v}{\tunit}{\ppn}$$
For brevity, we abbreviate propositions like the above to
$\tfun{\overline{x}}{\overline{\typ}}{\ttref{\ppn}}$
and we call $\ppn$ the \emph{proposition's refinement}.
For simplicity we assume that $\freevars{\typ_i} = \emptyset$.

\mypara{``Validity''}
%

A proposition $\tfun{\overline{x}}{\overline{\typ}}{\ttref{\ppn}}$
is \emph{valid under} $\defn$ if
$$\forall \overline{w} \in \interp{\overline{\typ}}.\
  \evalsto{\applysub{\sto(\defn)}{\SUBST{\ppn}{\overline{x}}{\overline{w}}}}{\etrue}$$
That is, the proposition is valid if its refinement
evaluates to $\etrue$ for every (well typed)
interpretation for its parameters $\overline{x}$
under $\defn$.

\mypara{``Proofs''}
A binder $\bd$ \emph{proves} a proposition $\gtyp$ under $\defn$ if
$$\hastype{\emptyset}{\defn[\eletb{x}{\typ}{\bd}{\eunit}]}{\tunit}$$
That is, if the binder $\bd$ has the proposition's type $\gtyp$
under the definition $\defn$'s environment.

\begin{theorem}{[Proofs]} \label{thm:validity}
If $\bd$ proves $\typ$ under $\defn$ then $\typ$ is valid under $d$.
\end{theorem}

\begin{proof}
As $\bd$ proves $\typ$ under $\defn$, we have
\begin{align}
\hastype{\emptyset}{\defn[\eletb{x}{\typ}{\bd}{\eunit}]}{\tunit}
\label{pf:1} \\
\intertext{By Theorem~\ref{thm:appendix:safety} on \ref{pf:1} we get}
\sto(\defn) \in \interp{\env(\defn)} \label{pf:2}\\
\intertext{Furthermore,  by the typing rules \ref{pf:1}
implies $\hastype{\env(\defn)}{\bd}{\typ}$ and hence, via Theorem~\ref{thm:appendix:safety}}
\forall \sub \in \interp{\env(\defn)}.\ \applysub{\sub}{\bd} \in \interp{\applysub{\sub}{\typ}}
\label{pf:3} \\
\intertext{Together, \ref{pf:2} and \ref{pf:3} imply}
\applysub{\sto(\defn)}{\bd} \in \interp{\applysub{\sto(\defn)}{\typ}}
\label{pf:4}
\intertext{By the definition of type denotations, we have}
\interp{\applysub{\sto(\defn)}{\typ}}
  \defeq \{ f\ |\ \typ \mbox{ is valid under}\ \defn\}
  \label{pf:5}
\end{align}
By \ref{pf:4}, the above set is not empty, and hence $\typ$ is valid under $\defn$.
\end{proof}

\mypara{Example: Fibonacci is increasing}
In \S~\ref{sec:overview} we verified that
under a definition $\defn$ that includes \fibname,
the term \fibincrname proves
$${\tfun{n}{\tnat}{\ttref{\fib{n} \leq \fib{(n+1)}}}}$$
Thus, by Theorem~\ref{thm:validity} we get
%
$$
\forall n. \evalsto{0 \leq n}{\etrue} \Rightarrow \evalsto{\fib{n} \leq \fib{(n+1)}}{\etrue}
$$

\section{Proof of Soundness}

We prove Theorem~\ref{thm:safety}
of \S~\ref{sec:appendix:types-reflection}
by reduction to Soundness of \undeclang~\citep{Vazou14}. 

\begin{theorem}{[Denotations]}~\label{tech:thm:denotations}
If $\hastype{\env}{\prog}{\typ}$ then
$\forall \sto\in \interp{\env}. \applysub{\sto}{\prog} \in \interp{\applysub{\sto}{\typ}}$.
\end{theorem}
\begin{proof}
We use the proof from~\citep{Vazou14-tech} and specifically Lemma 4
that is identical to the statement we need to prove. 
Since the proof proceeds by induction in the type derivation, 
we need to ensure that all the modified rules satisfy the statement. 
\begin{itemize}
\item\rtexact
 Assume 
 	\hastype{\env}{e}{\tref{v}{\btyp}{\reft_r\land v = e}}.
 By inversion
	\hastype{\env}{e}{\tref{v}{\btyp}{\reft_r}}(1). 
 By (1) and IH we get 
 $\forall \sto\in \interp{\env}. 
   \applysub{\sto}{e} \in \interp{\applysub{\sto}{\tref{v}{\btyp}{\reft_r}}}$.
 We fix a $\sto\in \interp{\env}$
 We get that if \evalsto{\applysub{\sto}{e}}{w}, 
 then $\evalsto{\applysub{\sto}{\reft_r}\subst{v}{w}}{\etrue}$.  
 By the Definition of $=$ we get that 
 $\evalsto{w = w}{\etrue}$. 
 Since $\evalsto{\applysub{\sto}{(v = e)}\subst{v}{w}}{w = w}$, 
 then $\evalsto{\applysub{\sto}{(\reft_r\land v = e)}\subst{v}{w}}{\etrue}$.  
 Thus
   $\applysub{\sto}{e} \in \interp{\applysub{\sto}{\tref{v}{\btyp}{\reft_r\land v = e}}}$
  and since this holds for any fixed $\sto$,  
 $\forall \sto\in \interp{\env}. 
   \applysub{\sto}{e} \in \interp{\applysub{\sto}{\tref{v}{\btyp}{\reft_r\land v = e}}}$.
\item\rtlet
  Assume 
	\hastype{\env}{\eletb{x}{\gtyp_x}{e_x}{\prog}}{\typ}.  
  By inversion
	\hastype{\env, \tbind{x}{\gtyp_x}}{e_x}{\gtyp_x} (1), 
	\hastype{\env, \tbind{x}{\gtyp_x}}{\prog}{\gtyp} (2), and
    \iswellformed{\env}{\typ} (3). 
 By IH 
	$\forall \sto\in \interp{\env, \tbind{x}{\gtyp_x}}. 
	\applysub{\sto}{e_x} \in \interp{\applysub{\sto}{\gtyp_x}}$ (1')
	$\forall \sto\in \interp{\env, \tbind{x}{\gtyp_x}}. 
	\applysub{\sto}{\prog} \in \interp{\applysub{\sto}{\gtyp}}$ (2'). 
 By (1') and by the type of $\efix{}$ 
	$\forall \sto\in \interp{\env, \tbind{x}{\gtyp_x}}. 
	\applysub{\sto}{\efix{x}\ e_x} \in \interp{\applysub{\sto}{\gtyp_x}}$. 
 By which,  (2') and (3)
	$\forall \sto\in \interp{\env}. 
	\applysub{\sto}{\SUBST{\prog}{x}{\efix{x}\ {e_x}}} \in \interp{\applysub{\sto}{\gtyp}}$.  
\item\rtreflect
  Assume 
  \hastype{\env}{\erefb{f}{\gtyp_f}{e}{\prog}}
			    {\typ}. 
  By inversion, 
    \hastype{\env}{\eletb{f}{\exacttype{\gtyp_f}{e}}{e}{\prog}}
			     {\typ}. 
  By IH, 
 	$\forall \sto\in \interp{\env}. 
	\applysub{\sto}{\eletb{f}{\exacttype{\gtyp_f}{e}}{e}{\prog}} \in \interp{\applysub{\sto}{\gtyp}}$.  
  Since denotations are closed under evaluation, 
	$\forall \sto\in \interp{\env}. 
	\applysub{\sto}{\erefb{f}{\exacttype{\gtyp_f}{e}}{e}{\prog}} \in \interp{\applysub{\sto}{\gtyp}}$.  

\item\rtfix
  In Theorem 8.3 from~\citep{Vazou14-tech} (and using the textbook proofs from~\citep{PLC})
  we proved that for each type $\typ$, $\efix{}_\typ \in \interp{(\typ \rightarrow \typ) \rightarrow \typ}$.
\end{itemize}
\end{proof}

\begin{theorem}{[Preservation]}
If \hastype{\emptyset}{\prog}{\typ}
       and $\evalsto{\prog}{w}$ then $\hastype{\emptyset}{w}{\typ}$.
\end{theorem}
\begin{proof}
In~\citep{Vazou14-tech} proof proceeds by iterative application 
of Type Preservation Lemma 7. 
Thus, it suffices to ensure Type Preservation in \corelan, which 
it true by the following Lemma.
\end{proof}

\begin{lemma}
If \hastype{\emptyset}{\prog}{\typ}
       and $\evals{\prog}{\prog'}$ then $\hastype{\emptyset}{\prog'}{\typ}$.
\end{lemma}
\begin{proof}
Since Type Preservation in \undeclang is proved by induction on the type derivation tree, 
we need to ensure that all the modified rules satisfy the statement. 
\begin{itemize}
\item\rtexact
 Assume 
 	\hastype{\emptyset}{\prog}{\tref{v}{\btyp}{\reft_r\land v = \prog}}.
 By inversion
	\hastype{\emptyset}{\prog}{\tref{v}{\btyp}{\reft_r}}.
 By IH we get 
	\hastype{\emptyset}{\prog'}{\tref{v}{\btyp}{\reft_r}}.
 By rule \rtexact we get 
 	\hastype{\emptyset}{\prog'}{\tref{v}{\btyp}{\reft_r\land v = \prog'}}.
 Since subtyping is closed under evaluation, we get 
 	\issubtype{\emptyset}{\tref{v}{\btyp}{\reft_r\land v = \prog'}}
 	                {\tref{v}{\btyp}{\reft_r\land v = \prog}}.
 By rule \rtsub we get 
 	\hastype{\emptyset}{\prog'}{\tref{v}{\btyp}{\reft_r\land v = \prog}}.

\item\rtlet
  Assume 
	\hastype{\emptyset}{\eletb{x}{\gtyp_x}{e_x}{\prog}}{\typ}.
 By inversion, 
   \hastype{\tbind{x}{\gtyp_x}}{e_x}{\gtyp_x}  (1), 
   \hastype{\tbind{x}{\gtyp_x}}{\prog}{\gtyp} (2), and
   \iswellformed{\env}{\typ} (3). 
 By rule \rtfix
   \hastype{\tbind{x}{\gtyp_x}}{\efix{x}\ {e_x}}{\gtyp_x}  (1').
 By (1'), (2) and Lemma 6 of~\citep{Vazou14-tech}, we get 
   \hastype{}{\SUBST{\prog}{x}{\efix{x}\ {e_x}}}{\SUBST{\gtyp}{x}{\efix{x}\ e_x}}. 
 By (3)
   $ \SUBST{\gtyp}{x}{\efix{x}\ e_x} \equiv \gtyp$.    
 Since 
   $\prog' \equiv \SUBST{\prog}{x}{\efix{x}\ {e_x}}$, 
 we have 
 \hastype{\emptyset}{\prog'}{\gtyp}. 

\item\rtreflect
  Assume 
	\hastype{\emptyset}{\erefb{x}{\gtyp_x}{e_x}{\prog}}{\typ}.
 By double inversion, with $\gtyp_x' \equiv \exacttype{\gtyp_x}{e_x} $; 
   \hastype{\tbind{x}{\gtyp_x'}}{e_x}{\gtyp_x'}  (1), 
   \hastype{\tbind{x}{\gtyp_x'}}{\prog}{\gtyp} (2), and
   \iswellformed{\env}{\typ} (3). 
 By rule \rtfix
   \hastype{\tbind{x}{\gtyp_x'}}{\efix{x}\ {e_x}}{\gtyp_x'}  (1').
 By (1'), (2) and Lemma 6 of~\citep{Vazou14-tech}, we get 
   \hastype{}{\SUBST{\prog}{x}{\efix{x}\ {e_x}}}{\SUBST{\gtyp}{x}{\efix{x}\ e_x}}. 
 By (3)
   $ \SUBST{\gtyp}{x}{\efix{x}\ e_x} \equiv \gtyp$.    
 Since 
   $\prog' \equiv \SUBST{\prog}{x}{\efix{x}\ {e_x}}$, 
 we have 
 \hastype{\emptyset}{\prog'}{\gtyp}. 

\item\rtfix
  This case cannot occur, as $\efix{}$ does not evaluate to any program. 
\end{itemize}
\end{proof}

\section{Algorithmic Checking \smtlan: Extended Version}
\label{sec:appendix:algorithmic}

Next, we describe \smtlan, a conservative approximation
of \corelan where the undecidable type subsumption rule
is replaced with a decidable one, yielding an SMT-based
algorithmic type system that enjoys the same soundness
guarantees.

\subsection{The SMT logic \smtlan}

\begin{figure}[t!]
\centering
$$
\begin{array}{rrcl}
\emphbf{Predicates} \quad
  & \pred & ::= &
    \pred \binop \pred \spmid
    \unop \pred \\
  && \spmid & n \spmid b \spmid x \spmid \dc \spmid  x\ \overline{\pred}\\
  && \spmid & \eif{\pred}{\pred}{\pred}
\\[0.03in]

\emphbf{Integers} \quad
  & n
  & ::= & 0, -1, 1, \dots
\\[0.03in]

\emphbf{Booleans} \quad
  & b
  & ::= & \etrue \spmid \efalse
\\[0.03in]

\emphbf{Bin Operators} \quad
  & \binop
  & ::= & = \spmid < \spmid \land \spmid + \spmid - \spmid \dots
\\[0.03in]

\emphbf{Un Operators} \quad
  & \unop
  & ::= & \lnot \spmid \dots 
\\[0.03in]

\emphbf{Model} \quad
  & \sigma
  & ::= & \sigma, (x:\pred) \spmid \emptyset
\\[0.03in]

\emphbf{Sort Arguments} \quad
  & \sort_a
  & ::= & \tint \spmid \tbool \spmid \tuniv 
         \spmid \tsmtfun{\sort_a}{\sort_a}
\\[0.03in]
\emphbf{Sorts} \quad
  & \sort
  & ::=  & \sort_a \rightarrow \sort
\end{array}
$$
\caption{\textbf{Syntax of \smtlan}}
\label{fig:appendix:smtsyntax}
\end{figure}

\mypara{Syntax: Terms \& Sorts}
Figure~\ref{fig:appendix:smtsyntax} summarizes the syntax
of \smtlan, the \emph{sorted} (SMT-)
decidable logic of quantifier-free equality,
uninterpreted functions and linear
arithmetic (QF-EUFLIA) ~\citep{Nelson81,SMTLIB2}.
The \emph{terms} of \smtlan include
integers $n$,
booleans $b$,
variables $x$,
data constructors $\dc$ (encoded as constants),
fully applied unary \unop and binary \binop operators,
and application $x\ \overline{\pred}$ of an uninterpreted function $x$.
The \emph{sorts} of \smtlan include built-in
integer \tint and \tbool for representing
integers and booleans.
%
The interpreted functions of \smtlan, \eg
the logical constants $=$ and $<$,
have the function sort $\sort \rightarrow \sort$.
Other functional values in \corelan, \eg
reflected \corelan functions and
$\lambda$-expressions, are represented as
first-order values with
uninterpreted sort \tsmtfun{\sort}{\sort}.
%
The universal sort \tuniv represents all other values.

\mypara{Semantics: Satisfaction \& Validity}
An assignment $\sigma$ is a mapping from
variables to terms
${\sigma \defeq \{ \assignto{x_1}{\pred_1}, \ldots, \assignto{x_n}{\pred_n} \}}$.
We write
${\sigma \models \pred}$
if the assignment $\sigma$ is a
\emph{model of} $\pred$, intuitively
if $\sigma\ \pred$ ``is true''~\cite{Nelson81}.
A predicate $\pred$ \emph{is satisfiable} if
there exists ${\sigma\models\pred}$.
A predicate $\pred$ \emph{is valid} if
for all assignments ${\sigma\models\pred}$.

\subsection{Transforming \corelan into \smtlan}
\label{subsec:appendix:embedding}

\newcommand\emptyaxioms{\ensuremath{\emptyset}\xspace}
\newcommand\andaxioms[2]{\ensuremath{{#1}\cup {#2}}\xspace}

\begin{figure}
\emphbf{Transformation}\hfill{\fbox{\tologicshort{\Gamma}{e}{\typ}{\pred}{\sort}{\smtenv}{\axioms}}}
$$
\inference{
}{
	\tologicshort{\env}{b}{\tbool}{b}{\tbool}{\emptyset}{\emptyaxioms}
}[\lgbool]
\qquad
\inference{
}{
	\tologicshort{\env}{n}{\tint}{n}{\tint}{\emptyset}{\emptyaxioms}
}[\lgint]
$$

$$
\inference{
    \tologicshort{\env}{e_1}{\typ}{\pred_1}{\embed{\typ}}{\smtenv}{\axioms_1} &
    \tologicshort{\env}{e_2}{\typ}{\pred_2}{\embed{\typ}}{\smtenv}{\axioms_2}
}{
	\tologicshort{\env}{e_1\binop e_2}{\tbool}{\pred_1 \binop\pred_2}{\tbool}{\smtenv}{\andaxioms{\axioms_1}{\axioms_2}}
}[\lgbinGEN]
$$

$$
\inference{
	\tologicshort{\env}{e}{\tbool}{\pred}{\tbool}{\smtenv}{\axioms}
}{
	\tologicshort{\env}{\unop e}{\tbool}{\unop\pred}{\tbool}{\smtenv}{\axioms}
}[\lgun]
\qquad
\inference{
}{
	\tologicshort{\env}{x}{\env(x)}{x}{\embed{\env(x)}}{\emptyset}{\emptyaxioms}
}[\lgvar]
$$

$$
\inference{
}{
	\tologicshort{\env}{c}{\constty{\odot}}{\smtvar{c}}{\embed{\constty{\odot}}}{\emptyset}{\emptyaxioms}
}[\lgpop]
\qquad
\inference{
}{
	\tologicshort{\env}{\dc}{\constty{\dc}}{\smtvar{\dc}}{\embed{\constty{\dc}}}{\emptyset}{\emptyaxioms}
}[\lgdc]
$$


$$
\inference{
    \tologicshort{\env, \tbind{x}{\typ_x}}{e}{}{\pred}{}{}{} &
  	\hastype{\env}{(\efun{x}{}{e})}{(\tfun{x}{\typ_x}{\typ})}\\
}{
	\tologicshort{\env}{\efun{x}{}{e}}{(\tfun{x}{\typ_x}{\typ})}
	        {\smtlamname{\embed{\typ_x}}{\embed{\typ}}\ {x}\ {\pred}}
	        {\sort'}{\smtenv, \tbind{f}{\sort'}}{\andaxioms{\{\axioms_{f_1}, \axioms_{f_2}\}}{\axioms}}
}[\lgfun]
$$

$$
\inference{
	\tologicshort{\env}{e'}{\typ_x}{\pred'}{\embed{\typ_x}}{\smtenv}{\axioms'}
	&
	\tologicshort{\env}{e}{\tfun{x}{\typ_x}{\typ}}{\pred}{\tsmtfun{\embed{\typ_x}}{\embed{\typ}}}{\smtenv}{\axioms}
	& 
	\hastype{\env}{e}{{\typ_x}\rightarrow{\typ}}
}{
	\tologicshort{\env}{e\ e'}{\typ}{\smtappname{\embed{\typ_x}}{\embed{\typ}}\ {\pred}\ {\pred'}}{\embed{\typ}}{\smtenv}{\andaxioms{\axioms}{\axioms'}}
}[\lgapp]
$$

$$
\inference{
	\tologicshort{\env}{e}{\tbool}{\pred}{\tbool}{\smtenv}{\axioms} & 
	\tologicshort{\env}{e_i\subst{x}{e}}{\typ}{\pred_i}{\embed{\typ}}{\smtenv}{\axioms_i}
}{
	\tologicshorttwolines{\env}{\ecaseexp{x}{e}{\etrue \rightarrow e_1; \efalse \rightarrow e_2}}{\typ}
	 {\eif{\pred}{\pred_1}{\pred_2}}{\embed{\typ}}{\smtenv}{\andaxioms{\axioms}{\axioms_i}}
}[\lgcaseBool]
$$

$$
\inference{
	\tologicshort{\env}{e}{\typ_e}{\pred}{\embed{\typ_e}}{\smtenv}{\axioms}\\
	\tologicshort{\env}{e_i\subst{\overline{y_i}}{\overline{\selector{\dc_i}{}\ x}}\subst{x}{e}}{\typ}{\pred_i}{\embed{\typ}}{\smtenv}{\axioms_i}
}{
	\tologicshorttwolines{\env}{\ecase{x}{e}{\dc_i}{\overline{y_i}}{e_i}}{\typ}
	 {\eif{\smtappname{}{}\ \checkdc{\dc_1}\ \pred}{\pred_1}{\ldots} \ \mathtt{else}\ \pred_n}{\embed{\typ}}{\smtenv}
	 {\andaxioms{\axioms}{\axioms_i}}
}[\lgcase]
$$
\caption{\textbf{Transforming \corelan terms into \smtlan.}}
\label{fig:defunc}
\end{figure}

The judgment
\tologicshort{\env}{e}{\typ}{\pred}{\sort}{\smtenv}{\axioms}
states that a $\corelan$ term $e$ is transformed,
under an environment $\env$, into a
$\smtlan$ term $\pred$.
The transformation rules are summarized in Figure~\ref{fig:defunc}.

\mypara{Embedding Types}
We embed \corelan types into \smtlan sorts as:
$$
\begin{array}{rclcrcl}
\embed{\tint}                       & \defeq &  \tint &  &
\embed{T}                           & \defeq &  \tuniv \\
\embed{\tbool}                      & \defeq &  \tbool & &
\embed{\tfun{x}{\typ_x}{\typ}} & \defeq & \tsmtfun{\embed{\typ_x}}{\embed{\typ}}
\end{array}
$$

\mypara{Embedding Constants}
Elements shared on both \corelan and \smtlan
translate to themselves.
These elements include
booleans (\lgbool),
integers (\lgint),
variables (\lgvar),
binary (\lgbinGEN)
and unary (\lgun)
operators.
SMT solvers do not support currying,
and so in \smtlan, all function symbols
must be fully applied.
Thus, we assume that all applications
to primitive constants and data
constructors are \emph{saturated},
\ie fully applied, \eg by converting
source level terms like "(+ 1)" to
"(\z -> z + 1)".
%



\mypara{Embedding Functions}
As \smtlan is a first-order logic, we
embed $\lambda$-abstraction and application
using the uninterpreted functions
\smtlamname{}{} and \smtappname{}{}.
We embed $\lambda$-abstractions
using $\smtlamname{}{}$ as shown in rule~\lgfun.
The term $\efun{x}{}{e}$ of type
${\typ_x \rightarrow \typ}$ is transformed
to
${\smtlamname{\sort_x}{\sort}\ x\ \pred}$
of sort
${\tsmtfun{\sort_x}{\sort}}$, where
$\sort_x$ and $\sort$ are respectively
$\embed{\typ_x}$ and $\embed{\typ}$,
${\smtlamname{\sort_x}{\sort}}$
is a special uninterpreted function
of sort
${\sort_x \rightarrow \sort\rightarrow\tsmtfun{\sort_x}{\sort}}$,
and
$x$ of sort $\sort_x$ and $r$ of sort $\sort$ are
the embedding of the binder and body, respectively.
As $\smtlamname{}{}$ is just an SMT-function,
it \emph{does not} create a binding for $x$.
Instead, the binder $x$ is renamed to
a \emph{fresh} name pre-declared in
the SMT environment.

\mypara{Embedding Applications}
Dually, we embed applications via
defunctionalization~\citep{Reynolds72}
using an uninterpreted \emph{apply}
function
$\smtappname{}{}$ as shown in rule~\lgapp.
The term ${e\ e'}$, where $e$ and $e'$ have
types ${\typ_x \rightarrow \typ}$ and $\typ_x$,
is transformed to
${\tbind{\smtappname{\sort_x}{\sort}\ \pred\ \pred'}{\sort}}$
where
$\sort$ and $\sort_x$ are respectively $\embed{\typ}$ and $\embed{\typ_x}$,
the
${\smtappname{\sort_x}{\sort}}$
is a special uninterpreted function of sort
${\tsmtfun{\sort_x}{\sort} \rightarrow \sort_x \rightarrow \sort}$,
and
$\pred$ and $\pred'$ are the respective translations of $e$ and $e'$.

\mypara{Embedding Data Types}
Rule~\lgdc translates each data constructor to a
predefined \smtlan constant ${\smtvar{\dc}}$ of
sort ${\embed{\constty{\dc}}}$.
Let $\dc_i$ be a non-boolean data constructor such that
$$
\constty{\dc_i} \defeq \typ_{i,1} \rightarrow \dots \rightarrow \typ_{i,n} \rightarrow \typ
$$
Then the \emph{check function}
${\checkdc{{\dc_i}}}$ has the sort
$\tsmtfun{\embed{\typ}}{\tbool}$,
and the \emph{select function}
${\selector{\dc}{i,j}}$ has the sort
$\tsmtfun{\embed{\typ}}{\embed{\typ_{i,j}}}$.
Rule~\lgcase translates case-expressions
of \corelan into nested $\mathtt{if}$
terms in \smtlan, by using the check
functions in the guards, and the
select functions for the binders
of each case.
%
%
%
For example, following the above, the body of the list append function
%
\begin{code}
  []     ++ ys = ys
  (x:xs) ++ ys = x : (xs ++ ys)
\end{code}
is reflected into the \smtlan refinement:
$$
\ite{\mathtt{isNil}\ \mathit{xs}}
    {\mathit{ys}}
    {\mathtt{sel1}\ \mathit{xs}\
       \dcons\
       (\mathtt{sel2}\ \mathit{xs} \ \mathtt{++}\  \mathit{ys})}
$$
We favor selectors to the axiomatic translation of
HALO~\citep{HALO} and \fstar~\cite{fstar} to avoid
universally quantified formulas and the resulting
instantiation unpredictability.


\subsection{Correctness of Translation}

Informally, the translation relation $\tologicshort{\env}{e}{}{\pred}{}{}{}$
is correct in the sense that if $e$ is a terminating boolean expression
then $e$ reduces to \etrue \textit{iff} $\pred$ is SMT-satisfiable
by a model that respects $\beta$-equivalence.

\begin{definition}[$\beta$-Model]\label{def:beta-model}
A $\beta-$model $\bmodel$ is an extension of a model $\sigma$
where $\smtlamname{}{}$ and $\smtappname{}{}$
satisfy the axioms of $\beta$-equivalence:
$$
\begin{array}{rcl}
\forall x\ y\ e. \smtlamname{}{}\ x\ e
  & = & \smtlamname{}{}\ y\ (e\subst{x}{y}) \\
\forall x\ e_x\ e. (\smtappname{}{}\ (\smtlamname{}{}\ x\ e)\ e_x
  & = &  e\subst{x}{e_x}
\end{array}
$$
\end{definition}

\mypara{Semantics Preservation}
We define the translation of a \corelan term
into \smtlan under the empty environment as
${\embed{e} \defeq \pred}$
if ${\tologicshort{\emptyset}{\refa}{}{\pred}{}{}{}}$.
A \emph{lifted substitution}
$\theta^\perp$ is a set of models $\sigma$
where each ``bottom'' in the substitution
$\theta$ is mapped to an arbitrary logical
value of the respective sort~\citep{Vazou14}.
We connect the semantics of \corelan and translated
\smtlan via the following theorems:

\begin{theorem}\label{thm:embedding-general}
If ${\tologicshort{\env}{\refa}{}{\pred}{}{}{}}$,
then for every ${\sub\in\interp{\env}}$
and every ${\sigma\in {\sub^\perp}}$,
if $\evalsto{\applysub{\sub^\perp}{\refa}}{v}$
then $\sigma^\beta \models \pred = \embed{v}$.
\end{theorem}


\begin{corollary}\label{thm:appendix:embedding}
If ${\hastype{\env}{\refa}{\tbool}}$, $e$ reduces to a value and
${\tologicshort{\env}{\refa}{\tbool}{\pred}{\tbool}{\smtenv}{\axioms}}$,
then for every ${\sub\in\interp{\env}}$
and every ${\sigma\in {\sub^\perp}}$,
$\evalsto{\applysub{\sub^\perp}{\refa}}{\etrue}$ iff
$\sigma^\beta \models \pred$.
\end{corollary}

\subsection{Decidable Type Checking}
\begin{figure}[t!]
\centering
$$
\begin{array}{rrcl}
\emphbf{Refined Types} \quad
  & \typ
  & ::=   & \tref{v}{\btyp^{[\tlabel]}}{\reft} \spmid \tfun{x}{\typ}{\typ}
\\[0.10in]
\end{array}
$$
\emphbf{Well Formedness}\hfill{\fbox{\aiswellformed{\env}{\typ}}}\\
$$
\inference{
  \ahastype{\env,\tbind{v}{\btyp}}{\refa}{\tbool^{\tlabel}}
}{
  \aiswellformed{\env}{\tref{v}{\btyp}{\refa}}
}[\rwbase]
$$
\emphbf{Subtyping}\hfill{\fbox{\aissubtype{\env}{\typ}{\typ'}}}\\
$$
\inference{
\env' \defeq \env,\tbind{v}{\{\btyp^\tlabel | \refa\}} &
\tologicshort{\env'}{\refa'}{\tbool}{\pred'}{}{}{} &
\smtvalid{\vcond{\env'}{\pred'}}
}{
 \aissubtype{\env}{\tref{v}{\btyp}{\refa}}{\tref{v}{\btyp}{\refa'}}
}[\rsubbase]
$$
\caption{\textbf{Algorithmic Typing (other rules in Figs~\ref{fig:appendix:syntax} and \ref{fig:appendix:typing}.)}}
\label{fig:modifications}
\end{figure}

Figure~\ref{fig:modifications} summarizes the modifications required
to obtain decidable type checking.
Namely, basic types are extended with labels that track termination
and subtyping is checked via an SMT solver.

\mypara{Termination}
Under arbitrary beta-reduction semantics
(which includes lazy evaluation), soundness
of refinement type checking requires checking
termination, for two reasons:
(1)~to ensure that refinements cannot diverge, and
(2)~to account for the environment during subtyping~\citep{Vazou14}.
We use \tlabel to mark provably terminating
computations, and extend the rules to use
refinements to ensure that if
${\ahastype{\env}{e}{\tref{v}{\btyp^\tlabel}{r}}}$,
then $e$ terminates~\citep{Vazou14}.
%

\mypara{Verification Conditions}
The \emph{verification condition} (VC)
${\vcond{\env}{\pred}}$
is \emph{valid} only if the set of values
described by $\env$, is subsumed by
the set of values described by $\pred$.
$\env$ is embedded into logic by conjoining
(the embeddings of) the refinements of
provably terminating binders~\cite{Vazou14}:
%
\begin{align*}
\embed{\env} \defeq & \bigwedge_{x \in \env} \embed{\env, x} \\
\intertext{where we embed each binder as}
\embed{\env, x} \defeq & \begin{cases}
                           \pred  & \text{if } \env(x)=\tref{v}{\btyp^{\tlabel}}{e},\
                                    \tologicshort{\env}{e\subst{v}{x}}{\btyp}{\pred}{\embed{\btyp}}{\smtenv}{\axioms} \\
                           \etrue & \text{otherwise}.
                         \end{cases}
\end{align*}


\mypara{Subtyping via SMT Validity}
We make subtyping, and hence, typing decidable,
by replacing the denotational base subtyping
rule $\rsubbase$ with a conservative,
algorithmic version that uses an SMT
solver to check the validity of the subtyping VC.
We use Corollary~\ref{thm:appendix:embedding} to prove
soundness of subtyping. 
\begin{lemma}\label{lem:appendix:subtyping} 
If {\aissubtype{\env}{\tref{v}{\btyp}{e_1}}{\tref{v}{\btyp}{e_2}}}
then {\issubtype{\env}{\tref{v}{\btyp}{e_1}}{\tref{v}{\btyp}{e_2}}}.
\end{lemma}

\mypara{Soundness of \smtlan}
Lemma~\ref{lem:appendix:subtyping} directly implies the soundness of \smtlan.
\begin{theorem}[Soundness of \smtlan]\label{thm:appendix:soundness-smt}
If \ahastype{\env}{e}{\typ} then \hastype{\env}{e}{\typ}.
\end{theorem}

\section{Soundness of Algorithmic Verification}
In this section we prove soundness of Algorithmic verification, 
by proving the theorems of \S~\ref{sec:algorithmic}
by referring to the proofs in~\citep{Vazou14-tech}. 

\subsection{Transformation}

\begin{definition}[Initial Environment]\label{def:initialsmt}
 We define the initial SMT environment \smtenvinit to include
 $$
 \begin{array}{rcll}
 \smtvar{c}  &\colon &\embed{\constty{c}}
   &\forall c\in \corelan\\
 \smtlamname{\sort_x}{\sort}&\colon&\sort_x \rightarrow \sort\rightarrow\tsmtfun{\sort_x}{\sort}
   &\forall \sort_x, \sort\in \smtlan\\
 \smtappname{\sort_x}{\sort}&\colon&\tsmtfun{\sort_x}{\sort} \rightarrow \sort_x \rightarrow \sort
   &\forall \sort_x, \sort\in \smtlan\\
 \smtvar{\dc}&\colon&\embed{\constty{\dc}}
   &\forall\dc\in\corelan\\
 \checkdc{\dc}&\colon&\embed{T \rightarrow \tbool}
   &\forall \dc\in \corelan\ \text{of data type}\ T \\
 \selector{\dc}{i}&\colon&\embed{T \rightarrow \typ_i}
   &\forall \dc\in \corelan\ \text{of data type}\ T \\
   &&&\text{and}\ i\text{-th argument}\ \typ_i \\
 {x^{\sort}_i} & \colon&{\sort}& \forall \sort \in \smtlan \text{and} 1\leq i\leq \maxlamarg\\
 \end{array}
 $$
Where $x^{\sort}_i$ are $\maxlamarg$ global names that only appear as lambda arguments.
\end{definition}

We modify the $\lgfun$ rule to ensure that 
logical abstraction is performed 
using the minimum globally defined lambda argument that is not already abstracted. 
We do so, using the helper function \maxlam{\sort}{\pred}:

\begin{align*}
\maxlam{\sort}{\smtlamname{\sort}{\sort'}\ {x^{\sort}_i}\ \pred} =& \mathtt{max}(i, \maxlam{\sort}{\pred})\\
\maxlam{\sort}{r\ \overline{r}} =& \mathtt{max}(\maxlam{\sort}{\pred, \overline{\pred}}) \\
\maxlam{\sort}{\pred_1 \binop \pred_2} = &  \mathtt{max}(\maxlam{\sort}{\pred_1, \pred_2})\\
\maxlam{\sort}{\unop \pred} =& \maxlam{\sort}{\pred}\\
\maxlam{\sort}{\eif{\pred}{\pred_1}{\pred_2}} =& \mathtt{max}(\maxlam{\sort}{\pred, \pred_1, \pred_2}) \\
\maxlam{\sort}{\pred} =& 0 \\
\end{align*}
$$
\inference{
    i = \maxlam{\embed{\typ_x}}{\pred} & i < \maxlamarg & y = x^{\embed{\typ_x}}_{i+1} \\ 
    \tologicshort{\env, \tbind{y}{\typ_x}}{e\subst{x}{y}}{}{\pred}{}{}{} &
  	\hastype{\env}{(\efun{x}{}{e})}{(\tfun{x}{\typ_x}{\typ})}\\
}{
	\tologicshort{\env}{\efun{x}{}{e}}{}
	        {\smtlamname{\embed{\typ_x}}{\embed{\typ}}\ {y}\ {\pred}}
	        {}{}{}
}[\lgfun]
$$

\begin{lemma}[Type Transformation]
If \tologicshort{\env}{e}{\typ}{p}{\sort}{\smtenv}{\axioms},
and \hastype{\env}{e}{\typ}, then
\smthastype{\smtenvinit, \embed{\env}}{p}{\embed{\typ}}.
\end{lemma}
\begin{proof}
We proceed by induction on the translation
\begin{itemize}
\item \lgbool : Since $\embed{\tbool} = \tbool$,
If \hastype{\env}{b}{\tbool}, then  
\smthastype{\smtenvinit, \embed{\env}}{b}{\embed{\tbool}}.

\item \lgint :  Since $\embed{\tint} = \tint$,
If \hastype{\env}{n}{\tint}, then  
\smthastype{\smtenvinit, \embed{\env}}{n}{\embed{\tint}}.

\item \lgun : 
Since $\hastype{\env}{\lnot\ e}{\typ}$, then it should be 
$\hastype{\env}{e}{\tbool}$ and $\typ \equiv \tbool$.
By IH,  
\smthastype{\smtenvinit, \embed{\env}}{\pred}{\embed{\tbool}}, 
thus 
\smthastype{\smtenvinit, \embed{\env}}{\lnot \pred}{\embed{\tbool}}. 

\item \lgbinGEN
Assume 	
 $\tologicshort{\env}{e_1\binop e_2}{\tbool}{\pred_1 \binop\pred_2}{\tbool}{\smtenv}{\andaxioms{\axioms_1}{\axioms_2}}$.  
By inversion
    \tologicshort{\env}{e_1}{\typ}{\pred_1}{\embed{\typ}}{\smtenv}{\axioms_1}, and
    \tologicshort{\env}{e_2}{\typ}{\pred_2}{\embed{\typ}}{\smtenv}{\axioms_2}.
Since 
  \hastype{\env}{e_1\binop e_2}{\typ}, then 
  \hastype{\env}{e_1}{\typ_1} and  
  \hastype{\env}{e_1}{\typ_2}.
By IH,  
\smthastype{\smtenvinit, \embed{\env}}{\pred_1}{\embed{\typ_1}} and  
\smthastype{\smtenvinit, \embed{\env}}{\pred_2}{\embed{\typ_2}}.
We split cases on $\binop$
\begin{itemize}
\item If $\binop \equiv =$, then 
  $\typ_1 = \typ_2$, thus $\embed{\typ_1} = \embed{\typ_2}$
  and $\embed{\typ} = \typ = \tbool$.

\item If $\binop \equiv <$, then 
  $\typ_1 = \typ_2 = \tint$, thus $\embed{\typ_1} = \embed{\typ_2} = \tint$
  and $\embed{\typ} = \typ = \tbool$.

\item If $\binop \equiv \land$, then 
  $\typ_1 = \typ_2 = \tbool$, thus $\embed{\typ_1} = \embed{\typ_2} = \tbool$
  and $\embed{\typ} = \typ = \tbool$.
\item If $\binop \equiv +$ or $\binop \equiv -$, then 
  $\typ_1 = \typ_2 = \tint$, thus $\embed{\typ_1} = \embed{\typ_2} = \tint$
  and $\embed{\typ} = \typ = \tint$.
\end{itemize}

\item \lgvar : 
Assume 
	\tologicshort{\env}{x}{\env(x)}{x}{\embed{\env(x)}}{\emptyset}{\emptyaxioms}
Then 
   \hastype{\env}{x}{\env (x)} and 
   \smthastype{\smtenvinit, \embed{\env}}{x}{\embed{\env} (x)}.
But by definition 
  $(\embed{\env}) (x) = \embed{\env(x)}$. 

\item \lgpop : 
Assume 
	\tologicshort{\env}{c}{\constty{\odot}}{\smtvar{c}}{\embed{\constty{\odot}}}{\emptyset}{\emptyaxioms}
Also, 
   \hastype{\env}{c}{\constty{c}} and 
   \smthastype{\smtenvinit, \embed{\env}}{\smtvar{c}}{\smtenvinit (\smtvar{c})}.
But by Definition~\ref{def:initialsmt}
  $\smtenvinit (\smtvar{c}) = \embed{\constty{c}}$. 

\item \lgdc : 
Assume 
	\tologicshort{\env}{\dc}{\constty{\dc}}{\smtvar{\dc}}{\embed{\constty{\dc}}}{\emptyset}{\emptyaxioms}
Also, 
   \hastype{\env}{\dc}{\constty{\dc}} and 
   \smthastype{\smtenvinit, \embed{\env}}{\smtvar{\dc}}{\smtenvinit (\smtvar{\dc})}.
But by Definition~\ref{def:initialsmt}
  $\smtenvinit(\smtvar{\dc}) = \embed{\constty{c}}$. 

\item \lgfun : 
Assume 
	\tologicshort{\env}{\efun{x}{}{e}}{(\tfun{x}{\typ_x}{\typ})}
	        {\smtlamname{\embed{\typ_x}}{\embed{\typ}}\ {x^{\embed{\typ_x}}_{i}}\ {\pred}}
	        {\sort'}{\smtenv, \tbind{f}{\sort'}}{\andaxioms{\{\axioms_{f_1}, \axioms_{f_2}\}}{\axioms}}. 
By inversion 
    $i \leq \maxlamarg$
	\tologicshort{\env, \tbind{x^{\embed{\typ_x}}_{i}}{\typ_x}}{e\subst{x}{x^{\embed{\typ_x}}_{i}}}{}{\pred}{}{}{}, and
  	\hastype{\env}{(\efun{x}{}{e})}{(\tfun{x}{\typ_x}{\typ})}.
By the Definition~\ref{def:initialsmt} on $\smtlamname{}{}$, $x^{\sort}_i$ and induction, we get
   \smthastype{\smtenvinit, \embed{\env}}
     {\smtlamname{\embed{\typ_x}}{\embed{\typ}}\ {x^{\embed{\typ_x}}_{i}}\ {\pred}}
     {\tsmtfun{\embed{\typ_x}}{\embed{\typ}}}.
By the definition of the type embeddings we have
$\embed{\tfunbasic{x}{\typ_x}{\typ}} \defeq \tsmtfun{\embed{\typ_x}}{\embed{\typ}}$.

\item \lgapp : 
Assume 
\tologicshort{\env}{e\ e'}{\typ}
  {\smtappname{\embed{\typ_x}}{\embed{\typ}}\ {\pred}\ {\pred'}}{\embed{\typ}}{\smtenv}{\andaxioms{\axioms}{\axioms'}}. 
By inversion, 
	\tologicshort{\env}{e'}{\typ_x}{\pred'}{\embed{\typ_x}}{\smtenv}{\axioms'}, 
	\tologicshort{\env}{e}{\tfun{x}{\typ_x}{\typ}}{\pred}{\tsmtfun{\embed{\typ_x}}{\embed{\typ}}}{\smtenv}{\axioms},
	\hastype{\env}{e}{\tfun{x}{\typ_x}{\typ}}. 
By IH and the type of $\smtappname{}{}$ we get that 
   \smthastype{\smtenvinit, \embed{\env}}
     {\smtappname{\embed{\typ_x}}{\embed{\typ}}\ {\pred}\ {\pred'}}
     {\embed{\typ}}.

\item \lgcaseBool : 
Assume
	\tologicshort{\env}{\ecaseexp{x}{e}{\etrue \rightarrow e_1; \efalse \rightarrow e_2}}{\typ}
	 {\eif{\pred}{\pred_1}{\pred_2}}{\embed{\typ}}{\smtenv}{\andaxioms{\axioms}{\axioms_i}}
Since 
  \hastype{\env}{\ecaseexp{x}{e}{\etrue \rightarrow e_1; \efalse \rightarrow e_2}}{\typ}, then 
  \hastype{\env}{e}{\tbool}, 
  \hastype{\env}{e_1}{\typ}, and
  \hastype{\env}{e_2}{\typ}.
By inversion and IH, 
  \smthastype{\smtenvinit, \embed{\env}}{\pred}{\tbool}, 
  \smthastype{\smtenvinit, \embed{\env}}{\pred_1}{\embed{\typ}}, and
  \smthastype{\smtenvinit, \embed{\env}}{\pred_2}{\embed{\typ}}.
Thus, 
  \smthastype{\smtenvinit, \embed{\env}}{\eif{\pred}{\pred_1}{\pred_2}}{\embed{\typ}}.

\item \lgcase : 
Assume 
	\tologicshort{\env}{\ecase{x}{e}{\dc_i}{\overline{y_i}}{e_i}}{\typ}
	 {\eif{\checkdc{\dc_1}\ \pred}{\pred_1}{\ldots} \ \mathtt{else}\ \pred_n}{\embed{\typ}}{\smtenv}
	 {\andaxioms{\axioms}{\axioms_i}} and
	 \hastype{\env}{\ecase{x}{e}{\dc_i}{\overline{y_i}}{e_i}}{\typ}.
By inversion we get 
	\tologicshort{\env}{e}{\typ_e}{\pred}{\embed{\typ_e}}{\smtenv}{\axioms} and
	\tologicshort{\env}{e_i\subst{\overline{y_i}}{\overline{\selector{\dc_i}{}\ x}}\subst{x}{e}}{\typ}{\pred_i}{\embed{\typ}}{\smtenv}{\axioms_i}.
By IH and the Definition~\ref{def:initialsmt} on the checkers and selectors, we get
  \smthastype{\smtenvinit, \embed{\env}}{\eif{\checkdc{\dc_1}\ \pred}{\pred_1}{\ldots} \ \mathtt{else}\ \pred_n}{\embed{\typ}}.
\end{itemize}
\end{proof}

\newcommand\tsub[1]{\ensuremath{{\theta_{#1}^\perp}}\xspace}
\newcommand\track[2]{\ensuremath{\langle #1; #2\rangle}\xspace}
\begin{theorem}\label{thm:approximation}
If \tologicshort{\env}{\refa}{}{\pred}{}{}{},
then for every substitution $\sub\in\interp{\env}$
and every model $\sigma\in\interp{\theta^\perp}$,
if $\evalsto{\applysub{\theta^\perp}{\refa}}{v}$ then $\sigma^\beta \models \pred = \interp{v}$.
\end{theorem}
\begin{proof}
We proceed using the notion of tracking substitutions from Figure 8 of~\citep{Vazou14-tech}. 
Since $\evalsto{\applysub{\theta^\perp}{\refa}}{v}$, 
there exists a sequence of evaluations via tracked substitutions, 
$$
\track{\tsub{1}}{e_1} \hookrightarrow \dots \track{\tsub{i}}{e_i} \dots \hookrightarrow \track{\tsub{n}}{e_n}
$$
with $\tsub{1}\equiv\tsub{}$, $e_1\equiv e$, and $e_n\equiv v$. 
Moreover, each $e_{i+1}$ is well formed under $\Gamma$, 
thus it has a translation 
$\tologicshort{\Gamma}{\refa_{i+1}}{}{\pred_{i+1}}{}{}{}$. 
Thus we can iteratively apply Lemma~\ref{lemma:approximation} $n-1$ times and 
since $v$ is a value the extra variables in $\tsub{n}$ are irrelevant, thus we
get the required
$\sigma^\beta \models \pred = \interp{v}$. 
\end{proof}

For Boolean expressions we specialize the above to
\begin{corollary}\label{corollary:embedding}
If \hastype{\env}{\refa}{\tbool^\downarrow} and
\tologicshort{\env}{\refa}{\tbool}{\pred}{\tbool}{\smtenv}{\axioms},
then for every substitution $\sub\in\interp{\env}$
and every model $\sigma\in\interp{\theta^\perp}$,
$\evalsto{\applysub{\theta^\perp}{\refa}}{\etrue} \iff \sigma^\beta \models \pred $
\end{corollary}
\begin{proof}
We prove the left and right implication separately:
\begin{itemize}
\item $\Rightarrow$
By direct application of Theorem~\ref{thm:approximation} for $v \equiv \etrue$. 

\item $\Leftarrow$ 
Since $\refa$ is terminating, 
$\evalsto{\applysub{\theta^\perp}{\refa}}{v}$. 
with either $v \equiv \etrue$ or $v \equiv \efalse$. 
Assume $v \equiv \efalse$, then by Theorem~\ref{thm:approximation}, 
$\bmodel \models \lnot \pred$, which is a contradiction. 
Thus, $v \equiv \etrue$.
\end{itemize}
\end{proof}

\begin{lemma}[Equivalence Preservation]\label{lemma:approximation}
If \tologicshort{\env}{\refa}{}{\pred}{}{}{},
then for every substitution $\sub\in\interp{\env}$
and every model $\sigma\in\interp{\theta^\perp}$,
if  
$\track{\tsub{}}{\refa}\hookrightarrow\track{\tsub{2}}{\refa_2}$
and
for $\Gamma \subseteq \Gamma_2$ so that $\tsub{2} \in \interp{\Gamma_2}$
and $\bmodel_2\in\interp{\tsub{2}}$,
$\tologicshort{\Gamma_2}{\refa_2}{}{\pred_2}{}{}{}$
then 
$\bmodel \cup (\bmodel_2 \setminus \bmodel) \models \pred = \pred_2$.
\end{lemma}

\begin{proof}
We proceed by case analysis on the derivation 
$\track{\tsub{}}{\refa}\hookrightarrow\track{\tsub{2}}{\refa_2}$.
\begin{itemize}
\item 
Assume
	$\track{\tsub{}}{\refa_1\ \refa_2}\hookrightarrow\track{\tsub{2}}{\refa_1'\ \refa_2}$. 
By inversion 
	$\track{\tsub{}}{\refa_1}\hookrightarrow\track{\tsub{2}}{\refa_1'}$.
Assume
  $\tologicshort{\Gamma}{\refa_1}{}{\pred_1}{}{}{}$, 	
  $\tologicshort{\Gamma}{\refa_2}{}{\pred_2}{}{}{}$, 	
  $\tologicshort{\Gamma_2}{\refa_1'}{}{\pred_1'}{}{}{}$.  	
By IH 
  $\bmodel \cup (\bmodel_2 \setminus \bmodel)
    \models \pred_1 = \pred_1'$, 
thus  
  $\bmodel \cup (\bmodel_2 \setminus \bmodel)
    \models \smapp{\pred_1}{\pred_2} = \smapp{\pred_1'}{\pred_2}$. 
    
\item 
Assume
	$\track{\tsub{}}{c\ \refa}\hookrightarrow\track{\tsub{2}}{c\ \refa'}$. 
By inversion 
	$\track{\tsub{}}{\refa}\hookrightarrow\track{\tsub{2}}{\refa'}$.
Assume
  $\tologicshort{\Gamma}{\refa}{}{\pred}{}{}{}$,	
  $\tologicshort{\Gamma}{\refa'}{}{\pred'}{}{}{}$.  	
By IH 
  $\bmodel \cup (\bmodel_2 \setminus \bmodel)
    \models \pred = \pred'$, 
thus  
  $\bmodel \cup (\bmodel_2 \setminus \bmodel)
    \models \smapp{c}{\pred} = \smapp{c}{\pred'}$. 

\item 
Assume
	$\track{\tsub{}}{\ecase{x}{\refa}{\dc_i}{\overline{y_i}}{e_i}}
	  \hookrightarrow\track
	  {\tsub{2}}{\ecase{x}{\refa'}{\dc_i}{\overline{y_i}}{e_i}}$. 
By inversion 
	$\track{\tsub{}}{\refa}\hookrightarrow\track{\tsub{2}}{\refa'}$.
Assume
  $\tologicshort{\Gamma}{\refa}{}{\pred}{}{}{}$,	
  $\tologicshort{\Gamma}{\refa'}{}{\pred'}{}{}{}$.  	
  $\tologicshort{\env}{e_i\subst{\overline{y_i}}{\overline{\selector{\dc_i}{}\ x}}\subst{x}{e}}{\typ}{\pred_i}{\embed{\typ}}{\smtenv}{\axioms_i}$.
By IH 
  $\bmodel \cup (\bmodel_2 \setminus \bmodel)
    \models \pred = \pred'$, 
thus  
  $\bmodel \cup (\bmodel_2 \setminus \bmodel)
   \models
	 {\eif{\checkdc{\dc_1}\ \pred}{\pred_1}{\ldots} \ \mathtt{else}\ \pred_n}{\embed{\typ}}
 $\\ $=
	 {\eif{\checkdc{\dc_1}\ \pred'}{\pred_1}{\ldots} \ \mathtt{else}\ \pred_n}{\embed{\typ}}
  $.	 

\item 
Assume
	$\track{\tsub{}}{D\ \overline{\refa_i}\ \refa\ \overline{\refa_j}}
	\hookrightarrow\track{\tsub{2}}{D\ \overline{\refa_i}\ \refa'\ \overline{\refa_j}}$. 
By inversion 
	$\track{\tsub{}}{\refa}\hookrightarrow\track{\tsub{2}}{\refa'}$.
Assume
  $\tologicshort{\Gamma}{\refa}{}{\pred}{}{}{}$,	
  $\tologicshort{\Gamma}{\refa_i}{}{\pred_i}{}{}{}$,	
  $\tologicshort{\Gamma}{\refa'}{}{\pred'}{}{}{}$.  	
By IH 
  $\bmodel \cup (\bmodel_2 \setminus \bmodel)
    \models \pred = \pred'$, 
thus  
  $\bmodel \cup (\bmodel_2 \setminus \bmodel)
    \models \smapp{D}{\overline{\pred_i}\ \pred\ \overline{\pred_j}} 
          = \smapp{D}{\overline{\pred_i}\ \pred'\ \overline{\pred_j}}$. 

\item 
Assume
	$\track{\tsub{}}{c\ w}
	\hookrightarrow\track{\tsub{}}{\delta(c, w)}$. 
By the definition of the syntax, $c\ w$ is a fully applied logical operator, thus
  $\bmodel \cup (\bmodel_2 \setminus \bmodel)
    \models {c\ w} = \interp{\delta(c, w)}$

\item 
Assume
	$\track{\tsub{}}{(\efun{x}{}{\refa}) \refa_x}
	\hookrightarrow\track{\tsub{}}{\refa\subst{x}{\refa_x}}$. 
Assume
  $\tologicshort{\Gamma}{\refa}{}{\pred}{}{}{}$,	
  $\tologicshort{\Gamma}{\refa_x}{}{\pred_x}{}{}{}$. 
Since $\bmodel$ is defined to satisfy the $\beta$-reduction axiom, 
  $\bmodel \cup (\bmodel_2 \setminus \bmodel)
    \models \smapp{(\smlam{x}{\refa})}{\pred_x} =\pred\subst{x}{\pred_x}$. 
  
\item 
Assume
	$\track{\tsub{}}{\ecase{x}{D_j\ \overline{e}}{\dc_i}{\overline{y_i}}{e_i}}
	\hookrightarrow\track{\tsub{}}{e_j\subst{x}{D_j\ \overline{e}}\subst{y_i}{\overline{e}}}$. 
Also, let 
  $\tologicshort{\Gamma}{\refa}{}{\pred}{}{}{}$,	
  $\tologicshort{\Gamma}{\refa_i\subst{x}{D_j\ \overline{\refa}}\subst{y_i}{\overline{\refa_i}}}{}{\pred_i}{}{}{}$.
By the axiomatic behavior of the measure selector $\checkdc{\dc_j\ \overline{\pred}}$, we get 
  $\bmodel\models \checkdc{\dc_j\ \overline{\pred}}$.
Thus, 
  $\bmodel
	 {\eif{\checkdc{\dc_1}\ \pred}{\pred_1}{\ldots} \ \mathtt{else}\ \pred_n}
	= \pred_j
  $.

\item 
Assume
	$\track{(x, e_x)\tsub{}}{x}
	\hookrightarrow
	\track{(x, e_x')\tsub{2}}{x}$.
By inversion
	$\track{\tsub{}}{e_x}
	\hookrightarrow
	\track{\tsub{2}}{e_x'}$.
By identity of equality, 
  $(x, \pred_x)\bmodel \cup (\bmodel_2 \setminus \bmodel)
    \models x = x$. 

\item 
Assume
	$\track{(y, e_y)\tsub{}}{x}
	\hookrightarrow
	\track{(y, e_y)\tsub{2}}{e_x}$.
By inversion
	$\track{\tsub{}}{x}
	\hookrightarrow
	\track{\tsub{2}}{e_x}$.
Assume 
  $\tologicshort{\Gamma}{\refa_x}{}{\pred_x}{}{}{}$.
By IH
  $\bmodel \cup (\bmodel_2 \setminus \bmodel)
    \models x = \pred_x$. 
Thus  
  $(y,\pred_y)\bmodel \cup (\bmodel_2 \setminus \bmodel)
    \models x = \pred_x$. 

\item 
Assume
	$\track{(x, w)\tsub{}}{x}
	\hookrightarrow
	\track{(x, w)\tsub{}}{w}$.
Thus  
  $(x,\interp{w})\bmodel 
    \models x = \interp{w}$. 

\item 
Assume
	$\track{(x, D\ \overline{y})\tsub{}}{x}
	\hookrightarrow
	\track{(x, D\ \overline{y})\tsub{}}{D\ \overline{y}}$.
Thus  
  $(x,\smapp{D}{\overline{y}})\bmodel 
    \models x = \smapp{D}{\overline{y}}$. 

\item 
Assume
	$\track{(x, D\ \overline{e})\tsub{}}{x}
	\hookrightarrow
	\track{(x, D\ \overline{y}), \overline{(y_i, e_i)}\tsub{}}{D\ \overline{y}}$.
Assume 
  $\tologicshort{\Gamma}{\refa_i}{}{\pred_i}{}{}{}$.
Thus  
  $(x,\smapp{D}{\overline{y}}), \overline{(y_i, \pred_i)}\bmodel 
    \models x = \smapp{D}{\overline{y}}$. 
\end{itemize}
\end{proof}

\subsection{Soundness of Approximation}
\begin{theorem}[Soundness of Algorithmic]
If \ahastype{\env}{e}{\typ} then \hastype{\env}{e}{\typ}.
\end{theorem}
\begin{proof}
To prove soundness it suffices to prove that subtyping is appropriately approximated, 
as stated by the following lemma.
\end{proof}

\begin{lemma}
If \aissubtype{\env}{\tref{v}{\btyp}{e_1}}{\tref{v}{\btyp}{e_2}}
then \issubtype{\env}{\tref{v}{\btyp}{e_1}}{\tref{v}{\btyp}{e_2}}.
\end{lemma}
\begin{proof}
By rule \rsubbase, we need to show that
$\forall \sub\in\interp{\env}.
  \interp{\applysub{\sub}{\tref{v}{\btyp}{\refa_1}}}
  \subseteq
  \interp{\applysub{\sub}{\tref{v}{\btyp}{\refa_2}}}$.
We fix a $\sub\in\interp{\env}$.
and get that forall bindings
$(\tbind{x_i}{\tref{v}{\btyp^{\downarrow}}{\refa_i}}) \in \env$,
$\evalsto{\applysub{\sub}{e_i\subst{v}{x_i}}}{\etrue}$.

Then need to show that for each $e$,
if $e \in \interp{\applysub{\sub}{\tref{v}{\btyp}{\refa_1}}}$,
then $e \in \interp{\applysub{\sub}{\tref{v}{\btyp}{\refa_2}}}$.

If $e$ diverges then the statement trivially holds.
Assume $\evalsto{e}{w}$.
We need to show that
if $\evalsto{\applysub{\sub}{e_1\subst{v}{w}}}{\etrue}$
then $\evalsto{\applysub{\sub}{e_2\subst{v}{w}}}{\etrue}$.

Let \vsub the lifted substitution that satisfies the above.
Then  by Lemma~\ref{corollary:embedding}
for each model $\bmodel \in \interp{\vsub}$,
$\bmodel\models\pred_i$, and $\bmodel\models q_1$
for
$\tologicshort{\env}{e_i\subst{v}{x_i}}{\btyp}{\pred_i}{\embed{\btyp}}{\smtenv_i}{\axioms_i}$
$\tologicshort{\env}{e_i\subst{v}{w}}{\btyp}{q_i}{\embed{\btyp}}{\smtenv_i}{\beta_i}$.
Since \aissubtype{\env}{\tref{v}{\btyp}{e_1}}{\tref{v}{\btyp}{e_2}} we get
$$
\bigwedge_i \pred_i
\Rightarrow q_1 \Rightarrow q_2
$$
thus $\bmodel\models q_2$.
By Theorem~\ref{thm:appendix:embedding} we get $\evalsto{\applysub{\sub}{\refa_2\subst{v}{w}}}{\etrue}$.
\end{proof}

\section{Reasoning About Lambdas}\label{sec:lambdas}

Encoding of $\lambda$-abstractions and applications via
uninterpreted functions, while sound, is \emph{imprecise}
as it makes it hard to prove theorems that require $\alpha$-
and $\beta$-equivalence or extensional equality.
Using the universally quantified $\alpha$- and
$\beta$- equivalence axioms would let the type checker
accept more programs, but would render validity, and
hence, type checking undecidable.
Next, we identify a middle ground by describing an
not provably complete, but sound and decidable approach to
increase the precision of type checking by
strengthening the VCs with instances of
the $\alpha$- and $\beta$- equivalence axioms
~\S~\ref{subsec:equivalences} and by introducing
a combinator for safely asserting extensional
equality~\S~\ref{subsec:extensionality}.
In the sequel, we omit $\smtappname{}{}$
when it is clear from the context.

\subsection{Equivalence}\label{subsec:equivalences}
As soundness relies on satisfiability under a
\bmodel  (see Definition~\ref{def:beta-model}),
we can safely \emph{instantiate} the axioms of
$\alpha$- and $\beta$-equivalence on any set of
terms of our choosing and still preserve soundness
(Theorem~\ref{thm:soundness-smt}).
That is, instead of checking the validity
of a VC
${p \Rightarrow q}$,
we check the validity of a \emph{strengthened VC},
${a \Rightarrow p \Rightarrow q}$,
where $a$ is a (finite) conjunction
of \emph{equivalence instances}
derived from $p$ and $q$,
as discussed below.

\mypara{Representation Invariant}
The lambda binders,
for each SMT sort, are drawn from a
pool of names $x_i$ where the index
$i=1,2,\ldots$.
When representing
$\lambda$ terms we enforce
a \emph{normalization invariant}
that for each lambda term
$\slam{x_i}{e}$, the index $i$
is greater than any lambda
argument appearing in $e$.

\mypara{$\alpha$-instances}
For each syntactic term
${\slam{x_i}{e}}$ and $\lambda$-binder
$x_j$ such that $i < j$ appearing in the VC,
we generate an \emph{$\alpha$-equivalence instance predicate}
(or \emph{$\alpha$-instance}):
$$\slam{x_i}{e} = \slam{x_j}{e \subst{x_i}{x_j}}$$

%

The conjunction of $\alpha$-instances
can be more precise than De Bruijn
representation, as they let the SMT
solver deduce more equalities via
congruence.
For example, this VC is needed
to prove the applicative laws for "Reader":
$$
d = \slam{x_1}{(\sapp{x}{x_1})}
  \quad
  \Rightarrow
  \quad
  \slam{x_2}{(\sapp{(\slam{x_1}{(\sapp{x}{x_1})})}{x_2})}
  = \slam{x_1}{(\sapp{d}{x_1})}
$$
The $\alpha$ instance
${\slam{x_1}{(\sapp{d}{x_1})} = \slam{x_2}{(\sapp{d}{x_2})}}$
derived from the VC's hypothesis,
combined with congruence immediately
yields the VC's consequence.

\mypara{$\beta$-instances}
For each syntactic term $\smapp{(\slam{x}{e})}{e_x}$,
with $e_x$ not containing any $\lambda$-abstractions,
appearing in the VC, we generate a \emph{$\beta$-equivalence
instance predicate} (or \emph{$\beta$-instance}):
$$
\smapp{(\slam{x_i}{e})}{e_x} = \SUBST{e}{x_i}{e_x},
  \mbox{ s.t. } e_x \mbox{ is $\lambda$-free}
$$
The $\lambda$-free restriction is a simple way
to enforce that the reduced term ${\SUBST{e}{x_i}{e'}}$
enjoys the representation invariant.
%
%
For example, consider the following VC
needed to prove that the bind operator for
lists satisfies the monadic associativity law.
$$(\sapp{f}{x} \ebind g) = \smapp{(\slam{y}{(\sapp{f}{y} \ebind g)})}{x}$$
The right-hand side of the above VC generates
a $\beta$-instance that corresponds directly
to the equality, allowing the SMT solver to
prove the (strengthened) VC.


\mypara{Normalization}
The combination of $\alpha$- and $\beta$-instances
is often required to discharge proof obligations.
For example, when proving that the bind operator
for the "Reader" monad is associative, we need
to prove the VC:
$$\slam{x_2}{(\slam{x_1}{w})} =
  \slam{x_3}{(\smapp{(\slam{x_2}{(\slam{x_1}{w})})}{w})}$$
The SMT solver proves the VC via the equalities
corresponding to an $\alpha$ and then $\beta$-instance:
$$
\slam{x_2}{(\slam{x_1}{w})}
  \ =_{\alpha}\  \slam{x_3}{(\slam{x_1}{w})} \\
  \ =_{\beta}\  \slam{x_3}{(\smapp{(\slam{x_2}{(\slam{x_1}{w})})}{w})}
$$

\subsection{Extensionality} \label{subsec:extensionality}

Often, we need to prove that two
functions are equal, given the
definitions of reflected binders.
Consider
\begin{code}
    reflect id
    id x = x
\end{code}
\toolname accepts the proof that
"id x = x" for all "x":
\begin{code}
    id_x_eq_x :: x:a -> {id x = x}
    id_x_eq_x = \x -> #id# x =. x ** QED
\end{code}
as ``calling'' "id" unfolds its definition,
completing the proof.
However, consider this $\eta$-expanded variant of
the above proposition:
\begin{code}
    type Id_eq_id = {(\x -> id x) = (\y -> y)}
\end{code}
\toolname \emph{rejects} the proof:
\begin{code}
    fails :: Id_eq_id
    fails =  (\x -> #id# x) =. (\y -> y) ** QED
\end{code}
The invocation of "id" unfolds the
definition, but the resulting equality
refinement "{id x = x}" is \emph{trapped}
under the $\lambda$-abstraction.
That is, the equality is absent from the
typing environment at the \emph{top} level,
where the left-hand side term is compared to "\y -> y".
%
Note that the above equality requires
the definition of "id" and hence is
outside the scope of purely the
$\alpha$- and $\beta$-instances.

\mypara{An Exensionality Operator}
To allow function equality via
extensionality, we provide the
user with a (family of)
\emph{function comparison operator(s)}
that transform an \emph{explanation} "p"
which is a proof that "f x = g x" for every
argument "x", into a proof that "f = g".
\begin{mcode}
    =* :: f:(a -> b) -> g:(a -> b) -> exp:(x:a -> {f x = g x}) -> {f = g}
\end{mcode}
Of course, "=*" cannot be implemented;
its type is \emph{assumed}. We can use
"=*" to prove "Id_eq_id" by providing
a suitable explanation:
\begin{mcode}
    pf_id_id :: Id_eq_id
    pf_id_id = (\y -> y) =* (\x -> id x) $\because$ expl ** QED  where expl = (\x -> #id# x =. x ** QED)
\end{mcode}
The explanation is the second argument to $\because$ which has
the following type that syntactically fires $\beta$-instances:
\begin{code}
    x:a -> {(\x -> id x) x = ((\x -> x) x}
\end{code}

\section{Implementation} 
\label{appendix:impl}
Refinement reflection and \pbesym are implemented in
Liquid Haskell.
The implementation can be found in the
\href{https://git.io/vQ6tv}{Liquid Haskell GitHub repository},
all the benchmarks of
\S~\ref{sec:overview} and~\S~\ref{sec:evaluation}
are included in the
\href{https://git.io/vQ6tf}{nople}
and
\href{https://git.io/vQ6tJ}{ple}
test directories.
The benchmarks for deterministic parallelism can be found at
\href{https://git.io/vQ6Lz}{class-laws} and
\href{https://git.io/vQ6L2}{detpar-laws}.


%
Next, we describe the file
\href{https://git.io/vQ6tU}{ProofCombinators.hs},
the library of proof combinators used by our benchmarks
and discuss known limitations of our implementation.

\newcommand\libname{\ensuremath{\texttt{ProofCombinators}}\xspace}

\subsection{\libname: The Proof Combinators Library}
\label{subsec:library}

In this section we present \libname,
a Haskell library used to structure proof terms.
\libname is inspired by Equational Reasoning Data Types
in Adga~\citep{agdaequational}, providing operators to
construct proofs for equality and linear arithmetic in Haskell.
The constructed proofs are checked by an SMT-solver via Liquid Types.

\mypara{Proof terms} are defined in \libname as a type alias for unit,
a data type that curries no run-time information
\begin{code}
  type Proof = ()
\end{code}
Proof types are refined to express theorems about program functions.
For example, the following "Proof" type expresses that
"fib 2 == 1"
\begin{code}
  fib2 :: () -> {v:Proof | fib 2 == 1}
\end{code}
We simplify the above type by omitting the irrelevant
basic type "Proof" and variable "v"
\begin{code}
  fib2 :: () -> { fib 2 == 1 }
\end{code}

\libname provides primitives to construct proof terms
by casting expressions to proofs.
To resemble mathematical proofs,
we make this casting post-fix.
We write "p *** QED" to cast "p" to a proof term,
by defining two operators "QED" and "***" as
\begin{code}
  data QED = QED

  (***) :: a -> QED -> Proof
  _ *** _  = ()
\end{code}

\mypara{Proof construction.}
To construct proof terms, \libname
provides a proof constructor "op."
for logical operators of the theory of
linear arithmetic and equality:
$\{=, \not =, \leq, <, \geq, > \} \in \odot$.
"op. x y" ensures that $x \odot y$ holds, and returns "x"
\begin{code}
  op.:: x:a -> y:{a| x op y} -> {v:a| v==x}
  op. x _ = x

  -- for example
  ==.:: x:a -> y:{a| x==y} -> {v:a| v==x}
\end{code}
For instance, using "==."
we construct a proof, in terms of Haskell code,
that "fib 2 == 1":
\begin{code}
  fib2 _
    =   fib 2
    ==. fib 1 + fib 0
    ==. 1
    *** QED
\end{code} 

\mypara{Reusing proofs: Proofs as optional arguments.}
Often, proofs require reusing existing proof terms.
For example, to prove "fib 3 == 2" we can reuse the above
"fib2" proof.
We extend the proof combinators, to receive
an \textit{optional} third argument of "Proof" type.
\begin{code}
  op.:: x:a -> y:a -> {x op y} -> {v:a|v==x}
  op. x _ _ = x
\end{code}
"op. x y p" returns "x" while the third argument "p"
explicitly proves $x \odot y$.

\mypara{Optional Arguments.}
The proof term argument is optional.
To implement optional arguments in Haskell we use the standard technique
where for each operator "op!" we define a type class "Optop"
that takes as input two expressions "a" and returns a result "r",
which will be instantiated with either the result value "r:=a"
or a function form a proof to the result "r:=Proof ->  a".
\begin{code}
  class Optop a r where
    (op.) :: a -> a -> r
\end{code}
When no explicit proof argument is required,
the result type is just an "y:a" that curries the proof "x op y"
\begin{code}
  instance Optop a a where
  (op.) :: x:a->y:{a| x op y}->{v:a | v==x }
  (op.) x _ = x
\end{code}
Note that Haskell's type inference~\citep{Sulzmann06}
requires both type class parameters "a" and "r" to be constrainted at class instance
matching time.
In most our examples, the result type parameter "r" is not constrained
at instance matching time, thus
due to the Open World Assumption
the matching instance could not be determined.
To address the above, we used another common Haskell trick,
of generalizing the instance to type arguments "a" and "b" and then
constraint "a" and "b" to be equal "a~b".
This generalization allows the instance to always match and
imposed the equality constraint after matching.
\begin{code}
  instance (a~b)=>Optop a b where
  (op.) :: x:a->y:{x op y}->{v:b | v==x }
  (op.) x _ = x
\end{code}

To explicitly provide a proof argument,
the result type "r" is instantiated to "r:= Proof -> a".
For the same instance matching restrictions as above,
the type is further generalized to return some "b"
that is constraint to be equal to "a".
\begin{code}
  instance (a~b)=>Optop a (Proof->b) where
  (op.) :: x:a->y:a->{x op y}->{v:b | v==x }
  (op.) x _ _ = x
\end{code}
As a concrete example, we define the equality operator "==."
via the type class "OptEq" as
\begin{code}
  class OptEq a r where
   (==.):: a -> a -> r

  instance (a~b)=>OptEq a b where
   (==.)::x:a->y:{a|x==y}->{v:b|v==x}
   (==.) x _ = x

  instance (a~b)=>OptEq a (Proof->b) where
   (==.)::x:a->y:a->{x==y}->{v:b|v==x}
   (==.) x _ _ = x
\end{code}

\mypara{Explanation Operator.}
The ``explanation operator'' "(?)", or "($\because$)",
is used to better structure the proofs.
"(?)" is an infix operator with same fixity as "(op.)"
that allows for the equivalence
" x op. y ? p == (op.) x y p"
\begin{code}
  (?) :: (Proof -> a) -> Proof -> a
  f ? y = f y
\end{code}

\mypara{Putting it all together}
Using the above operators,
we prove that "fib 3 == 2",
reusing the previous proof of "fib 2 == 1",
in a Haskell term that resembles mathematical proofs
\begin{code}
  fib3 :: () ->  {fib 3 == 2}
  fib3 _
    =   fib 3
    ==. fib 2 + fib 1
    ==. 2             ? fib2 ()
    *** QED
\end{code}

\mypara{Unverified Operators}
All operators in \libname, but two are implemented in Haskell
with implementations verified by Liquid Haskell.
The ''unsound`` operators are the assume
(1). "(==?)" that eases proof construction by assuming equalities, to be proven later
and (2). "(=*)" extentional proof equality.

\mypara{Assume Operator} "(==?)" eases proof construction by
assuming equalities while the proof is in process.
It is not implemented in that its body is "undefined".
Thus, if we run proof terms including assume operator, the proof will merely crash
(instead of returning "()").
Proofs including the assume operator are not considered complete,
as via assume operator any statement can be proven,

\mypara{Function Extensional Equality}
Unlike the assume operator that is undefined and
included in unfinished thus unsound proofs,
the functions extensionality is included in valid proofs
that assume function extensionality, an axioms that is assumed,
as it cannot be proven by our logic.

To allow function equality via extensionality,
we provide the user with a function comparison operator
that for each function "f" and "g" it transforms a proof
that for every argument "x", "f x = g x" to a proof
on function equality "f = g".
\begin{code}
(=*) :: Arg a => f:(a -> b) -> g:(a -> b)
     -> p:(x:a -> {f x = g x})
     -> {f = g}
\end{code}
The function "(=*)" is not implemented in the library:
it returns () and its type is assumed.
But soundness of its usage requires the argument type variable "a"
to be constrained by a type class constraint "Arg a",
for both operational and type theoretic reasons.

From \textit{operational} point of view,
an implementation of "(=*)" would require checking
equality of "f x = g x" \textit{forall} arguments "x" of type "a".
This equality would hold due to the proof argument "p".
The only missing point is a way to enumerate all the argument "a",
but this could be provided by a method of the type clas "Arg a".
Yet, we have not implement "(=*)" because we do not know how to
provide such an implementation that can provably satisfy "(=*)"'s type.

From \textit{type theoretic} point of view,
the type variable argument "a"
appears only on negative positions.
Liquid type inference is smart enough to infer that
since "a" appears only negative "(=*)" cannot use any "a"
and thus will not call any of its argument arguments "f", "g", nor the "p".
Thus, at each call site of "(=*)" the type variable `a` is instantiated
with the refinement type "{v:a | false}" indicating dead-code
(since "a"s will not be used by the callee.)
Refining the argument "x:a" with false at each call-site though
leads to unsoundness, as each proof argument "p" is a valid proof under
the false assumption.
What Liquid inference cannot predict is our intention to call
"f", "g" and "p" at \textit{every possible argument}.
This information is capture by the type class constraint "Arg a"
that (as discussed before~\citep{Vazou13}) states that methods of
the type class
"Arg a" may create values of type "a", thus,
due to lack of information on the values that are created by the
methods of "Arg a", "a" can only be refined with "True".

With extensional equality, we can prove
that "\x -> x" is equal to "\x -> id x",
by providing an explicit explanation that
if we call both these functions with the same
argument "x", they return the same result, for each "x".
\begin{mcode}
 safe :: Arg a => a
       -> {(\x -> id x) = (\x -> x)}
 safe _  = (\x -> x)
         =*(\x -> id x) $\because$ (exp ())

 exp :: Arg a => a -> x:a
     -> {(\x -> id x) x = (\x -> x) x}
 exp _ x =  id x
         ==. x
         *** QED
\end{mcode}
Note that the result of "exp"
is an equality of the redexes
"(\x -> id x) x" and "((\x -> x) x".
Extentional function equality requires
as argument an equality on such redexes.
Via $\beta$ equality instantiations,
both such redexes will automatically reduce,
requiring "exp" to prove "id x = x",
with is direct.

Admittedly, proving function equality via extensionality
is requires a cumbersome indirect proof.
For each function equality in the main proof
one needs to define an explanation function
that proves the equality for every argument.

\subsection{Engineering Limitations}
The theory of refinement reflection is fully implemented in Liquid Haskell.
Yet, to make this extension \textit{usable} in \textit{real world applications}
there are four known engineering limitations that need to be addressed.
All these limitations seem straightforward to address
and we plan to fix them soon.

\mypara{The language of refinements}
is typed lambda calculus. That is the types of the lambda arguments are
explicitly specified
instead of being inferred.
As another minor limitation, the refinement language parser
requires the argument to be enclosed in parenthesis
in applications where the function is not a variable.
Thus the Haskell expression
"(\x -> x) e" should be written as "(\x:a -> x) (e)"
in the refinement logic,

\mypara{Class instances methods} can not be reflected.
Instead, the methods we want to use in the theorems/propositions
should be defined as Haskell functions.
This restriction has two major implications.
Firstly, we can not verify correctness of library provided
instances but we need to redifine them ourselves.
Secondly, we cannot really verify class instances with class preconditions.
For example, during verification of monoid associativity of the Maybe instance
\begin{code}
  instance (Monoid a) => Monoid (Maybe a)
\end{code}
there is this "Monoid a" class constraint assumption we needed to raise to
proceed verification.

\mypara{Only user defined data types}
can currently used in verification.
The reason for this limitation is that
reflection of case expressions
requires checker and projector measures for each
data type used in reflected functions.
Thus, not only should these data types be defined in
the verified module, but also should be
be injected in the logic by providing a refined version of
the definition that can (or may not) be trivially refined.

For example, to reflect a function that uses
"Peano" numbers, the Haskell \textit{and} the refined "Peano"
definitions should be provided
\begin{code}
data Peano = Z | S Peano

{-@ data Peano [toInt]
     = Z
     | S {prev :: Peano}
  @-}
\end{code}
Note that the termination function "toInt"
that maps "Peano" numbers to natural numbers
is also crucial for soundness of reflection.

\mypara{There is no module support.}
All reflected definitions,
including, measures (automatically generated checkers and selector,
but also the classic lifted Haskell functions to measures)
and the reflected types of the reflected functions,
are not exposed outside of the module they are defined.
Thus all definitions and propositions should exist in the same module.

\section{Verified Deterministic Parallelism}\label{sec:eval-parallelism}


Finally, we evaluate our deterministic
parallelism prototypes.  Aside from the lines of proof code added,
we evaluate the impact on runtime performance.
Were we using a proof tool external to Haskell, this would not be necessary.
But our proofs are Haskell programs---they are necessarily visible to the
compiler.  In particular, this means a proliferation of unit values and functions
returning unit values.  Also, typeclass instances are witnessed at runtime by
``dictionary'' data structures passed between functions.  Layering proof methods
on top of existing classes like "Ord"
could potentially add indirection or change the
code generated, depending on the details of the optimizer.
%
%
In our experiments we find little or no effect on runtime performance.
Benchmarks were run on a single-socket Intel{\textregistered}
Xeon{\textregistered} CPU E5-2699 v3 with 18 physical cores and 64GiB RAM.

\subsection{LVish: Concurrent Sets}
\label{sec:set}

First, we use the "verifiedInsert" operation
to observe the runtime slowdown imposed by the extra proof methods
of "VerifiedOrd".
We benchmark concurrent sets storing 64-bit integers.
Figure~\ref{fig:set} compares the parallel speedups
for a fixed number of parallel "insert" operations
against parallel "verifiedInsert" operations, varying the number of concurrent
threads.
There is a slight observable difference between the two
lines because the extra proof methods do exist at runtime.
We repeat the experiment for two set implementations: a concurrent skiplist
(SLSet) and a purely functional set inside an atomic reference (PureSet)
as described in~\citet{kuper2014freeze}.


\begin{figure}
  \begin{center}
    \includegraphics[width=3.7in]{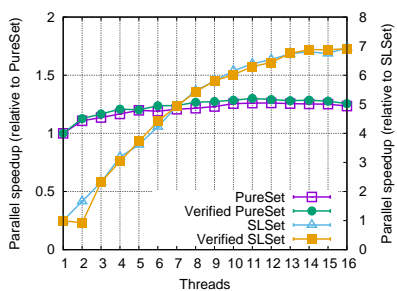}
  \end{center}
  \caption{Parallel speedup for doing 1 million parallel inserts over 10
    iterations, verified and unverified, relative to the unverified version,
    for PureSet and SLSet.}
  \label{fig:set}
\end{figure}

\subsection{\texttt{monad-par}: $n$-body simulation}
\label{sec:nbody}
Next, we verify deterministic behavior of an
$n$-body simulation program that leverages "monad-par", a Haskell library which
provides deterministic parallelism for pure code \cite{monad-par}.

Each simulated particle is represented by a
type "Body" that stores its position, velocity, and mass.
The function "accel" computes the relative acceleration
between two bodies:
\begin{mcode}
  accel :: Body -> Body -> Accel
\end{mcode}
where "Accel" represents the three-dimensional acceleration
\begin{mcode}
  data Accel = $\texttt{Accel}$ Real Real Real
\end{mcode}
To compute the total acceleration
of a body "b" we
(1) compute the relative acceleration between "b"
and each body of the system ("Vec Body") and
(2) we add each acceleration component.
For efficiency, we use a parallel "mapReduce" for the above
computation that
first \textit{maps} each vector body to get the acceleration relative to "b"
("accel b") and then adds each "Accel" value by pointwise addition.
"mapReduce" is only deterministic if the element is a "VerifiedMonoid".
\begin{mcode}
  mapReduce :: VerifiedMonoid b => (a -> b) -> Vec a -> b
\end{mcode}
To enforce the determinism of an $n$-body simulation, we need to provide a
"VerifiedMonoid" instance for "Accel".
We can prove that ("Real", "+", "0.0")
is a monoid.
By product proof composition, we get a verified monoid instance for
\begin{mcode}
  type Accel' = (Real, (Real, Real))
\end{mcode}
which is isomorphic to "Accel" (\ie "Iso Accel' Accel").

Figure~\ref{fig:nbody} shows the results of running two versions of the $n$-body
simulation with 2,048 bodies over 5 iterations, with and without verification,
using floating point doubles for \texttt{Real}\footnote{Floating point numbers
  notoriously violate associativity, but we use this approximation because
  Haskell does net yet have an implementation of {\em
    superaccumulators}~\cite{superaccumulation}.}. Notably, the two programs have almost
identical runtime performance.  This demonstrates that even when verifying code
that is run in a tight loop (like "accel"), we can expect that our programs
will not be slowed down by an unacceptable amount.

\begin{figure}
  \begin{center}
    \includegraphics[width=3.7in]{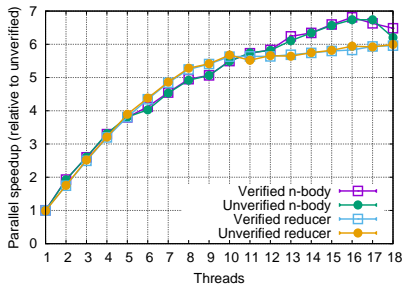}
  \end{center}
  \caption{Parallel speedup for doing a parallel $n$-body simulation and
      parallel array reduction. The speedup is relative to the
    unverified version of each respective class of program.}
  \label{fig:nbody}
\end{figure}

\subsection{DPJ: Parallel Reducers}
\label{sec:reducer}
The Deterministic Parallel Java (DPJ) project provides
a deterministic-by-default semantics for the Java programming language~\cite{DPJ}.
In DPJ, one can declare a method as "commutative" and thus {\em assert} that
racing instances of that method result in a deterministic outcome.
For example:
%
\begin{code}
  commutative void updateSum(int n) writes R { sum += n; }
\end{code}
But, DPJ provides no means to formally prove commutativity and thus
determinism of parallel reduction.
In Liquid Haskell, we specified commutativity as an extra proof method
that extends the "VerifiedMonoid" class.
\begin{mcode}
  class VerifiedMonoid a => VerifiedCommutativeMonoid a where
    commutes :: x:a -> y:a -> { x <> y = y <> x }
\end{mcode}
Provably commutative appends can be used to deterministically
update a reducer variable, since the result is the same regardless
of the order of appends.
We used LVish~\cite{kuper2014freeze} to
encode a reducer variable with a value "a" and a region "s"
as "RVar s a".
\begin{mcode}
  newtype RVar s a
\end{mcode}
We specify that safe (\ie deterministic)
parallel updates require provably commutative appending.
\begin{mcode}
  updateRVar :: VerifiedCommutativeMonoid a => a -> RVar s a -> Par s ()
\end{mcode}
Following the DPJ program, we used "updateRVar"'s provably deterministic interface
to compute, in parallel, the sum of an array with $3$x$10^9$ elements by updating
a single, global reduction variable using a varying number of threads.
Each thread sums segments of an array, sequentially, and updates the variable
with these partial sums.
In Figure \ref{fig:nbody}, we compare the verified and unverified versions
of our implementation to observe no appreciable
difference in performance.

\end{document}